%% file: ow2014.tex
\pdfoutput=1
\documentclass[preprint]{imsart}
\usepackage{amsthm,amsmath}
\usepackage[numbers]{natbib}
\RequirePackage[colorlinks,citecolor=blue,urlcolor=blue]{hyperref}
\usepackage{geometry}                
\usepackage{graphicx}
\usepackage{amssymb}
\usepackage{epstopdf}
\usepackage{xargs}
\usepackage{ifthen}
\usepackage{bbm}
\usepackage{aliascnt}
\usepackage{stmaryrd}
\usepackage{mathtools}
\usepackage{todonotes}
\usepackage{enumerate}
\usepackage{ushort}
\usepackage{algpseudocode}
\usepackage{algorithm}
\usepackage{tikz}
\usetikzlibrary{positioning}

\numberwithin{equation}{section}

\startlocaldefs

\input{ow2014_def}

\endlocaldefs

\begin{document}

\begin{frontmatter}


\title{Efficient particle-based online smoothing in general hidden Markov models: the P{a}RIS algorithm}
\runtitle{The P{a}RIS algorithm}

\begin{aug}
\author{\fnms{Jimmy} \snm{Olsson}\thanksref{t1}
\ead[label=e1]{jimmyol@kth.se}}
\and
\author{\fnms{Johan} \snm{Westerborn}
\ead[label=e2]{johawes@kth.se}}

\thankstext{t1}{This work is supported by the Swedish Research Council, Grant 2011-5577.}
\runauthor{J.~Olsson and J.~Westerborn}

\affiliation{KTH Royal Institute of Technology}

\address{Department of Mathematics \\
KTH Royal Institute of Technology \\
SE-100 44  Stockholm, Sweden \\
\printead{e1}, \printead*{e2}}
\end{aug}

\runauthor{J.~Olsson and J.~Westerborn}

\begin{abstract}
    This paper presents a novel algorithm, the particle-based, rapid incremental smoother (PaRIS), for efficient online approximation of smoothed expectations of additive state functionals in general hidden Markov models. The algorithm, which has a linear computational complexity under weak assumptions and very limited memory requirements, is furnished with a number of convergence results, including a central limit theorem. An interesting feature of PaRIS, which samples on-the-fly from the retrospective dynamics induced by the particle filter, is that it requires two or more backward draws per particle in order to cope with degeneracy of the sampled trajectories and to stay numerically stable in the long run with an asymptotic variance that grows only linearly with time. 
\end{abstract}



\end{frontmatter}



\section{Introduction}
\label{sec:introduction}
\input{ow2014_intro}


\section{Preliminiaries}
\label{sec:preliminaries}
\input{ow2014_prel}


\section{Main results}
\label{sec:main:results}
\input{ow2014_main}


\section{Simulations}
\label{sec:simulations}
\input{ow2014_sim}


\section{Conclusions}
\label{sec:con}
\input{ow2014_con}

\appendix


\section{Proofs}
\label{sec:proofs}
\input{ow2014_proofs}


\section{Technical results}
\label{sec:technical:lemmas}
\input{ow2014_lemmas}

\bibliographystyle{plain}
\bibliography{biblio}

\end{document}

%% file: ow2014_def.tex
\newtheorem{theorem}{Theorem}
\newaliascnt{proposition}{theorem}
\newtheorem{proposition}[proposition]{Proposition}
\aliascntresetthe{proposition}

\newaliascnt{lemma}{theorem}
\newtheorem{lemma}[lemma]{Lemma}
\aliascntresetthe{lemma}

\newaliascnt{corollary}{theorem}
\newtheorem{corollary}[corollary]{Corollary}
\aliascntresetthe{corollary}

\newaliascnt{definition}{theorem}

\aliascntresetthe{definition}

\newaliascnt{example}{theorem}

\aliascntresetthe{example}

\newaliascnt{remark}{theorem}
\newtheorem{remark}[remark]{Remark}
\aliascntresetthe{remark}

\newtheorem{assumption}{Assumption}

\newcounter{hypA}
\newenvironment{hypA}{\refstepcounter{hypA}\begin{itemize}
  \item[{\bf (A\arabic{hypA})}]}{\end{itemize}}

\newcounter{hypB}
\newenvironment{hypB}{\refstepcounter{hypB}\begin{itemize}
  \item[{\bf (B\arabic{hypB})}]}{\end{itemize}}

\def\equationautorefname~#1\null{%
  Equation~(#1)\null
}

\newcommand{\1}[1]{\mathbbm{1}_{#1}}

\newcommandx\A[2][1=]{
\ifthenelse{\equal{#1}{}}
{\hspace{-1mm}(\textbf{A\ref{#2}})\hspace{-1mm}}
{\hspace{-1mm}(\textbf{A\ref{#1}--\ref{#2}})\hspace{-1mm}}
}
\newcommand{\af}[1]{h_{#1}} 
\newcommand{\alg}[1]{\mathcal{#1}} 
\newcommand{\ans}[3]{\mathcal{J}_{#1}^{(#2,#3)}}
\newcommand{\addf}[1]{\termletter_{#1}} 
\newcommand{\adds}[1]{\af{#1}} 
\newcommandx{\arr}[2][1=]{
\ifthenelse{\equal{#1}{}}
{\upsilon_{\N}^{#2}}
{(\upsilon_{\N}^{#2})^{#1}}
}
\newcommandx{\arrterm}[3][1=]{
\ifthenelse{\equal{#1}{}}
{\tilde{\upsilon}_{\N}(#3,#2)}
{\tilde{\upsilon}_{\N}^{#1}(#3,#2)}
}
\newcommandx{\asvar}[4][1=]{
\ifthenelse{\equal{#1}{}}
{\sigma_{#2} \langle #3, #4 \rangle(\af{#1})}
{\sigma_{#2}^2 \langle #3, #4 \rangle(\af{#2})}
}
\newcommandx{\asvarFFBSm}[4][1=]{
\ifthenelse{\equal{#1}{}}
{\tilde{\sigma}_{#2} \langle #3, #4 \rangle(\af{#1})}
{\tilde{\sigma}_{#2}^2 \langle #3, #4 \rangle(\af{#2})}
}
\newcommandx{\asvarstd}[2][1=]{
\ifthenelse{\equal{#1}{}}
{\sigma_{#2}(\af{#2})}
{\sigma_{#2}^2(\af{#2})}
}
\newcommandx{\asvarFFBSmstd}[2][1=]{
\ifthenelse{\equal{#1}{}}
{\tilde{\sigma}_{#2}(\af{#2})}
{\tilde{\sigma}_{#2}^2(\af{#2})}
}

\newcommandx\B[2][1=]{
\ifthenelse{\equal{#1}{}}
{\hspace{-1mm}(\textbf{B\ref{#2}})\hspace{-1mm}}
{\hspace{-1mm}(\textbf{B\ref{#1}--\ref{#2}})\hspace{-1mm}}
}
\newcommandx{\BF}[3][1=]{
\ifthenelse{\equal{#1}{}}
{\kernel{D}_{#2, #3}}
{\kernel{D}_{#2, #3}^{#1}}
}
\newcommandx{\BFcent}[3][1=]{
\ifthenelse{\equal{#1}{}}
{\tilde{\kernel{D}}_{#2, #3}}
{\tilde{\kernel{D}}_{#2, #3}^{#1}}
}
\newcommand{\bi}[3]{J_{#1}^{(#2, #3)}}
\newcommandx{\bk}[2][1=]{ 
\ifthenelse{\equal{#1}{}}
{\overleftarrow{\kernel{Q}}_{#2}}
{\overleftarrow{\kernel{Q}}_{#2}^{#1}}
}

\newcommand{\bmf}[1]{\set{F}(#1)} 
\newcommand{\borel}[1]{\mathcal{B}(\set{#1})} 

\newcommandx{\cexp}[3][1=]{
\ifthenelse{\equal{#1}{}}
{\mathbb{E}\left[ #2 \mid #3 \right]} 
{\mathbb{E}[ #2 \mid #3 ]}
}

\newcommand{\convd}{\overset{\mathcal{D}}{\longrightarrow}}
\newcommand{\convp}{\overset{\prob}{\longrightarrow}}
\newcommand{\covterm}[2]{\eta_{#1}(#2)}

\newcommand{\E}{\mathbb{E}}
\newcommand{\epart}[2]{\xi_{#1}^{#2}}
\newcommand{\eqdef}{\vcentcolon=} 
\newcommand{\eqsp}{}

\newcommandx{\genfd}[1][1=]{
\ifthenelse{\equal{#1}{}}
{\mathcal{F}}
{\mathcal{F}_{\N}}
}

\newcommand{\hbd}{| \tilde{h} |_{\infty}}
\newcommand{\hk}{\kernel{Q}} 
\newcommand{\hklow}{\ushort{\varepsilon}}
\newcommand{\hkup}{\bar{\varepsilon}}
\newcommand{\hd}{q} 

\newcommand{\intvect}[2]{\llbracket #1, #2 \rrbracket}

\renewcommand{\k}{j}
\newcommand{\kernel}[1]{\mathbf{#1}}
\newcommand{\kletter}{\tilde{\N}}
\newcommandx{\K}[1][1=]{
\ifthenelse{\equal{#1}{}}{{\kletter}}{{\tilde{\N}^{#1}}}
}

\newcommand{\leba}{a_\N}
\newcommand{\lebb}{\tilde{a}_\N}
\newcommandx{\lebfun}[1][1=]{
\ifthenelse{\equal{#1}{}}
{\lebfunletter}
{\lebfunletter_{#1}}
}
\newcommand{\lebfunbd}{|\lebfunletter|_{\infty}}
\newcommand{\lebref}{\lambda}
\newcommand{\lebfunletter}{\varphi}
\newcommand{\lebden}{\psi}
\newcommand{\lebker}{\boldsymbol{\Psi}}

\newcommand{\lk}[1]{\kernel{L}_{#1}} 
\newcommand{\Lp}[1]{\set{L}_{#1}}

\newcommandx{\M}[1][1=]{
\ifthenelse{\equal{#1}{}}
{\N}
{\N}
}
\newcommand{\md}[1]{g_{#1}} 
\newcommand{\mdlow}{\ushort{\delta}}
\newcommand{\mdup}{\bar{\delta}}
\newcommand{\meas}[1]{\mathsf{M}(#1)}
\newcommand{\measSpace}[1]{(\set{#1},\alg{#1})} 
\newcommand{\mr}{\varrho}

\newcommand{\N}{N}
\newcommand{\noises}[1]{\sletter_{#1}}
\newcommand{\noiseo}[1]{\oletter_{#1}}
\newcommand{\nset}{\mathbb{N}}
\newcommand{\nsetpos}{\mathbb{N}^*}

\newcommand{\oletter}{\zeta}
\newcommand{\ordo}{\mathcal{O}}
\newcommandx{\oscn}[2][1=]{
\ifthenelse{\equal{#1}{}}{\operatorname{osc}(#2)}{\operatorname{osc}^{#1}(#2)}
}

\newcommand{\partfilt}[1]{\mathcal{F}_{#1}^{\N}}
\newcommand{\partfiltbar}[1]{\tilde{\mathcal{F}}_{#1}^\N }
\newcommand{\PF}{\mathsf{PF}}
\newcommandx\post[2][1=]{
\ifthenelse{\equal{#1}{}}
	{\phi_{#2}}
	{\phi_{#2}^\N}
}
\newcommand{\prob}{\mathbb{P}} 
\newcommand{\probdist}{\mathsf{Pr}}
\newcommand{\probmeas}[1]{\mathsf{M}_1(#1)}
\newcommand{\probSpace}{(\Omega,\alg{F},\prob)} 

\newcommand{\refM}{\mu} 
\newcommand{\rmd}{\mathrm{d}}
\newcommand{\rset}{\mathbb{R}}
\newcommand{\rsetnonneg}{\mathbb{R}_+}
\newcommand{\rsetpos}{\mathbb{R}_+^*}

\newcommand{\set}[1]{\mathsf{#1}} 
\newcommand{\sigmao}{\sigma_{\oletter}}
\newcommand{\sigmas}{\sigma_{\sletter}}
\newcommand{\sletter}{\varepsilon}
\newcommand{\sptime}{{\hat{s}}}

\newcommand{\supn}[1]{\|#1\|_{\infty}} 

\newcommand{\term}[1]{\termletter_{#1}}
\newcommand{\termletter}{\tilde{h}}
\newcommand{\testfsymb}{f}
\newcommandx{\testf}[1][1=]{  
\ifthenelse{\equal{#1}{}}{\testfsymb}{\testfsymb_{#1}}
}
\newcommand{\testfpsymb}{\tilde{f}}
\newcommandx{\testfp}[1][1=]{  
\ifthenelse{\equal{#1}{}}{\testfpsymb}{\testfpsymb_{#1}}
}

\newcommand{\tstatletter}{\kernel{T}}
\newcommandx\tstat[2][1=]{
\ifthenelse{\equal{#1}{}}
	{\tstatletter_{#2}}
	{\tau_{#2}^{#1}}
}
\newcommand{\tstattil}[2]{\tilde{\tau}_{#2}^{#1}}

\newcommand{\unif}{\mathsf{U}}

\newcommand{\wgt}[2]{\omega_{#1}^{#2}}
\newcommand{\wgtsum}[1]{\Omega_{#1}}

\newcommand{\Xinit}{\chi}


%% file: ow2014_intro.tex

This paper deals with the problem of state estimation in \emph{general state-space hidden Markov models} (HMMs) using \emph{sequential Monte Carlo} (SMC) \emph{methods} (also known as \emph{particle filters}), and presents a novel online algorithm for the computation of smoothed expectations of \emph{additive state functionals} in models of this sort. The algorithm, which copes with the well-known problem of \emph{particle ancestral path degeneracy} at a computational complexity that is only \emph{linear} in the number of particles, is provided with a rigorous theoretical analysis establishing its convergence and long-term stability as well as a simulation study illustrating its computational efficiency. 

Given measurable spaces $(\set{X}, \alg{X})$ and $(\set{Y}, \alg{Y})$, an HMM is a bivariate stochastic process $\{ (X_t, Y_t) \}_{t \in \nset}$ (where $t$ will often be referred to as ``time'' without being necessarily a temporal index) taking its values in the product space $(\set{X} \times \set{Y}, \alg{X} \varotimes \alg{Y})$, where the $\set{X}$-valued marginal process $\{ X_t \}_{t \in \nset}$ is a \emph{Markov chain} (often referred to as the \emph{state sequence}) which is only partially observed through the $\set{Y}$-valued \emph{observation process} $\{ Y_t \}_{t \in \nset}$.  Conditionally on the unobserved state sequence $\{ X_t \}_{t \in \nset}$, the observations are assumed to be independent and such that the conditional distribution of each $Y_t$ depends on the corresponding state $X_t$ only. HMMs are nowadays used within a large variety of scientific and engineering disciplines such as econometrics \cite{chib:02}, speech recognition \cite{rabiner:juang:1993}, and computational biology \cite{koski:2001} (the more than 360 references in \cite{cappe:2001:hmmbib} for the period 1989--2000 gives an idea of the applicability of these models). 

Any kind of statistical inference in HMMs involves typically the computation of conditional distributions of  unobserved states given observations. Of particular interest are the sequences of \emph{filter distributions}, i.e., the conditional distributions of $X_t$ given $Y_{0:t} \eqdef (Y_0, \ldots, Y_t)$ (this will be our generic notation for vectors), and \emph{smoothing distributions}, i.e., the joint conditional distributions of $X_{0:t}$ given $Y_{0:t}$, for $t \in \nset$. We will denote these distributions by $\post{t}$ Êand $\post{0:t}$, respectively (precise definitions of these measures are given in \autoref{sec:HMM}). In this paper we are focusing on the problem of computing, recursively in time, smoothed expectations 
\begin{equation} \label{eq:smoothed:exp}
    \post{0:t} \af{t} = \int \af{t}(x_{0:t}) \, \post{0:t}(\rmd x_{0:t}) \quad (t \in \nset) \eqsp,
\end{equation}
for additive functionals $\af{t}$ of form 
\begin{equation} \label{eq:additive:functional}
    \adds{t}(x_{0:t}) \eqdef \sum_{\ell = 0}^{t - 1} \addf{\ell}(x_{\ell: \ell + 1}) \quad (x_{0:t} \in \set{X}^{t + 1}) \eqsp. 
\end{equation}
Expectations of the form \eqref{eq:smoothed:exp} appear naturally in the context of parameter estimation using the \emph{maximum-likelihood method}, e.g., when computing the \emph{score-function} (the gradient of the log-likelihood function) via the \emph{Fisher identity} or when computing the \emph{intermediate quantity} of the \emph{expectation-maximization} (EM) \emph{algorithm}. Of particular relevance is the situation where the HMM belongs to an \emph{exponential family}. We refer to \cite[Sections 10 and 11]{cappe:moulines:ryden:2005} for a comprehensive treatment of these matters. Moreover, \emph{online} implementations of EM (see, e.g., \cite{mongillo:deneve:2008,cappe:2009}) require typically such smoothed expectations to be computed in an online fashion. In the case of \emph{marginal smoothing} the interest lies in computing conditional expectations of some state $X_{\sptime}$ given $Y_{0:t}$ for $t \geq \sptime$, which can be cast into our framework by letting, in \eqref{eq:additive:functional}, $\addf{\ell} = 0$ for $\ell \neq \sptime$ and $\addf{\sptime}(x_{\sptime:\sptime + 1}) = \addf{\sptime}(x_{\sptime})$. Nevertheless, since exact computation of smoothed expectations is possible only in the cases of linear Gaussian HMMs or HMMs with finite state space, we are in general referred to finding approximations of these quantities, and the present paper focuses on the use of SMC-based techniques for this task. A particle filter approximates the flow $\{Ê\post{t} \}_{t \in \nset}$ of filter distributions by a sequence of occupation measures associated with samples $\{ (\epart{t}{i}, \wgt{t}{i}) \}_{i = 1}^\N$, $t \in \nset$, of random draws, \emph{particles} (the $\epart{t}{i}$'s), with associated non-negative importance weights (the $\wgt{t}{i}$'s). Particle filters revolve around two operations: a \emph{selection step} duplicating/discarding particles with large/small importance weights, respectively, and a \emph{mutation step} evolving randomly the selected particles in the state space. The first and most basic implementation, the so-called \emph{bootstrap particle filter} \cite{gordon:salmond:smith:1993} (see also \cite{kitagawa:1996}), propagates, in the mutation step, the particles according to the dynamics of the hidden Markov chain and selects the same multinomially according to importance weights proportional to the local likelihood of each particle given the current observation. This scheme imposes a dynamics of the particle cloud that resembles closely that of the filter distribution flow. Due to its very strong potential to solve nonlinear/non-Gaussian filtering problems, SMC methods have been subject to extensive research during the last two decades, resulting in a broad range of developments and variations of the original scheme; see, e.g., \cite{doucet:defreitas:gordon:2001,cappe:moulines:ryden:2005,cappe:godsill:moulines:2007,doucet:johansen:2009} and the references therein.   

Interestingly, the particle filter provides, as a by-product, approximations also of the joint smoothing distributions in the sense that for each $t \in \nset$,Ê the occupation measure associated with the \emph{ancestral lines} of the particles $\{ \epart{t}{i} \}_{i = 1}^\N$ forms, when the lines are assigned the corresponding weights $\{ \wgt{t}{i} \}_{i = 1}^\N$, an estimate of $\post{0:t}$. Unfortunately, this \emph{Poor man's smoother} (using the terminology of \cite{douc:moulines:stoffer:2014}) has a major flaw in that resampling systematically the particles leads to significant depletion of the trajectories and the existence of a random time before which all the ancestor paths coincide. In fact, \cite{jacob:murray:rubenthaler:2013} established, in the case of a compact state space $\set{X}$, a bound on the expected hight of the ``crown'' of the ancestral tree, i.e., the expected time distance from the last generationÊ back to the most recent common ancestor, which is proportional to $N \log (N)$ and \emph{uniform} in time. Thus, the ratio of the length of the ``crown'' to that of the ``trunk'' tends to zero when time increases, implying that the Monte Carlo approximation obtained through this naive approach will, for long observation records, be based on practically a \emph{single draw}, leading to a depleted estimator with a variance that grows quadratically with time. 

\subsection{Previous work}

In the case of additive state functionals it is possible to cope partly with the degeneracy problem described above by means of a \emph{fixed-lag smoothing} technique \cite{kitagawa:sato:2001,olsson:cappe:douc:moulines:2006,olsson:strojby:2011}. This approach avoids the particle path degeneracy by ``localizing'' the smoothing of a certain state around observations that are only significantly statistically dependent of the state in question and discarding remote and weakly influential observations, i.e., subsequent observations located at a time distance from the state exceeding a lag chosen by the user. The method is expected to work well if the mixing properties of the model allow the lag to be smaller than the length of the ``crown'' of the ancestral tree. Still, such truncation introduces a mixing-dependent bias, and designing the size of the lag is thus a non-trivial task. 

A completely different way of approaching the problem goes via the so-called \emph{forward-filtering backward-smoothing decomposition}, which is based on the fact that the latent process still satisfies the Markov property when evolving backward in time and conditionally on the observations. Consequently, each smoothing measure $\post{0:t}$ Êcan be represented as the joint law of this inhomogeneous backward chain with initial distribution given by the corresponding filter $\post{t}$. Since the transition kernels of the backward chain depend on the filter distributions, which may be estimated efficiently by a particle filter, a particle-based approximation of the smoothing distribution can thus be naturally obtained by running, in a prefatory filtering pass, the particle filter up to time $t$ (if $\post{0:t}$ is the distribution of interest) and, in a backward pass, forming particle-based estimates of the backward kernels (and consequently the smoothing distribution) by modifying the particle weights computed in the forward pass. This scheme, which avoids completely the path degeneracy problem at the cost of a rather significant computational complexity, is referred to as the \emph{forward-filtering backward smoothing} (FFBSm) \emph{algorithm}  \cite{doucet:godsill:andrieu:2000,huerzeler:kuensch:1998,kitagawa:1996}. As an alternative, the \emph{forward-filtering backward simulation} (FFBSi) \emph{algorithm} \cite{godsill:doucet:west:2004} generates, in order to reduce the computational overhead of FFBSm, trajectories being approximately distributed according to the smoothing distribution by \emph{simulating} transitions according to the backward dynamics induced by the particle filter approximations produced by the forward pass; as a consequence, FFBSm can be viewed as a Rao-Blackwellized version of FFBSi. These two algorithms correspond directly to the \emph{Rauch-Tung-Striebel smoother} \cite{rauch:tung:striebel:1965} for linear Gaussian HMMs or the \emph{Baum-Welch algorithm} \cite{baum:petrie:soules:weiss:1970} for HMMs with finite state space. FFBSm and FFBSi were analyzed theoretically in \cite{douc:garivier:moulines:olsson:2010} (see also \cite{delmoral:doucet:singh:2010}), which provides exponential concentration inequalities and well as central limit theorems (CLTs) for these algorithms. Since each backward draw of FFBSi requires a normalizing constant with $\N$ terms to be computed, the overall complexity of the algorithm is $\ordo ( \N^2 )$. Under the mild assumption that the transition density of the latent chain is uniformly bounded, this complexity can be reduced to $\ordo ( \N )$ by means of simple accept-reject approach. The latter technique, which was found in \cite{douc:garivier:moulines:olsson:2010}, will play a key role also in the development of the present paper. Since the Markov transition kernels of the backward chain depend on the filter distributions, the FFBSm and FFBSi algorithms require in general \emph{batch mode} processing of the observations. This is the case also for the $\ordo ( \N )$ smoother proposed in \cite{fearnhead:wyncoll:tawn:2010}, which is based on the \emph{two-filter representation} of each marginal smoothing distribution. 

When the objective consists in \emph{online} smoothing of additive state functionals \eqref{eq:additive:functional}, \emph{recursive} approximation of the forward-filtering backward-smoothing decomposition can be achieved by introducing the auxiliary statistics 
$$
    \tstat{t} \af{t}(x_t) = \E [ \af{t}( X_{0:t} ) \mid X_t = x_t, Y_{0:t} ] \quad ( t \in \nset, x_t \in \set{X} ) \eqsp, 
$$
where $\E$ denotes expectation associated with the law of the canonical version of the HMM (more precisely, in the previous expression $\tstat{t}$ Êis a normalized transition kernel which will be defined in \autoref{sec:HMM}). This auxiliary statistic can be updated online according to  
\begin{equation} \label{eq:T:stat:recursion}
    \tstat{t + 1} \af{t + 1}( x_{t + 1} ) = \E [ \tstat{t} \af{t} ( X_t ) + \addf{t} ( X_t, x_{t + 1} ) \mid X_{t + 1} = x_{t + 1}, Y_{0:t} ] \quad ( t \in \nset , x_{t + 1} \in \set{X} ) \eqsp;
\end{equation}
see \cite{mongillo:deneve:2008,cappe:2009,delmoral:doucet:singh:2010}. In this recursive formula, the expectation is taken under the backward kernel describing the conditional distribution of $X_t$ given $X_{t + 1}$ and $Y_{0:t}$. On the basis of the auxiliary statistics, each smoothed additive functional may be computed as 
$$
    \post{0:t} \af{t} = \int \tstat{t} \af{t} (x_t) \, \post{t}(\rmd x_t)  \quad (t \in \nset) \eqsp.   
$$
Following \cite{delmoral:doucet:singh:2010}, a particle representation of the recursion \eqref{eq:T:stat:recursion} is naturally formed using the estimates of the retrospective dynamics provided by the FFBSm algorithm. Interestingly, this yields a procedure that estimates, as new observations become available, the smoothing distribution flow in a forward-only manner while avoiding completely any problems of particle path degeneracy. However, since the method requires the normalizations of the backward kernels to be computed for each  forward particle, the overall complexity of this algorithm is again $\ordo(\N^2)$, which is unrealistic for large particle sample sizes.    

\subsection{Our approach}

Our novel algorithm, which we will refer to as the \emph{particle-based, rapid incremental smoother} (PaRIS), is, similarly to the forward-only implementation of FFBSm proposed in \cite{delmoral:doucet:singh:2010}, based on \eqref{eq:T:stat:recursion} and can be viewed as an adaptation of the FFBSi algorithm to this recursion. It also shares some similarities with the ancestor sampling approach within the framework of particle Gibbs sampling \cite{lindsten:schon:2013}. Appealingly, we are able to adopt the accept-technique proposed by \cite{douc:garivier:moulines:olsson:2010}, yielding a fast algorithm with $\ordo(\N)$ complexity. PaRIS differs from the forward-only implementation of FFBSm in the way the update \eqref{eq:T:stat:recursion} of the auxiliary function is implemented; more specifically, instead of computing each subsequent auxiliary statistic as the expected sum of the previous statistic and the incremental term under the retrospective dynamics induced by the particle filter, PaRIS \emph{simulates} $\K$ such sums using the backward kernel and updates each statistic by taking the \emph{sample mean} of these draws. Thus, as for the FFBSi algorithm, forward-only FFBSm can be viewed as a Rao-Blackwellization of PaRIS. Interestingly, the design of the sample size $\K$ is ultimately critical, as the naive choice $\K = 1$ leads to a degeneracy phenomenon that resembles closely that of the Poor man's smoother and, consequently, a variance that grows quadratically with $t$; on the other hand, for all $\K \geq 2$ Êthe algorithm stays numerically stable in the long run with a \emph{linearly} increasing variance. The main objective of the present paper is to investigate theoretically this phase transition by, first, deriving, via a non-asymptotic Hoeffding-type inequality, the asymptotic (as $\N$ Êtends to infinity) variance of the Monte Carlo estimates produced by the algorithm (which is highly nontrivial due to the complex dependence structures induced by the backward simulation) and, second, verifying that this asymptotic variance is, for any $\K \geq 2$, of order $\ordo(t)$ and $\ordo(1)$ in the cases of joint smoothing and marginal smoothing, respectively.ÊThe authors are not aware of any similar analysis in the SMC literature.  The stability results are obtained under strong mixing assumptions that are standard in the literature of SMC analysis (see, e.g., \cite{delmoral:guionnet:2001,delmoral:2004,cappe:moulines:ryden:2005}). Also the numerical performance of algorithm is investigated in a simulation study, comprising a linear Gaussian state space model (for which any quantity of interest may be computed exactly using the Rauch-Tung-Striebel smoother) and a stochastic volatility model. 

We finally point out that the PaRIS algorithm was outlined by us in the conference note \cite{olsson:westerborn:2014} without any theoretical support; in the present paper we are able to confirm, through a rigorous theoretical analysis, the conjectures made in the note in question concerning the stability properties of the algorithm. 

To sum up, the smoothing algorithm we propose 
\begin{itemize}
    \item is computationally very efficient and easy to implement, 
    \item does not suffer from particle lineage degeneracy,
    \item allows the observed data of the HMM to be processed online with minimal memory requirements, and 
    \item is furnished with rigorous theoretical results describing the convergence and numeric stability of the same. 
\end{itemize}

\subsection{Outline}

After having introduced some kernel notation, HMMs, and the smoothing problem in~\autoref{sec:preliminaries}, we describe carefully, in~\autoref{sec:particle:smoothing}, particle filters, FFBSm (and its forward-only implementation), and FFBSi. \autoref{sec:PaRIS} contains the derivation of our novel algorithm as well as some discussion of the choice of the design parameter $\K$. Our theoretical results are presented in \autoref{sec:theoretical:results}, including a Hoeffding-type inequality (\autoref{thm:hoeffding:affine}) and a CLT (\autoref{thm:CLT:affine}). \autoref{sec:long:term:stability} is devoted to the numerical stability of PaRIS in the case $\K \geq 2$, and \autoref{thm:linear:bound:normalized:PaRIS} and \autoref{thm:stability:marginal:smoothing:PaRIS} provide variance bounds in the cases of joint and marginal smoothing, respectively. In~\autoref{sec:simulations} we test numerically the algorithm and some conclusions are drawn in \autoref{sec:con}. Finally, \autoref{sec:proofs} and \autoref{sec:technical:lemmas} provide all proofs and some technical results, respectively.

%% file: ow2014_prel.tex

\subsection{Notation}
\label{sec:notation}

Before going into the details concerning HMMs and particle filters we introduce some notation. For any measurable space $\measSpace{X}$, where $\alg{X}$ is a countably generated $\sigma$-algebra, we denote by $\bmf{\alg{X}}$ the set of bounded $\alg{X}/\borel{\mathbb{R}}$-measurable functions on $\set{X}$. For any $h \in \bmf{\alg{X}}$, we let $\supn{h} \eqdef \sup_{x \in \set{X}} |Êh(x) |$ and $\oscn{h} \eqdef \sup_{(x, x') \in \set{X}^2} | h(x) - h(x') |$ denote the sup and oscillator norms of $h$, respectively. Let $\meas{\alg{X}}$ be the set of $\sigma$-finite measures on $\measSpace{X}$ and $\probmeas{\alg{X}} \subset \meas{\alg{X}}$ the probability measures. Given $n \in \nset$, we will denote product sets and product $\sigma$-fields by $\set{X}^n \eqdef \set{X} \times \cdots \times \set{X}$ and $\alg{X}^n \eqdef \alg{X} \varotimes \cdots \varotimes \alg{X}$ ($n$ times), respectively.Ê For real numbers and integers we define the sets $\rsetnonneg \eqdef [0, \infty)$, $\rsetpos \eqdef (0, \infty)$, $\nset \eqdef \{0, 1, 2, \ldots \}$, and $\nsetpos \eqdef \{1, 2, 3, \ldots \}$. For any quantities $\{Êa_\ell \}_{\ell = m}^n$Êwe denote vectors as $a_{m:n} \eqdef (a_m, \ldots, a_n)$ and for any $(m, n) \in \nset^2$ such that $m \leq n$ we denote $\intvect{m}{n} \eqdef \{m, m + 1, \ldots, nÊ\}$. The cardinality of a set $\set{S}$ is denoted by $\# \set{S}$. 

An unnormalized transition kernel $\kernel{K}$ from $\measSpace{X}$ to $\measSpace{Y}$ induces two integral operators, one acting on functions and the other on measures. More specifically, let $h \in \bmf{\alg{X} \varotimes \alg{Y}}$ and $\nu \in \meas{\alg{X}}$, and define the measurable function
$$
    \kernel{K} h : \set{X} \ni x \mapsto \int h(x,y) \, \kernel{K}(x, \rmd y),
$$
and the measure
$$
    \nu \kernel{K} : \alg{Y} \ni \set{A} \mapsto \int \kernel{K}(x, \set{A}) \, \nu(\rmd x),
$$
whenever these quantities are well-defined. Moreover, let $\kernel{K}$ be defined as above and let $\kernel{L}$ be another unnormalized transition kernel from $\measSpace{Y}$ to a third measurable space $\measSpace{Z}$; we then define two different \emph{products} of $\kernel{K}$ and $\kernel{L}$, namely
$$
    \kernel{K} \kernel{L} : \set{X} \times \alg{Z} \ni (x, \set{A}) \mapsto \int \kernel{K}(x, \rmd y) \, \kernel{L}(y, \set{A})
$$
and 
$$
    \kernel{K} \varotimes \kernel{L} : \set{X} \times (\alg{Y} \varotimes \alg{Z}) \ni (x,\set{A}) \mapsto \int \1{\set{A}}(y,z) \, \kernel{K}(x, \rmd y) \, \kernel{L}(y, \rmd z) \eqsp,
$$
whenever these are well-defined. Note that the previous products form new transition kernels from $\measSpace{X}$ to $\measSpace{Z}$ and from $\measSpace{X}$ to $(\set{Y} \times \set{Z}, \alg{Y} \varotimes \alg{Z})$, respectively. We also define the $\varotimes$-product of a kernel $\kernel{K}$ and a measure $\nu \in \meas{\alg{X}}$ as the new measure
$$
    \nu \varotimes \kernel{K} : \alg{X} \varotimes \alg{Y} \ni \set{A} \mapsto \int \1{\set{A}}(x,y) \, \kernel{K}(x, \rmd y) \, \nu(\rmd x) \eqsp.
$$
The concept of \emph{reverse kernels} will be of importance in the coming developments. For a kernel $\kernel{K}$ from $\measSpace{X}$ to $\measSpace{Y}$ and a probability measure $\eta \in \probmeas{\alg{X}}$, the reverse kernel $\overleftarrow{\kernel{K}}_\eta$ associated with $(\eta, \kernel{K})$ is a transition kernel from $\measSpace{Y}$ to $\measSpace{X}$ satisfying, for all $h \in \bmf{\alg{X} \varotimes \alg{Y}}$, 
$$
    (\eta \varotimes \kernel{K}) h = (\eta \kernel{K}) \varotimes \overleftarrow{\kernel{K}}_\eta h \eqsp.
$$
A reverse kernel does not always exist; however if $\kernel{K}$ has a transition density $\kappa$ with respect to some reference measure in $\meas{\alg{Y}}$, then $\overleftarrow{\kernel{K}}_\eta$ exists and is given by 
\begin{equation} \label{eqn:bk}
    \overleftarrow{\kernel{K}}_\eta h(x) \eqdef \frac{ \int h(\tilde{x}) \kappa(\tilde{x}, x) \, \eta(\rmd \tilde{x})}{\int \kappa(\tilde{x}, x) \, \eta(\rmd \tilde{x})} \quad (h \in \bmf{\alg{X}}, x \in \set{X}) 
\end{equation}
(see \cite[Section~2.1]{cappe:moulines:ryden:2005} for details). 

Finally, for any kernel $\kernel{K}$ and any bounded measurable function $h$ we write $\kernel{K}^2 h \eqdef (\kernel{K} h)^2$ and $\kernel{K} h^2 \eqdef \kernel{K}(h^2)$. Similar notation will be used for measures.


\subsection{Hidden Markov models}
\label{sec:HMM}

Let $\measSpace{X}$ and $\measSpace{Y}$ be some measurable spaces, $\hk : \set{X} \times \alg{X} \rightarrow [0, 1]$ and $\kernel{G} : \set{X} \times \alg{Y} \rightarrow [0, 1]$ some Markov transition kernels, and $\Xinit \in \probmeas{\alg{X}}$. We define an HMM as the canonical version of the bivariate Markov chain $\{ (X_t, Y_t) \}_{t \in \nset}$ having transition kernel 
\begin{equation} \label{eq:HMM:kernel}
    \set{X} \times \set{Y} \times (\alg{X} \varotimes \alg{Y}) : ((x, y), \set{A}) \mapsto \hk \varotimes \kernel{G}(x, \set{A}) \eqsp.  
\end{equation}
and initial distribution $\Xinit \varotimes \kernel{G}$. The state process $\{ X_t \}_{t \in \nset}$ is assumed to be only partially observed through the observations process $\{ Y_t \}_{t \in \nset}$. The dynamics \eqref{eq:HMM:kernel} implies that (we refer to \cite[Section~2.2]{cappe:moulines:ryden:2005} for details)
\begin{enumerate}[(i)]
    \item the state sequence $\{ X_t \}_{t \in \nset}$ is a Markov chain with transition kernel $\hk$ and initial distribution $\Xinit$,   
    \item the observations are, conditionally on the states, independent and such that the conditional distribution of each $Y_t$ depends on the corresponding $X_t$ only and is given by the \emph{emission distribution} $\kernel{G}(X_t, \cdot)$.  
\end{enumerate}
We will throughout the paper assume that $\kernel{G}$ admits a density $\md{}$ (referred to as the \emph{emission density}) with respect to some reference measure $\nu \in \meas{\alg{Y}}$, i.e.,
$$
    \kernel{G} h (x) = \int h(y) \md{}(x, y) \, \nu(\rmd y) \quad (x \in \set{X}, h \in \bmf{\alg{Y}}) \eqsp.
$$

In the following we assume that we are given a distinguished sequence $\{Êy \}_{t \in \nset}$ of observations of $\{ Y_t \}_{t \in \nset}$, and will in general omit the dependence on these observations from the notation. Thus, define $\md{t}(x) \eqdef \md{}(x, y_t)$, $x \in \set{X}$. For any $(s, s', t) \in \nset^3$ such that $0 \leq s \leq s' \leq t$ we denote by $\post{s:s' \mid t}$ the conditional distribution (posterior) of $X_{s:s'}$ given the observations $Y_{0:t} = y_{0:t}$.  This distribution may be expressed as 
\begin{equation} \label{eq:def:posterior}
    \post{s:s' \mid t} h = \frac{\idotsint h(x_{s:s'}) \md{0}(x_0) \, \Xinit(\rmd x_0) \prod_{\ell = 0}^{t - 1}  \md{\ell + 1}(x_{\ell + 1}) \, \hk(x_\ell, \rmd x_{\ell + 1})}{\idotsint \md{0}(x_0) \, \Xinit(\rmd x_0) \prod_{\ell = 0}^{t - 1}  \md{\ell + 1}(x_{\ell + 1}) \, \hk(x_\ell, \rmd x_{\ell + 1})} \quad (h \in \bmf{\alg{X}^{s' - s + 1}})
\end{equation}
(assuming that the denominator is non-zero). If $s = s' = t$, we let $\post{t}$ be shorthand for $\post{t \mid t}$, i.e., the filter distribution at time $t$. If $s = 0$ and $s' = t$, then $\post{0:t \mid t}$ is the joint smoothing distribution. For $t \in \nset$, define the unnormalized transition kernels 
$$
    \lk{t} h(x) \eqdef \hk(\md{t + 1} h)(x) \quad (x \in \set{X}, h \in \bmf{\alg{X}}) \eqsp,
$$
with the convention that $\lk{s} \lk{t} \equiv \operatorname{id}$ whenever $s > t$. In addition, we let $\lk{-1}$ be the Boltzmann multiplicative operator associated with $\md{0}$, i.e., $\lk{-1} h(x) \equiv \md{0}(x) h(x)$ for all $h \in \bmf{\alg{X}}$ and $x \in \set{X}$. By combining this notation with \eqref{eq:def:posterior} we may express each filter distribution as  
$$
    \post{t} = \frac{\Xinit \lk{-1} \cdots \lk{t - 1}}{\Xinit \lk{-1} \cdots \lk{t - 1} \1{\set{X}}} \quad (t \in \nset) \eqsp,
$$
which implies immediately the \emph{filter recursion} 
\begin{equation} \label{eq:filter:recursion}
    \post{t + 1} = \frac{\post{t} \lk{t}}{\post{t} \lk{t} \1{\set{X}}} \quad (t \in \nset) \eqsp.
\end{equation}

In the following we will often deal with sums and products of functions with possibly different arguments. Since these functions will be defined on products of $\set{X}$, we will, when needed, with a slight abuse of notation, let subscripts define the domain and the values of such sums and products. For instance, $\testf[t] \testfp[t] : \set{X} \ni x_t \mapsto \testf[t](x_t)\testfp[t](x_t) $ while $\testf[t] + \testfp[t+1] : \set{X}^2 \ni (x_t, x_{t+1}) \mapsto \testf[t](x_t) + \testfp[t+1](x_{t+1})$. 

We will for simplicity assume that the HMM is \emph{fully dominated}, i.e., that also $\hk$ admits a transition density $\hd$ with respect to some reference measure $\refM \in \meas{\alg{X}}$. In this case the reverse kernel $\bk{\eta}$Ê of $\kernel{Q}$ with respect to any $\eta \in \probmeas{\alg{X}}$ is well-defined and specified by \eqref{eqn:bk} (with $\kappa = \hd$). It may be shown (see, e.g., \cite[Proposition~3.3.6]{cappe:moulines:ryden:2005}) that the state process has still the Markov property when evolving conditionally on $Y_{0:t} = y_{0:t}$ in the time-reversed direction; moreover, the distribution of $X_s$ given $X_{s + 1}$ and $Y_{0:t} = y_{0:t}$ is, for any $s \leq t$, given by $\bk{\post{s}}$, which is referred to as the \emph{backward kernel at time} $s$.Ê Consequently, we may express each joint smoothing distribution $\post{0:t \mid t}$ as 
\begin{equation} \label{eqn:back:decomp}
    \post{0:t \mid t} = \post{t} \tstat{t} \eqsp, 
\end{equation}
where we have defined the kernels 
$$
    \tstat{t} \eqdef 
    \begin{cases}
        \bk{\post{t - 1}} \varotimes \bk{\post{t - 2}} \varotimes \cdots \varotimes \bk{\post{0}} & \mbox{ for } t \in \nsetpos \eqsp, \\ 
        \operatorname{id} & \mbox{ for } t = 0 \eqsp. 
    \end{cases}
$$

As discussed in the introduction, the aim of this paper is, given a sequence $\{Ê\term{t} \}_{t \in \nset}$ of terms, to estimate the sequence $\{ \post{0:t \mid t} \af{t} \}_{t \in \nset}$, where each $\af{t}$ is given by \eqref{eq:additive:functional}. By convention, $\af{0} \equiv 0$ (implying, e.g., that $\tstat{0} \af{0} = 0$). Using \eqref{eqn:back:decomp}, each quantity of interest may be expressed as $\post{t} \tstat{t} \af{t} = \post{0:t \mid t} \af{t}$. In addition, note that $\{Ê\tstat{t} \af{t} \}_{t \in \nset}$ may be expressed recursively as 
\begin{equation} \label{eqn:tstat:update}
    \tstat{t + 1} \af{t + 1} = \bk{\post{t}}(\tstat{t} \af{t} + \addf{t}) \eqsp, 
\end{equation}
a formula that will play a key role in the coming developments. 

Finally, define, for $(s, t) \in \nset^2$ such that $s \leq t$, the \emph{retro-prospective kernels}
\[
    \begin{split}
        \BF{s}{t} h(x_s) &\eqdef \iint h(x_{0:t}) \, \tstat{s}(x_s, \rmd x_{0:s - 1}) \, \lk{s} \cdots \lk{t - 1}(x_s, \rmd x_{s + 1:t}) \\
        \BFcent{s}{t} h(x_s) &\eqdef \BF{s}{t}(h - \post{0:t \mid t} h)(x_s) 
    \end{split}  
    \quad (x_s \in \set{X}, h \in \bmf{\alg{X}^{t + 1}}) 
\]
operating simultaneously in the backward and forward directions. Note that the only difference between $\BF{s}{t}$ and $\BFcent{s}{t}$ is that the latter is centralized around the joint smoothing distribution. 


\subsection{Particle-based smoothing in HMMs}
\label{sec:particle:smoothing}

\subsubsection{The bootstrap particle filter}

In the following we assume that all random variables are defined on a common probability space $\probSpace$. The bootstrap particle filter updates sequentially in time a set of particles and associated weights in order to approximate the filter distribution flow $\{Ê\post{t} \}_{t \in \nset}$ given the sequence $\{Êy_t \}_{t \in \nset}$ of observations. Assume that we have at hand a particle sample $\{ (\wgt{t}{i}, \epart{t}{i}) \}_{i = 1}^{\N}$ approximating the filter distribution $\post{t}$ in the sense that for all $h \in \bmf{\alg{X}}$,
\begin{equation} \label{eq:initial:part:appr}
    \post[part]{t} h = \sum_{i = 1}^\N \frac{\wgt{t}{i}}{\wgtsum{t}} h(\epart{t}{i}) \stackrel{\tiny \N \rightarrow \infty}{\backsimeq} \post{t} h \eqsp, 
\end{equation}
where $\wgtsum{s} \eqdef \sum_{i = 1}^\N \wgt{t}{i}$ denotes the weight sum. To form a weighted particle sample $\{(\wgt{t + 1}{i}, \epart{t + 1}{i})\}_{i = 1}^\N$ targeting the subsequent filter $\post{t + 1}$ we simply the plug the approximation $\post[part]{t}$Êinto the filter recursion \eqref{eq:filter:recursion}, yielding the approximation of $\post{t + 1}$ by a mixture distribution proportional to $\sum_{i = 1}^\N \wgt{t}{i} \lk{t}(\epart{t}{i}, \cdot)$, and aim at updating the particle cloud by sampling from this mixture. However, since $\lk{t}$ is generally intractable, we augment the space by the index $i$Ê and apply importance sampling from the extended distribution proportional to $\wgt{t}{i} \lk{t}(\epart{t}{i}, \cdot)$ using the distribution proportional to $\wgt{t}{i} \hk(\epart{t}{i}, \cdot)$ as instrumental distribution. This yields a sampling schedule comprising two operations: selection and mutation. In the selection step, a set $\{ I_{t + 1}^i \}_{i = 1}^\N$ of indices are drawn multinomially according to probabilities proportional to $\{ \wgt{t}{i} \}_{i = 1}^\N$. After this, the mutation step propagates the particles forward according to the dynamics of the state process and assigns the mutated particles importance weights given by the emission density, i.e., for all $i \in \intvect{1}{\N}$, Ê
\[
    \begin{split}
        \epart{t + 1}{i} &\sim \hk(\epart{t}{I_{t + 1}^i}, \cdot) \eqsp, \\
        \wgt{t + 1}{i} &= \md{t + 1}(\epart{t + 1}{i}) \eqsp.
    \end{split}
\]
The algorithm, which is the standard bootstrap particle filter presented in~\cite{gordon:salmond:smith:1993}, is initialized by drawing $\{Ê\epart{0}{i} \}_{i = 1}^\NÊ\sim \Xinit^{\varotimes \N}$ and letting $\wgt{0}{i} = \md{0}(\epart{0}{i})$ Êfor all $i \in \intvect{1}{\N}$. In this basic scheme, which is summarized in \autoref{alg:PF}, the information provided by the most current observation $y_{t + 1}$ Êenters the algorithm via the importance weights only. However, instead of moving the particles ``blindly'' according to the latent dynamics $\hk$, it is, in order to direct the particle swarm toward regions of the state space with large posterior probability, possible to increase the influence of the last observation on the mutation moves as well as the selection mechanism step via the framework of \emph{auxiliary particle filters} \cite{pitt:shephard:1999}. Even though all the results of the present paper can be extended straightforwardly to auxiliary particle filters, we have chosen to limit the presentation to bootstrap-type particle filters only for clarity. 

\begin{algorithm}[htb]
    \caption{Bootstrap particle filter} \label{alg:PF}
    \begin{algorithmic}[1]
        \Require A weighted particle sample $\{ (\epart{t}{i}, \wgt{t}{i}) \}_{i = 1}^\N$ targeting $\post{t}$.
        \For{$i = 1 \to \N$}
        \State $I_{t + 1}^i \sim \probdist( \{Ê\wgt{t}{\ell} \}_{\ell = 1}^\N )$;
        \State draw $\epart{t + 1}{i} \sim \hk(\epart{t}{I_{t + 1}^i}, \cdot)$;
        \State set $\wgt{t + 1}{i} \gets \md{t + 1}(\epart{t + 1}{i})$;Ê
        \EndFor
        \State \Return $\{ (\epart{t}{i}, \wgt{t}{i}) \}_{i = 1}^\N$.
    \end{algorithmic}
\end{algorithm}

In the following we will express \autoref{alg:PF} in a compact form by writing 
$$
    \mbox{``}\{ (\epart{t + 1}{i}, \wgt{t + 1}{i}) \}_{i = 1}^\N \gets \PF\left( \{ (\epart{t}{i}, \wgt{t}{i}) \}_{i = 1}^\N \right)\mbox{''} \eqsp. 
$$
\subsubsection{Forward-filtering backward-smoothing (FFBSm)}

As discussed in the introduction, the bootstrap filter may also be used for smoothing, as the weighted occupation measures associated with the genealogical trees of the particle samples generated by the algorithm form consistent estimates of the joint smoothing distributions. A way of detouring the particle path degeneracy of this Poor man's smoother goes via the backward decomposition~\eqref{eqn:back:decomp}, granted that we are able to approximate  each kernel $\bk{\post{s}}$, $s \in \nset$. However, considering instead the reverse kernel associated with the particle filter $\post[part]{s}$, yields, via \eqref{eqn:bk}, the particle approximations  
\begin{equation} \label{eq:part:approx:bk}
    \bk{\post[part]{s}}h(x) = \sum_{i = 1}^\N \frac{\wgt{s}{i} \hd(\epart{s}{i}, x)}{\sum_{\ell = 1}^\N \wgt{s}{\ell} \hd(\epart{s}{\ell}, x)} h(\epart{s}{i}) \quad (x \in \set{X}, h \in \bmf{\alg{X}}) \eqsp. 
\end{equation}
FFBSm consists in simply inserting these approximations into~\eqref{eqn:back:decomp}, i.e., approximating, for $h \in \bmf{\alg{X}^{t + 1}}$,Ê $\post{0:t \mid t} h$ by Ê
\begin{equation} \label{eqn:ffbsm}
    \post[part]{0:t \mid t} h \eqdef \sum_{i_0 = 1}^\N \cdots \sum_{i_t = 1}^\N \left( \prod_{s = 0}^{t - 1} \frac{\wgt{s}{i_s} \hd(\epart{s}{i_s}, \epart{s + 1}{i_{s + 1}})}{\sum_{\ell = 1}^\N \wgt{s}{\ell} \hd(\epart{s}{\ell}, \epart{s + 1}{i_{s + 1}})} \right) \frac{\wgt{t}{i_t}}{\wgtsum{t}} h(\epart{0}{i_0}, \ldots, \epart{t}{i_t}) \eqsp. 
\end{equation}
For general objective functions $h$, this occupation measure is impractical as the cardinality of its support grows geometrically fast with time. In the case where the objective function $h$Ê is of additive form \eqref{eq:additive:functional} the computational complexity is still quadratic, since computation of the normalizing constants $\sum_{\ell = 1}^\N \wgt{s}{\ell} \hd(\epart{s}{\ell}, \epart{s + 1}{i})$ is required for all $i \in \intvect{1}{N}$ and $s \in \intvect{0}{t - 1}$. Consequently, FFBSm a computationally intensive approach. 

\subsubsection{Forward-only implementation of FFBSm}

Appealingly, as noted by \cite{delmoral:doucet:singh:2010}, in the case of additive state functionals the sequence $\{Ê\post[part]{t} \af{t} \}_{t \in \nset}$ Êcan be computed \emph{on-the-fly} as $t$ increases on the basis of the recursion \eqref{eqn:tstat:update}. More specifically, plugging the estimates \eqref{eq:part:approx:bk} into the recursion in question yields particle approximations $\{Ê\tstattil{i}{t} \}_{i = 1}^\N$ of the statistics $\{Ê\tstat{t} \af{t}(\epart{t}{i}) \}_{i = 1}^\N$ evaluated at the particle locations. After initializing $\tstattil{i}{0} = 0$ for all $i \in \intvect{1}{\N}$, these approximations may, when new observations Êbecome available, be updated by first evolving the particle filter sample one step and then setting   
\begin{equation} \label{eq:FFBSm:forward:update}
    \tstattil{i}{t + 1} = \sum_{j = 1}^\N \frac{\wgt{t}{j} \hd(\epart{t}{j}, \epart{t + 1}{i})}{\sum_{\ell = 1} \wgt{t}{\ell} \hd(\epart{t}{\ell}, \epart{t + 1}{i})} \{ \tstattil{j}{t} + \term{t}(\epart{t}{j}, \epart{t + 1}{i}) \} \quad (t \in \nset) \eqsp,
\end{equation}
yielding, via \eqref{eqn:back:decomp}, the estimate 
$$
    \post[part]{0:t \mid t} \af{t} = \sum_{i = 1}^\N \frac{\wgt{t}{i}}{\wgtsum{t}} \tstattil{i}{t} \eqsp,
$$
of $\post{0:t \mid t} \af{t}$. Besides allowing for online processing of the data, the algorithm has also the appealing property that only the current statistics $\{Ê\tstattil{i}{t} \}_{i = 1}^\N$ and particle sample $\{ (\epart{t}{i}, \wgt{t}{i}) \}_{i = 1}^\N$ need to be stored in the memory. Still, the complexity of the scheme is $\ordo(\N^2)$ due to the computation of the normalizing constants of the backward kernel induced by the particle filter. 


\subsubsection{Forward-filtering backward-simulation (FFBSi)}
\label{sec:FFBSi}

In order to remedy the high computational complexity of FFBSm, FFBSi generates trajectories on the \emph{index space} $\intvect{1}{\N}^{t + 1}$ by simulating repeatedly a time-reversed, inhomogeneous Markov chain $\{Ê\tilde{J}_s \}_{s = 0}^t$Ê with transition probabilities 
\begin{equation} \label{eq:def:trans:prob}
    \kernel{\Lambda}_s^\N (i, j) \eqdef \frac{\wgt{s}{j} \hd(\epart{s}{j}, \epart{s + 1}{i})}{\sum_{\ell = 1}^\N \wgt{t}{\ell} \hd(\epart{s}{\ell}, \epart{s + 1}{i})} \quad ((s, i, j) \in \intvect{0}{t - 1} \times \intvect{1}{\N}^2) 
\end{equation}
and initial distribution (i.e., distribution at time $t$) $\probdist( \{Ê\wgt{t}{j} \}_{j = 1}^\N )$. Given $\{Ê\tilde{J}_s \}_{s = 0}^t$, an approximate draw from the joint smoothing distribution is formed by the random vector $(\epart{0}{\tilde{J}_0}, \ldots, \epart{t}{\tilde{J}_t})$. Consequently, the uniformly weighted occupation measure associated with a set of conditionally independent such draws provides a finite-dimensional approximation of the smoothing distribution $\post{0:t \mid t}$; see~\cite{godsill:doucet:west:2004}. In this basic formulation of FFBSi, the backward sampling pass requires the normalizing constants of the particle-based backward kernels to be computed, and hence the algorithm suffers from a quadratic complexity. On the other hand, on the contrary to FFBSm, this complexity is the same for \emph{all} types of objective functions (whereas FFBSm has quadratic complexity only when applied to additive state functionals). However, following \cite{douc:garivier:moulines:olsson:2010} it is, under the assumption that there exists $\hkup \in \rsetpos$ such that $\hd(x, x') \leq \hkup$ for all $(x, x') \in \set{X}^2$ (an assumption that is satisfied for most models of interest), possible to reduce the computational complexity of FFBSi by simulating the approximate backward kernel using the following accept-reject technique. In order to sample from $\probdist( \{ \kernel{\Lambda}_s^\N(i, j) \}_{j = 1}^\N )$ Êfor given $s \in \intvect{0}{t - 1}$ and $i \in \intvect{1}{\N}$, a candidate $J^*$ drawn from the proposal distribution $\probdist( \{ \wgt{s}{j} \}_{j = 1}^{\N} )$ is accepted with probability $\hd(\epart{s}{J^*}, \epart{s + 1}{i}) / \hkup$. The procedure repeated until acceptance; see~\autoref{alg:accept:reject} for an efficient way of implementing this approach. Under the additional assumption that the transition density is bounded also from below (see \autoref{ass:strong:mixing} below) it can be shown (see~\cite[Proposition 2]{douc:garivier:moulines:olsson:2010}) that the computational complexity of this accept-reject-based FFBSi algorithm is indeed \emph{linear} (i.e., $\mathcal{O}(N)$).

%% file: ow2014_main.tex

Requiring separate forward and backward processing of the data, the standard design of FFBSi is not useful in online applications. We hence propose a novel algorithm which can be viewed as a hybrid between the forward-only implementation of the FFBSm algorithm and the FFBSi algorithm. In order to gain computational effort, it the replaces, in the spirit of FFBSi, exact computation of \eqref{eq:FFBSm:forward:update} by a Monte Carlo estimate. The algorithm, which is presented in the next section, is furnished with rigorous theoretical results concerning its convergence and numerical stability in \autoref{sec:theoretical:results}.    

\subsection{The particle-based, rapid incremental smoother (PaRIS)} \label{sec:PaRIS}

Given estimates $\{ \tstat[i]{t} \}_{i = 1}^\N$ of the auxiliary statistics $\{ \tstat{t} \af{t}(\epart{t}{i}) \}_{i = 1}^\N$ and a particle sample $\{(\epart{t}{i}, \wgt{t}{i})\}_{i = 1}^\N$ targeting the filter $\post{t}$, the algorithm updates the estimated auxiliary statistics by, first, propagating the particle cloud one step, yielding $\{ (\epart{t + 1}{i}, \wgt{t + 1}{i}) \}_{i = 1}^\N$, second, drawing, for each $i \in \intvect{1}{\N}$, conditionally independent and identically distributed indices $\{ \bi{t + 1}{i}{\k} \}_{\k = 1}^{\K}$, where $\K \in \nsetpos$ is some given sample size referred to as the \emph{precision parameter},Ê according to 
$$
    \{ \bi{t + 1}{i}{\k} \}_{\k = 1}^{\K} \sim \probdist( \{Ê\kernel{\Lambda}_t^\N(i, \ell) \}_{\ell = 1}^{\K} )^{\varotimes \K} \quad (i \in \intvect{1}{\N}) \eqsp,  
$$
where the transition probabilities $\kernel{\Lambda}_t^\N$ are defined in \eqref{eq:def:trans:prob}, and, third, letting  
$$
    \tstat[i]{t + 1} = \K^{-1} \sum_{\k = 1}^{\K} \left( \tstat[\bi{t + 1}{i}{\k}]{t} + \term{t}(\epart{t}{\bi{t + 1}{i}{\k}}, \epart{t + 1}{i}) \right) \quad (i \in \intvect{1}{\N}) \eqsp.
$$
Using the updated statistics $\{ \tstat[i]{t + 1} \}_{i = 1}^{\K}$, an estimate of $\post{0:t + 1 \mid t + 1} \af{t + 1} = \post{t + 1} \tstat{t + 1} \af{t + 1}$ is obtained as $\sum_{i = 1}^\N \wgt{t + 1}{i} \tstat[i]{t + 1} / \wgtsum{t + 1}$. As for FFBSm, the algorithm is initialized by setting $\tstat[i]{0} = 0$ for $i \in \intvect{1}{\N}$. The resulting smoother, which is summarized in \autoref{alg:PaRIS}, allows for online processing with constant memory requirements, as it requires only the current particle cloud and estimated auxiliary statistics to be stored at each iteration. In addition, applying, in Step~(4), the accept-reject technique described in the previous section yields, for a given $\K$,  an algorithm with \emph{linear complexity}. 

\begin{algorithm}[htb]
    \caption{Particle-based, rapid incremental smoother (PaRIS)}
    \label{alg:PaRIS}
    \begin{algorithmic}[1]
        \Require Particles sample $\{ (\epart{t}{i}, \wgt{t}{i}) \}_{i = 1}^\N$ targeting $\post{t}$ and estimated auxiliary statistics $\{ \tstat[i]{t} \}_{i = 1}^\N$.
        \State run $\{ (\epart{t + 1}{i}, \wgt{t + 1}{i}) \}_{i = 1}^\N \gets \PF( \{ (\epart{t}{i}, \wgt{t}{i}) \}_{i = 1}^\N )$;
        \For{$i = 1 \to \N$}
            \For{$\k = 1 \to \K$}
                \State draw $\bi{t + 1}{i}{j} \sim \probdist( \{ \wgt{t}{j} q(\epart{t}{\ell}, \epart{t + 1}{i} ) \}_{\ell = 1}^\N )$;
            \EndFor
            \State Set $\tstat[i]{t + 1} \gets \K^{-1} \sum_{\k = 1}^{\K} \left( \tstat[\bi{t + 1}{i}{\k}]{t} + \term{t}(\epart{t}{\bi{t + 1}{i}{\k}}, \epart{t + 1}{i}) \right)$;
        \EndFor 
        \State \Return $\{ \tstat[i]{t + 1} \}_{i = 1}^\N$ and $\{ (\epart{t + 1}{i}, \wgt{t + 1}{i}) \}_{i = 1}^\N$.
    \end{algorithmic}
\end{algorithm}

In the PaRIS scheme, the precision parameter $\K$ has to be set by the user. As shown in \autoref{sec:theoretical:results}, the algorithm is asymptotically consistent (as the particle sample size $\N$ tends to infinity)  \emph{for any fixed} $\K \in \nsetpos$ (i.e., the precision parameter does not need to be increased with $\N$ Êin order to guarantee consistency). Increasing the precision parameter increases the accuracy of the algorithm at the cost of additional computational complexity. Importantly, there is a \emph{significant qualitative difference between the cases $\K = 1$ and $\K \geq 2$}, and it turns out that the latter is required to keep PaRIS numerically stable. This will be clear from the theoretical bounds on the asymptotic variance obtained in \autoref{sec:theoretical:results} as well as from the numerical experiments in \autoref{sec:simulations}. 

In order to understand the fundamental difference between the cases $\K = 1$ and $\K \geq 2$ we may use the backward indices to connect the particles of different generations. Hence, let, for all $t \in \nset$ and $i \in \intvect{1}{\N}$, $\ans{t, t}{i}{\varnothing} \eqdef i$ and, for all $s \in \intvect{0}{t - 1}$ and $j_{s:t - 1} \in \intvect{1}{\K}^{t - s}$,
$$
    \ans{s, t}{i}{j_{s:t - 1}} \eqdef \bi{s + 1}{\ans{s + 1, t}{i}{j_{s + 1:t - 1}}}{j_s} 
$$
and let us write ``$\epart{s}{\ell} \dashleftarrow \epart{t}{i}$'' if there exists a sequence $j_{s:t - 1}$ of indices such that $\ell = \ans{s, t}{i}{j_{s:t - 1}}$. Note that the support of the PaRIS estimator at time $t$ is given by $\set{S}_t \eqdef \prod_{s = 0}^t \set{A}_{s, t} \subset \set{X}^{t + 1}$, where $\set{A}_{s, t} \eqdef \cup_{i = 1}^\N \{Ê\epart{s}{\ell} : \epart{s}{\ell} \dashleftarrow \epart{t}{i} \} \subset \{Ê\epart{s}{\ell} \}_{\ell = 1}^\N$ (so that, since $\epart{t}{i} \dashleftarrow \epart{t}{i}$ for all $i \in \intvect{1}{\N}$, $\set{A}_{t, t} = \{Ê\epart{t}{\ell} \}_{\ell = 1}^\N$). When $\K = 1$, the sequence $\{Ê\# \set{A}_{s, t} \}_{t = s}^\infty$ is non-decreasing, and $\# \set{A}_{s, t} = 1 \Rightarrow \# \set{A}_{u, t} = 1$ for all $u \in \intvect{0}{s}$. This implies a degeneracy phenomenon that resembles closely that of the Poor man's smoother. On the contrary, in the case $\K \geq 2$ it may well occur that $\# \set{A}_{s, t}Ê>  \# \set{A}_{s, t + 1}$, also when $\# \set{A}_{s, t + 1} = 1$. The previous is, for $\N = 3$ and $t = 4$, illustrated graphically in~\autoref{fig:IllDeg}, where columns of nodes represent particle clouds at different time steps (with time increasing rightward) and arrows indicate connections through the relation $\dashleftarrow$. Black-colored particles are included in the support $\set{S}_4$ of the final estimator, while gray-colored ones are inactive. As clear from \autoref{fig:IllDeg}(a), setting $\K = 1$ depletes quickly the support of the estimator, leading to a numerically unstable algorithm. \autoref{fig:IllDeg}(b) shows the same configuration as in (a), but with one additional backward sample (i.e., $\K = 2$). In this case, the sequence $\{Ê\# \set{A}_{s, 4} \}_{s = 0}^4$ is no longer non-decreasing, and a high degree of depletion at some time points (such as $s = 2$) has merely local effect of the support of the estimator. In the coming sections, the fact that PaRIS stays numerically stable for any fixed $\K \geq 2$ is established theoretically as well as through simulations. Ê

\input{tikz_backdraws.tex}

\subsection{Theoretical results}
\label{sec:theoretical:results}

The coming convergence analysis is driven by the following assumption. 

\begin{assumption} \label{ass:boundedness:g:q} \ 
    \begin{enumerate}[(i)]
        \item For all $t \in \nset$, $\md{t} \in \bmf{\alg{X}}$ and $\md{t}(x) > 0$, $x \in \set{X}$,
        \item $\hd \in \bmf{\alg{X}^2}$.
    \end{enumerate}
\end{assumption}
\autoref{ass:boundedness:g:q}(i) implies finiteness and positiveness of the particle weights; the boundedness of the transition density $\hd$ implied by \autoref{ass:boundedness:g:q}(ii) allows, besides certain technical arguments (formalized in \autoref{lemma:generalized:lebesgue}) based on the generalized Lebesgue theorem, the accept-reject sampling technique discussed in \autoref{sec:FFBSi} to be used. 

In turns out to be necessary to establish the convergence of PaRIS for a slightly more general \emph{affine modification} of the additive state functional \eqref{eq:additive:functional}Ê under consideration. More specifically, we will verify that for all $t \in \nset$ and $(\testf[t], \testfp[t]) \in \bmf{\alg{X}}^2$,
\begin{equation} \label{eq:generic:consistency:affine}
    \sum_{i = 1}^\N \frac{\wgt{t}{i}}{\wgtsum{t}}\{ \tstat[i]{t} \testf[t](\epart{t}{i}) + \testfp[t](\epart{t}{i}) \}  \stackrel{\tiny \N \rightarrow \infty}{\backsimeq}  \post{t}(\tstat{t} \af{t} \testf[t] + \testfp[t]) \eqsp,
\end{equation}
where $\{Ê\tstat[i]{t} \}_{i = 1}^\N$ and $\{ (\epart{t}{i}, \wgt{t}{i})Ê\}_{i = 1}^\N$Êare the output of \autoref{alg:PaRIS}, 
in the senses of exponential concentration, weak convergence, and $\Lp{p}$ error. The analogous results for the original additive state functional are then obtained as corollaries by simply applying \eqref{eq:generic:consistency:affine}Ê with $\testf[t] \equiv \1{\set{X}}$ and $\testfp[t] \equiv \1{\set{X}^c}$. Our proofs, which are presented in \autoref{sec:proofs}, are based on single-step analyses of the scheme and rely on techniques developed in \cite{douc:garivier:moulines:olsson:2010} and \cite{douc:moulines:2008}. Nevertheless, the analysis of PaRIS is, especially in the case of weak convergence, highly non-trivial due to the complex dependence between the ancestral lineages of the particles induced by the backward sampling approach (on the contrary to standard FFBSi, where the backward trajectories are conditionally independent; see the previous section).  

\subsubsection{Hoeffding-type inequalities} 
\label{sec:hoeffding}

Besides being a result of independent interest, the following exponential concentration inequality for finite sample sizes $\N$ Êplays an instrumental role in the proof of the CLT in the next section. For reasons that will be clear in the proof of \autoref{thm:CLT:affine}, the bound is established for the unnormalized (i) as well as normalized (ii) estimator. 

\begin{theorem} \label{thm:hoeffding:affine}
    Let \autoref{ass:boundedness:g:q} hold. Then for all $t \in \nset$, $(\testf[t], \testfp[t]) \in \bmf{\alg{X}}^2$, and $\K \in \nsetpos$ there exist constants $(c_t, \tilde{c}_t) \in (\rsetpos)^2$ (depending on $\af{t}$, $\K$, $\testf[t]$, and $\testfp[t]$) such that for all $\N \in \nsetpos$ and all $\varepsilon \in \rsetpos$, 
    \begin{itemize}
        \item[(i)] 
        $ 
            \displaystyle \prob \left( \left| \frac{1}{\N} \sum_{i = 1}^\N \wgt{t}{i} \{\tstat[i]{t} \testf[t](\epart{t}{i}) + \testfp[t](\epart{t}{i}) \} - \post{t - 1} \lk{t - 1} (\tstat{t} \af{t} \testf[t] + \testfp[t]) \right| \geq \varepsilon \right) \leq c_t \exp (- \tilde{c}_t \N \varepsilon^2) \eqsp, 
        $
        \item[(ii)]
        $
        \displaystyle \prob \left( \left| \sum_{i = 1}^\N \frac{\wgt{t}{i}}{\wgtsum{t}} \{ \tstat[i]{t} \testf[t](\epart{t}{i}) + \testfp[t](\epart{t}{i}) - \post{t}(\tstat{t} \af{t} \testf[t] + \testfp[t]) \}\right| \geq \varepsilon \right) \leq c_t \exp (- \tilde{c}_t \N \varepsilon^2)
        $
    \end{itemize} 
    (with the convention $\post{-1} \equiv \Xinit$).
\end{theorem}

The following is an immediate consequence of \autoref{thm:hoeffding:affine}. 

\begin{corollary} \label{thm:hoeffding}
        Let \autoref{ass:boundedness:g:q} hold. Then for all $t \in \nset$ and $\K \in \nsetpos$ there exist constants $(c_t, \tilde{c}_t) \in (\rsetpos)^2$ (depending on $\af{t}$ and $\K$) such that for all $\N \in \nsetpos$ and all $\varepsilon \in \rsetpos$, 
        $$
        \displaystyle \prob \left( \left| \sum_{i = 1}^\N \frac{\wgt{t}{i}}{\wgtsum{t}} (\tstat[i]{t} - \post{t} \tstat{t} \af{t}) \right| \geq \varepsilon \right) \leq c_t \exp (- \tilde{c}_t \N \varepsilon^2) \eqsp.
        $$
\end{corollary}

\subsubsection{Central limit theorems and asymptotic $\Lp{p}$ error} 
\label{sec:clt}

\begin{theorem} \label{thm:CLT:affine}
    Let \autoref{ass:boundedness:g:q} hold. Then for all $t \in \nset$, $(\testf[t], \testfp[t]) \in \bmf{\alg{X}}^2$, and $\K \in \nsetpos$, as $\N \rightarrow \infty$,
    $$
        \sqrt{\N} \sum_{i = 1}^{\N} \frac{\wgt{t}{i}}{\wgtsum{t}} \{\tstat[i]{t} \testf[t](\epart{t}{i}) + \testfp[t](\epart{t}{i}) - \post{t}(\tstat{t} \af{t} \testf[t] + \testfp[t]) \} \convd \asvar{t}{\testf[t]}{\testfp[t]} Z \eqsp, 
    $$
    where $Z$ has standard Gaussian distribution and
    \begin{multline} \label{eq:asvar:closed:form}
        \asvar[2]{t}{\testf[t]}{\testfp[t]} \eqdef \asvarFFBSm[2]{t}{\testf[t]}{\testfp[t]} \\
        + \sum_{s = 0}^{t - 1} \sum_{\ell = 0}^s \K^{\ell - (s + 1)} \frac{\post{\ell} \lk{\ell} \{ \bk{\post{\ell}} ( \tstat{\ell} \af{\ell} + \term{\ell} - \tstat{\ell + 1} \af{\ell + 1} )^2 \lk{\ell + 1} \cdots \lk{s} (\md{s + 1} \{\lk{s + 1} \cdots \lk{t - 1} \testf[t] \}^2 ) \}}{(\post{\ell} \lk{\ell} \cdots \lk{s - 1} \1{\set{X}}) (\post{s} \lk{s} \cdots \lk{t - 1} \1{\set{X}})^2} 
    \end{multline} 
    with
    $$
        \asvarFFBSm[2]{t}{\testf[t]}{\testfp[t]} \eqdef \sum_{s = 0}^{t - 1} \frac{\post{s} \lk{s} \{\md{s + 1} \BFcent[2]{s + 1}{t} ( \af{t} \testf[t] + \testfp[t] ) \}}{(\post{s} \lk{s} \cdots \lk{t - 1} \1{\set{X}})^2} 
    $$
    being the asymptotic variance of the FFBSm algorithm (where, by convention, $\lk{m} \lk{n} = \operatorname{id}$ if $m > n$). 
\end{theorem}

\begin{remark}
    Since for all $s \in \intvect{0}{t - 1}$ and $\ell \in \intvect{0}{s}$,Ê $\ell - (s + 1) \leq -1$, it holds, in~\eqref{eq:asvar:closed:form}, that  
    $$
        \lim_{\K \to \infty} \asvar[2]{t}{\testf[t]}{\testfp[t]} = \asvarFFBSm[2]{t}{\testf[t]}{\testfp[t]} \eqsp,
    $$
    i.e., for large $\K$ the asymptotic variance of PaRIS tends to that of the FFBSm algorithm. This is in line with our expectations, as the forward-only version of FFBSm can be viewed as a Rao-Blackwellization of PaRIS. 
\end{remark}

Again, the following is an immediate consequence of \autoref{thm:CLT:affine}.

\begin{corollary} \label{thm:CLT}
    Let \autoref{ass:boundedness:g:q} hold. Then for all $t \in \nset$ and $\K \in \nsetpos$, as $\N \rightarrow \infty$,
    $$
        \sqrt{\N} \sum_{i = 1}^{\N} \frac{\wgt{t}{i}}{\wgtsum{t}}(\tstat[i]{t} - \post{t} \tstat{t} \af{t}) \convd \asvarstd{t} Z \eqsp,
    $$
    where $Z$ has standard Gaussian distribution and 
    \begin{multline} \label{eq:asvar:non-affine}
        \asvarstd[2]{t} \eqdef \asvarFFBSmstd[2]{t} \\
        + \sum_{s = 0}^{t - 1} \sum_{\ell = 0}^s \K^{\ell - (s + 1)} \frac{\post{\ell} \lk{\ell} \{ \bk{\post{\ell}} ( \tstat{\ell} \af{\ell} + \term{\ell} - \tstat{\ell + 1} \af{\ell + 1} )^2 \lk{\ell + 1} \cdots \lk{s} (\md{s + 1} \{\lk{s + 1} \cdots \lk{t - 1} \1{\set{X}} \}^2 \}}{(\post{\ell} \lk{\ell} \cdots \lk{s - 1} \1{\set{X}}) (\post{s} \lk{s} \cdots \lk{t - 1} \1{\set{X}})^2} \eqsp. 
    \end{multline} 
    with
    $$
        \asvarFFBSmstd[2]{t} \eqdef \sum_{s = 0}^{t - 1} \frac{\post{s} \lk{s} ( \md{s + 1} \BFcent[2]{s + 1}{t} \af{t} )}{(\post{s} \lk{s} \cdots \lk{t - 1} \1{\set{X}})^2} 
    $$
    being the asymptotic variance of the FFBSm algorithm. 
\end{corollary}

By following identically the lines of the proof of \cite[Theorem~8]{douc:moulines:olsson:2014}, we may use \autoref{thm:hoeffding} and \autoref{thm:CLT} for deriving also the asymptotic $\Lp{p}$ error of the estimates produced by the algorithm. 
\begin{corollary} \label{thm:Lp:error}
    Let \autoref{ass:boundedness:g:q} hold. Then for all $p \in \rsetpos$, $t \in \nset$, and $\K \in \nsetpos$, 
    $$
        \lim_{\N \rightarrow \infty} \sqrt{\N} \left \| \sum_{i = 1}^{\N} \frac{\wgt{t}{i}}{\wgtsum{t}}(\tstat[i]{t} - \post{t} \tstat{t} \af{t}) \right \|_{\Lp{p}} = \sqrt{2} \asvarstd{t} \left( \frac{\Gamma\{(p + 1) / 2 \}}{\sqrt{2 \pi}} \right)^{1/p} \eqsp,
    $$
    where $\asvarstd[2]{t}$ is given in \eqref{eq:asvar:non-affine}. 
\end{corollary}

\subsubsection{Time uniform asymptotic variance bounds}
\label{sec:long:term:stability}

In the present section we establish the long-term numerical stability of the PaRIS algorithm by bounding the asymptotic variance \eqref{eq:asvar:non-affine} (and hence, by \autoref{thm:Lp:error}, the asymptotic $\Lp{p}$ error) using mixing-based arguments. We will treat separately joint smoothing and marginal smoothing, and derive, for precision parameters $\K \geq 2$, $\ordo(t)$ and $\ordo(1)$ bounds, respectively, on the asymptotic variances in these cases. Since such time dependence is the best possible for SMC error bounds on the path and marginal spaces, these results confirm the conjecture that the algorithm stays numerically stable for precision parameters of this sort. Similar results for the FFBSm and FFBSi algorithms were obtained in \cite{douc:garivier:moulines:olsson:2010,dubarry:lecorff:2013}. The analysis will be carried through under the following \emph{strong mixing} assumption, which is standard in the literature of SMC analysis (see \cite{delmoral:guionnet:2001} and, e.g., \cite{delmoral:2004,cappe:moulines:ryden:2005,crisan:heine:2008,douc:moulines:olsson:2014} for refinements) and points to applications where the state space $\set{X}$ is a compact set.   
\begin{assumption} \label{ass:strong:mixing}
    \ \\
    \begin{enumerate}[(i)]
    \item There exist constants $0 < \hklow < \hkup < \infty$ such that for all $(x, \tilde{x}) \in \set{X}^2$, 
    $$
        \hklow \leq \hd(x, \tilde{x}) \leq \hkup \eqsp.
    $$
    \item There exist constants $0 < \mdlow < \mdup < \infty$ such that for all $t \in \nset$, $\supn{\md{t}} \leq \mdup$ and $\mdlow \leq \lk{t} \1{\set{X}}(x)$, $x \in \set{X}$.
    \end{enumerate}
\end{assumption}

\subsubsection*{Joint smoothing}

The following assumption implies that the additive functional under consideration grows at most linearly with time,Ê which is a minimal requirement for obtaining an $\ordo(t)$ asymptotic variance.

\begin{assumption} \label{ass:term:bounded}
    There exists $\hbd \in \rsetpos$ such that for all $s \in \nset$,  $\oscn{\term{s}} \leq \hbd$. 
\end{assumption}

As an auxiliary result, we provide an $\ordo(t)$ bound on the asymptotic variance of the FFBSm algorithm; see \cite{dubarry:lecorff:2013} for a similar result on the $\Lp{p}$ error for finite particle sample sizes. 

\begin{proposition} \label{prop:bound:normalized:FFBSm}
    Let \autoref{ass:strong:mixing} and \autoref{ass:term:bounded} hold. Then 
    $$
        \limsup_{t \rightarrow \infty} \frac{1}{t} \asvarFFBSmstd[2]{t} \leq \hbd^2 \frac{4 \mdup}{\mdlow (1 - \mr)^4} \eqsp. 
    $$
\end{proposition}

In the light of \eqref{prop:bound:normalized:FFBSm} Êit suffices to bound the second term of \eqref{eq:asvar:non-affine} by a quantity of order $\ordo(t)$. This yields the following result, where, interestingly, the incremental asymptotic variance caused by the  backward simulation is \emph{inversely proportional to the precision parameter} $\K$. This is well in line with the theory of \emph{random weight SMC methods}, in which, in similarity to our algorithm, intractable quantities (the importance weights) are replaced by random and unbiased estimates of the same (see \cite{olsson:strojby:2011,olsson:strojby:2010}). 

\begin{theorem} \label{thm:linear:bound:normalized:PaRIS}
    Let \autoref{ass:strong:mixing} and \autoref{ass:term:bounded} hold. Then for all $\K \geq 2$,
    \begin{equation*}
        \limsup_{t \rightarrow \infty} \frac{1}{t} \asvarstd[2]{t} 
        \leq \hbd^2 \frac{\mdup}{\mdlow (1 - \mr)^4} \left( 4 + \frac{(4 \mdup \mr^2 + 1 - \mr)^2}{(\K - 1)(1 - \mr)} \right) \eqsp, 
    \end{equation*}
    where $\asvarstd[2]{t}$ is defined in \eqref{eq:asvar:non-affine}. 
\end{theorem}

\subsubsection*{Marginal smoothing}

We turn to marginal smoothing, i.e., the situation when all terms of the additive functional are zero but a single one. For such a particular objective function we able to construct a time uniform bound on the \emph{unnormalized} asymptotic variance of the same form as before, with one term representing the FFBSm asymptotic variance (see \cite[Theorem~12]{douc:garivier:moulines:olsson:2010}) and one additional term being inversely proportional to the precision parameter and representing the loss of accuracy introduced by backward sampling.   

\begin{assumption} \label{ass:marginal:smoothing}
    The additive functional has the following form. For some $\sptime \in \nset$,
    $$
        \term{s}(x_{s:s + 1}) = 
        \begin{cases}
            0 & \mbox{for } s \neq \sptime \eqsp, \\
            \term{\sptime}(x_\sptime) & \mbox{for } s = \sptime \eqsp.   
        \end{cases}
    $$
\end{assumption}

\begin{theorem} \label{thm:stability:marginal:smoothing:PaRIS}
    Let \autoref{ass:strong:mixing} and \autoref{ass:marginal:smoothing} hold. Then for all $t \in \nset$ and $\K \geq 2$,
    \begin{equation*}
        \asvarstd[2]{t} 
        \leq \oscn[2]{\term{\sptime}} \frac{\mdup}{(1 - \mr)^3} \left( \mdup \frac{1 + \mr^2}{1 + \mr} + 4 \frac{1}{ \mdlow (\K - 1)\{(1 - \mr^2) \wedge (1/2) \}} \right) \eqsp, 
    \end{equation*}
    where $\asvarstd[2]{t}$ is defined in \eqref{eq:asvar:non-affine}. 
\end{theorem}

\subsubsection{Computational complexity}

We conclude this section with some comments on the complexity of the algorithm. Under \autoref{ass:boundedness:g:q}(ii), we may cast the accept-reject technique proposed in \cite[Algorithm~1]{douc:garivier:moulines:olsson:2010} into the framework of PaRIS. A pseudo-code describing the resulting scheme is provided by \autoref{alg:accept:reject} in \autoref{sec:accept-reject}. For a given $t \in \nset$, we denote by $C_t \langle \N, \K \rangle$Ê the (random) number of elementary operations needed for executing the PaRIS algorithm parameterized by $(\N, \K) \in (\nsetpos)^2$ from time zero to time $t$. Note that $C_t \langle \N, \K \rangle$Ê is strongly data dependent, as the observations $\{Êy_s \}_{s = 0}^t$ effect, via the particle weights, the acceptance probabilities at the different time steps. Still, under the strong mixing assumption above it is possible to bound uniformly this random variable. The following result is an immediate consequence of \cite[Proposition~2]{douc:garivier:moulines:olsson:2010}. 

\begin{theorem} \label{thm:comlexity}
    Let \autoref{ass:strong:mixing} hold. Then there exists a constant $c \in \rsetpos$Êsuch that $\E[ÊC_t \langle \N, \K \rangle ] \leq c t \N \K / (1 - \mr)$ for all $t \in \nset$. 
\end{theorem}

Thus, the expected number of trials grows linearly with time, the number of particles, and the precision parameter, showing the importance of keeping the latter at a minimum. On the other hand, since the variance bound derived in \autoref{prop:bound:normalized:FFBSm} (and \autoref{thm:stability:marginal:smoothing:PaRIS}) is inversely proportional to $\K$, using an excessively large precision parameter will not pay off in terms of variance reduction (as the variance term controlled by the precision parameter will be negligible beside the variance corresponding to FFBSm). \emph{We hence advocate keeping} $\K$ \emph{at a highly moderate value}, and will return to this matter in connection to the numerical illustrations of the next section.

%% file: tikz_backdraws.tex
\begin{figure}[htb]
\begin{minipage}[t]{0.45\linewidth}
\centering
\centerline{\begin {tikzpicture}[-latex ,auto ,node distance = 0.5 cm,
semithick ,
state/.style ={ circle ,
fill = black,black , text=black , minimum width =0.5 cm, radius = 1}]
\node[state, fill = lightgray] (A1){};
\node[state] (A2) [above =of A1]{};
\node[state, fill = lightgray] (A3) [above =of A2]{};
\node[state, fill = lightgray] (B1) [right =of A1]{};
\node[state] (B2) [right =of A2]{};
\node[state, fill = lightgray] (B3) [right =of A3]{};
\node[state] (C1) [right =of B1]{};
\node[state, fill = lightgray] (C2) [right =of B2]{};
\node[state, fill = lightgray] (C3) [right =of B3]{};
\node[state, fill = lightgray] (D1) [right =of C1]{};
\node[state] (D2) [right =of C2]{};
\node[state] (D3) [right =of C3]{};
\node[state] (E1) [right =of D1]{};
\node[state] (E2) [right =of D2]{};
\node[state] (E3) [right =of D3]{};
\path[draw,dashed,color = lightgray] (B1) -- (A1);
\path[draw,dashed,color = black] (B2) -- (A2);
\path[draw,dashed,color = lightgray] (B3) -- (A2);
\path[draw,dashed,color = black] (C1) -- (B2);
\path[draw,dashed,color = lightgray] (C2) -- (B3);
\path[draw,dashed,color = lightgray] (C3) -- (B2);
\path[draw,dashed,color = lightgray] (D1) -- (C2);
\path[draw,dashed,color = black] (D2) -- (C1);
\path[draw,dashed,color = black] (D3) -- (C1);
\path[draw,dashed,color = black] (E1) -- (D2);
\path[draw,dashed,color = black] (E2) -- (D2);
\path[draw,dashed,color = black] (E3) -- (D3);
\end{tikzpicture}}
\centerline{(a) $\K = 1$}
\end{minipage}
\hfill
\begin{minipage}[t]{0.45\linewidth}
\centering
\centerline{\begin {tikzpicture}[-latex ,auto ,node distance = 0.5 cm,
semithick ,
state/.style ={ circle ,
fill = black,black , text=black , minimum width =0.5 cm, radius = 1}]
\node[state] (A1){};
\node[state] (A2) [above =of A1]{};
\node[state] (A3) [above =of A2]{};
\node[state] (B1) [right =of A1]{};
\node[state] (B2) [right =of A2]{};
\node[state] (B3) [right =of A3]{};
\node[state] (C1) [right =of B1]{};
\node[state] (C2) [right =of B2]{};
\node[state, fill = lightgray] (C3) [right =of B3]{};
\node[state, fill = lightgray] (D1) [right =of C1]{};
\node[state] (D2) [right =of C2]{};
\node[state] (D3) [right =of C3]{};
\node[state] (E1) [right =of D1]{};
\node[state] (E2) [right =of D2]{};
\node[state] (E3) [right =of D3]{};
\path[draw,dashed,color = black] (B1) -- (A1);
\path[draw,dashed,color = black] (B2) -- (A2);
\path[draw,dashed,color = black] (B3) -- (A2);
\path[draw,dashed,color = black] (C1) -- (B2);
\path[draw,dashed,color = black] (C2) edge [bend right] (B3);
\path[draw,dashed,color = lightgray] (C3) -- (B2);
\path[draw,dashed,color = lightgray] (D1) -- (C2);
\path[draw,dashed,color = black] (D2) -- (C1);
\path[draw,dashed,color = black] (D3) -- (C1);
\path[draw,dashed,color = black] (E1) -- (D2);
\path[draw,dashed,color = black] (E2) -- (D2);
\path[draw,dashed,color = black] (E3) -- (D3);
\path[draw,dashed,color = black] (B1) -- (A3);
\path[draw,dashed,color = black] (B2) -- (A1);
\path[draw,dashed,color = black] (B3) -- (A3);
\path[draw,dashed,color = black] (C1) -- (B1);
\path[draw,dashed,color = black] (C2) edge [bend left] (B3);
\path[draw,dashed,color = lightgray] (C3) -- (B1);
\path[draw,dashed,color = lightgray] (D1) -- (C1);
\path[draw,dashed,color = black] (D2) -- (C2);
\path[draw,dashed,color = black] (D3) -- (C2);
\path[draw,dashed,color = black] (E1) -- (D3);
\path[draw,dashed,color = black] (E2) -- (D3);
\path[draw,dashed,color = black] (E3) -- (D2);
\end{tikzpicture}}
\centerline{(b) $\K = 2$}
\end{minipage}
\caption{Genealogical traces corresponding to backward simulation in the PaRIS algorithm. Columns of nodes refer to different particle populations (with $\N = 3$) at different time points (with time increasing rightward) and arrows indicate connections through the relation $\dashleftarrow$. Black-colored particles are included in the support $\set{S}_4$ of the final estimator, while gray-colored ones are inactive.}
\label{fig:IllDeg}
\end{figure}
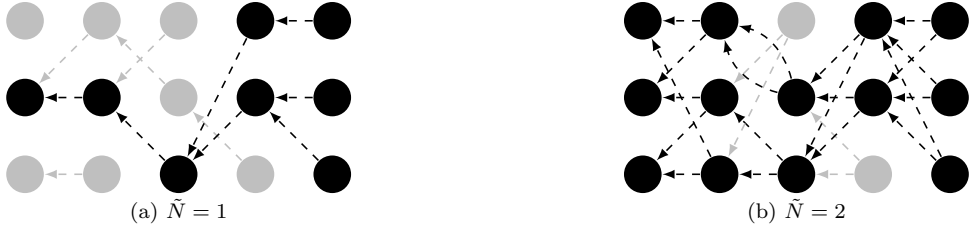

%% file: ow2014_sim.tex

An exhaustive study of the numerical aspects of PaRIS is beyond the scope of the present paper; nevertheless, we benchmark the algorithm on two different models, namely
\begin{itemize}
\item a linear Gaussian state-space model (for which all quantities of interest can be computed exactly for comparison) and 
\item a stochastic volatility model \cite{hull:white:1987}.
\end{itemize} 

\subsection{Linear Gaussian state-space model}

We first consider the linear Gaussian state-space model 
\begin{equation} \label{eq:lin:gaussian:hmm}
    \begin{split}
        X_{t + 1} &= a X_t + \sigmas \noises{t + 1} \\
        Y_t &= b X_t + \sigmao \noiseo{t} 
    \end{split}
    \quad (t \in \nset) \eqsp,
\end{equation}
where $\set{Y} = \set{X} = \rset$ and $\{Ê\noises{t} \}_{t \in \nsetpos}$ and $\{Ê\noiseo{t} \}_{t \in \nset}$  are sequences of mutually independent standard normally distributed random variables. The parameters $(a, b) \in \rset^2$ and $(\sigmas, \sigmao)\in (\rsetpos)^2$ are considered to be known. We aim at computing smoothed expectations of the sufficient statistics 
\begin{equation} \label{eqn:lg:af}
    \af{t}^{(1)}(x_{0:t}) \eqdef \sum_{s = 0}^t x_s, \quad \af{t}^{(2)}(x_{0:t}) \eqdef \sum_{s = 0}^t x_s^2, \quad \af{t}^{(3)}(x_{0:t}) = \sum_{s = 0}^{t - 1} x_s x_{s + 1} \quad (x_{0:t} \in \set{X}^{t + 1}) 
\end{equation}
under the dynamics governed by the parameter vector $(a, b, \sigmas, \sigmao) = (.7,1,.2,1)$, and assume for simplicity that the model is well-specified. For this model, the \emph{disturbance smoother} (see, e.g., \cite[Algorithm 5.2.15]{cappe:moulines:ryden:2005}) provides the exact values of the smoothed sufficient statistics, and we compared these values with approximations obtained using PaRIS as well as the forward-only implementation of FFBSm. With our implementation, parameterizing PaRIS and FFBSm with $(\N, \K) = (150, 2)$ and $\N = 50$, respectively, resulted in very similar computational times for the two algorithms, with PaRIS being slightly faster (recall that FFBSm has a quadratic complexity). As clear from the box plots (based on \emph{time-normalized} estimates) displayed in \autoref{fig:lg:bp2}, PaRIS outperforms clearly FFBSm as the former exhibits lower variance as well as smaller bias for equal computational time. 

\begin{figure}
	\centering
  	\centerline{\includegraphics[width=\textwidth]{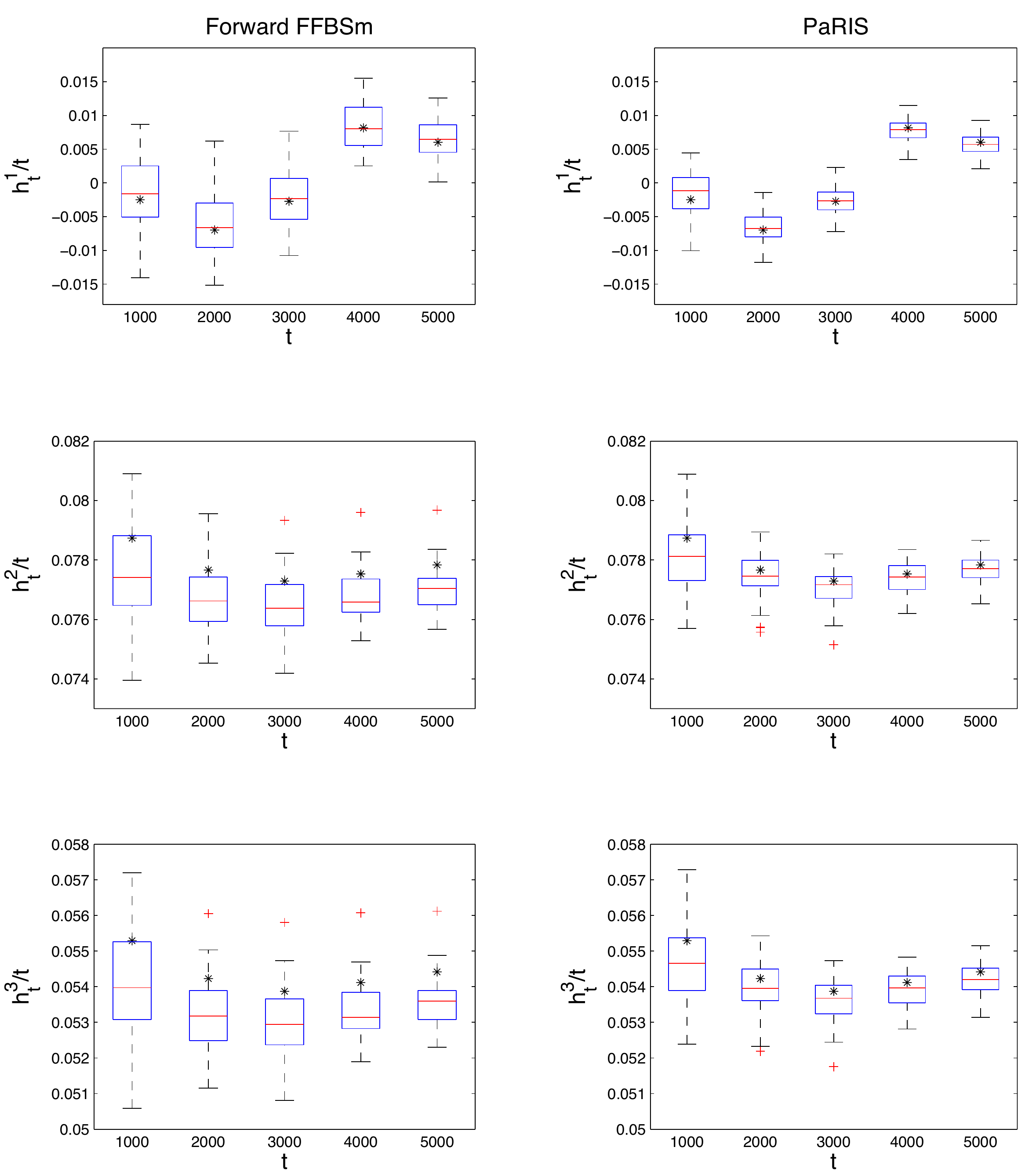}}
 	\caption{Box plots of estimates of smoothed sufficient statistics~\eqref{eqn:lg:af} for the linear Gaussian model \eqref{eq:lin:gaussian:hmm} produced by PaRIS (right column) and the forward-only version of FFBSm (left column) using $(\N, \K) = (150, 2)$ andÊ $\N = 50$, respectively (yielding close to identical computational times). The boxes are based on 50 replicates of the estimates for the same fixed observation sequence and asterisks indicate exact values obtained with the disturbance smoother.}
 	\label{fig:lg:bp2}
\end{figure}

As a measure of numerical performance, we define \emph{efficiency} as inverse sample variance over computational time. \autoref{fig:lg:eff} reports the efficiencies by which the PaRIS and forward-only FFBSm algorithms estimate $\post{0:t \mid t} \af{t}^{(1)}$ using each $\N = 500$ particles. As evident from the plot, PaRIS exhibits a higher efficiency uniformly over all time points. The variance estimates were based on 50 replicates. 

\begin{figure}
	\centering
  	\centerline{\includegraphics[width=0.75\textwidth]{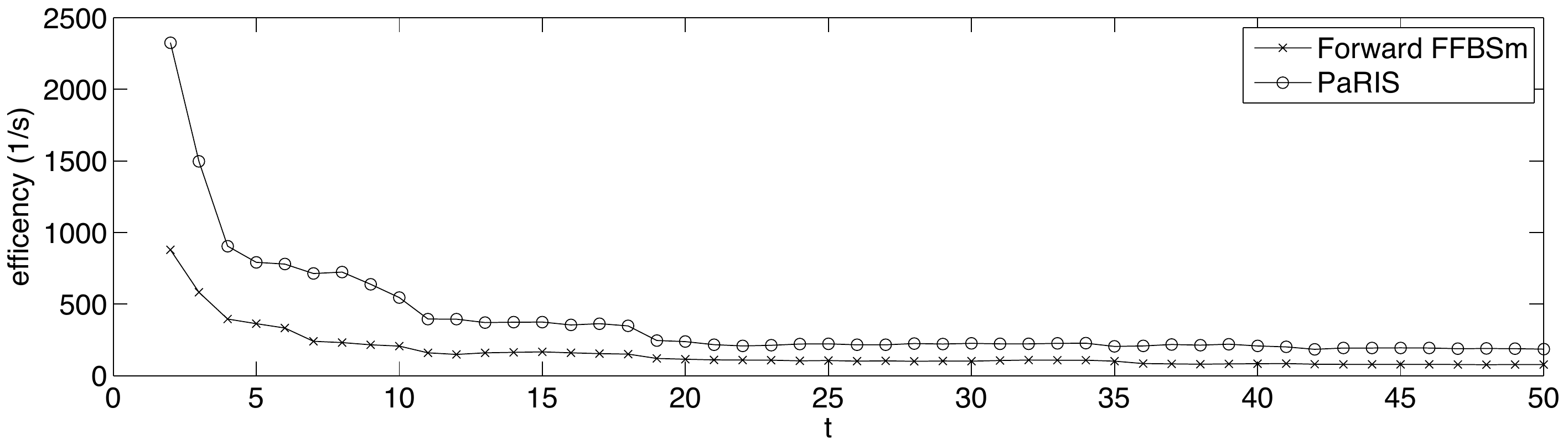}}
 	\caption{Estimated efficiencies for the PaRIS and forward-only FFBSm algorithms using each $\N = 500$ particles. (The first time step is removed from the plot due to very high efficiencies for both algorithms.)}
 	\label{fig:lg:eff}
\end{figure}

In order to examine the dependence of the performance of PaRIS on the design of the precision parameter $\K$, we produced estimates of $\post{0:t \mid t} \af{t}^{(1)}$Êfor $t \in \intvect{0}{1000}$ using the algorithm for each of the precision parameters $\K \in \{ 1, 2, 3, 4, 10, 30 \}$. All these estimators were computed on the basis of the same forward particles, so also an additional FFBSm-based estimator. This experiment was, in order to estimate the variances of the (seven) different estimators, replicated $100$ times for the same fixed sequence of observations. \autoref{fig:var}, displaying estimated variance as a function of time, shows a momentous difference between the cases $\K = 1$ and $\K > 1$ (note the difference in y-axis scale between the two graphs); the graphs in the top ($\K = 1$) and bottom ($\K > 1$) figures exhibit variance growths that appear to be close to quadratic and linear, respectively, which is well in accordance with the theory. Increasing the precision parameter $\K$ Êfrom $2$ to $4$ implies some decrease of variance, while increasing the same from $4$ to $30$ has only marginal effect on the accuracy of the estimator (the difference between the variances corresponding to $\K = 10$ and $\K = 30$ is close to indistinguishable). This is perfectly in line with the theoretical results obtained in \autoref{sec:main:results}, where the second term of the variance bound in \autoref{thm:linear:bound:normalized:PaRIS} is inversely proportional to the precision parameter. Finally, ratios of variances of estimators associated with different $\K$ are displayed in~\autoref{fig:quot}, which shows a linearly increasing ratio of the variances associated with $\K = 1$ and $\K = 2$ and a close to constant ratio of the variances associated with $\K = 2$ and $\K = 3$. 

Finally, in order to illustrate our algorithm's capacity of coping with particle path degeneracy, we report, in \autoref{fig:support}, the ratios $\# \set{S}_t / \N^{t + 1}$, $t \in \intvect{0}{1000}$, where $\# \set{S}_t$ is the cardinality of the support of the PaRIS algorithm at time $t$ (in the notation of \autoref{sec:PaRIS}), for the precision parameters $\K \in \{Ê1, 2, 3, 10, 30 \}$. Here $\N = 100$, and again the estimators associated with different precision parameters were based on the same forward particles. The $95\%$ confidence bounds displayed the same plot were obtained on the basis of $100$ replicates of observation record. Judging by these confidence bounds, the dependence of the copiousness of the support on the observations is fairly robust. Interestingly, for $\K = 1$ the sequence of ratios tends quickly to zero, while letting $\K > 1$ stabilizes completely the support of the estimator. Already $\K = 2$ yields a support that involves, on the average and in the long run, more than $50\%$ of all forward particles. Again, increasing the precision parameter has some effect for moderate values of the same, say, up to $\K = 10$, while increasing the parameter further from $10$ to $30$ (which implies a significant increase of computational overhead) effects only marginally the cardinality of the support. Also this observation is perfectly in line with the theory presented in \autoref{sec:main:results}, consolidating our apprehension that only a modest value of $\K$ is required as long as $\K \geq 2$. 

\begin{figure}[htb]
\centering
  \centerline{\includegraphics[width=\textwidth]{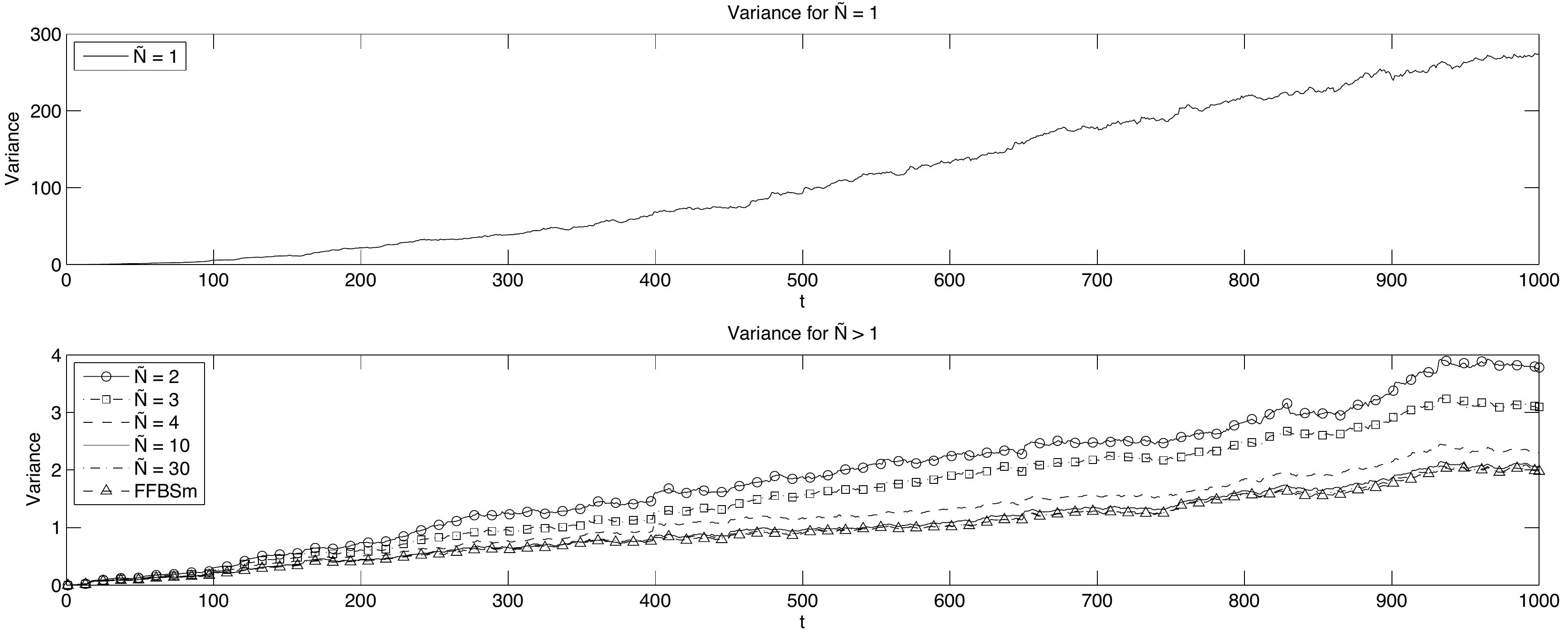}}
  \caption{Estimated variances of PaRIS estimators for $\K = 1$ (top graph) and $\K \in \{ 2 , 3, 4, 10, 30 \}$ (bottom graph) at different time steps $t \in \intvect{0}{1000}$. The bottom graph includes variance estimates of the forward-only FFBSm estimator. The variance estimates are based on $100$ replicates.}
  \label{fig:var}
\end{figure}

\begin{figure}[htb]
\centering
  \centerline{\includegraphics[width=\textwidth]{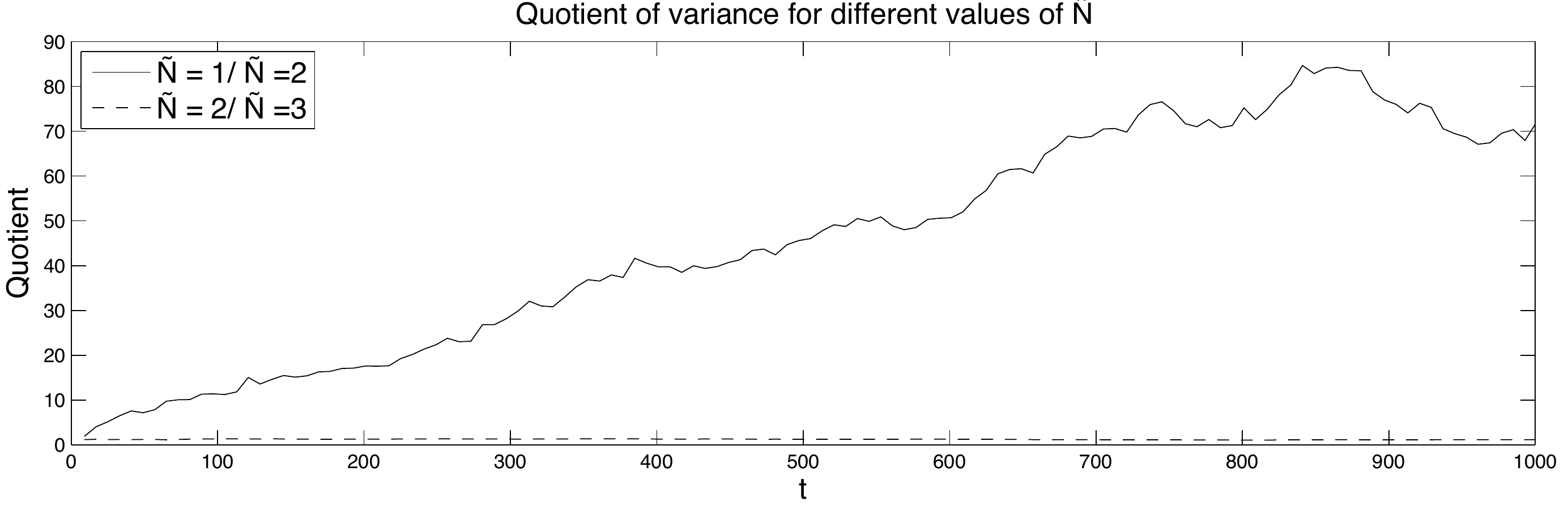}}
  \caption{Ratios of variances of the estimators associated with $\K = 1$ and $\K = 2$ (solid line) and $\K = 2$ and $\K = 3$ (dashed line).}
  \label{fig:quot}
\end{figure}

\begin{figure}[htb]
\centering
  \centerline{\includegraphics[width=\textwidth]{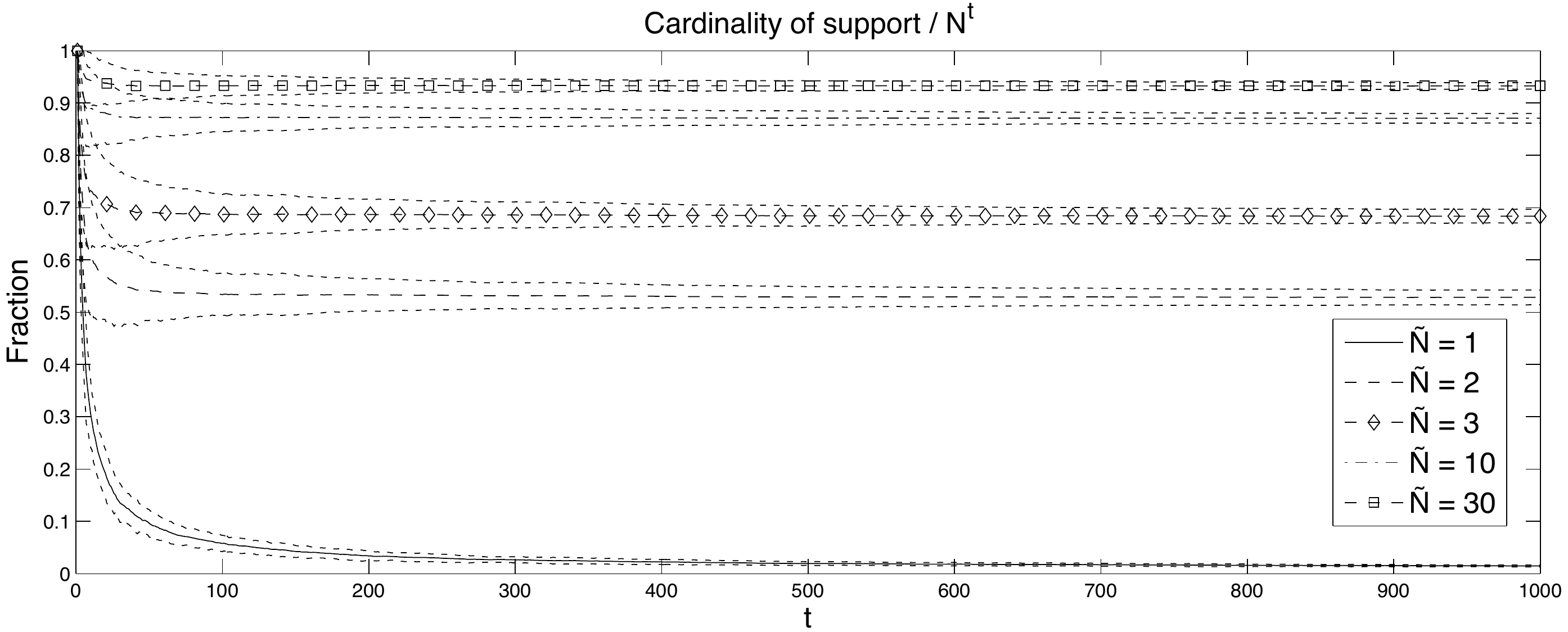}}
  \caption{Plot the ratios $\# \set{S}_t / \N^{t + 1}$, $t \in \intvect{0}{1000}$, for the precision parameters $\K \in \{Ê1, 2, 3, 10, 30 \}$ and $\N = 100$. The $95\%$ confidence bounds are obtained on the basis of $100$ replicates of the observation record.}
  \label{fig:support}
\end{figure}

\subsection{Stochastic volatility model}

For the sake of completeness we also consider a nonlinear model, namely the standard stochastic volatility model 
\begin{equation} \label{eq:stovol}
    \begin{split}
        X_{t + 1} &= \phi X_t + \sigmas \noises{t + 1} \\
        Y_t &= \beta \exp( X_t / 2 ) \noiseo{t} 
    \end{split}
    \quad (t \in \nset) \eqsp, 
\end{equation}
where $\set{X} = \set{Y} = \rset$Ê and $\{Ê\noises{t} \}_{t \in \nsetpos}$ and $\{Ê\noiseo{t} \}_{t \in \nset}$  are as in the previous example. We assume that the model parameters $\phi \in \mathbb{R}$ and $(\sigma, \beta) \in (\rsetpos)^2$ are known and that the model is well-specified. Our aim is to compute, using again PaRIS and the forward-only implementation of FFBSm, smoothed expectations of the sufficient statistics 
\begin{equation} \label{eqn:sv:af}
    \af{t}^{(1)}(x_{0:t}) \eqdef \sum_{s = 0}^t x_s^2, \quad \af{t}^{(2)}(x_{0:t}) \eqdef \sum_{s = 0}^{s - 1} x_s x_{s + 1} \quad (x_{0:t} \in \set{X}^{t + 1}) 
\end{equation}
for a model parameterized by $(\phi, \sigma, \beta) = ({.975}, {.16}, {.63})$. In this case, both algorithms used $\N = 250$ particles and the precision parameter of PaRIS was set to $\K = 2$. \autoref{fig:sv:bp} shows box plots based on $100$ replicates of estimates of $\post{0:t \mid t} \af{t}^{(i)} / t$, for $i \in \{ 1, 2 \}$ and $t \in \{ {2,\!000}, {4,\!000}, {6,\!000}, {8,\!000}, {10,\!000} \}$, obtained using these methods. Even though the variance and the bias of the estimates produced by the two algorithms are comparable, PaRIS was now \emph{5 times faster} than the FFBSm algorithm. 

\begin{figure}
	\centering
  	\centerline{\includegraphics[width=\textwidth]{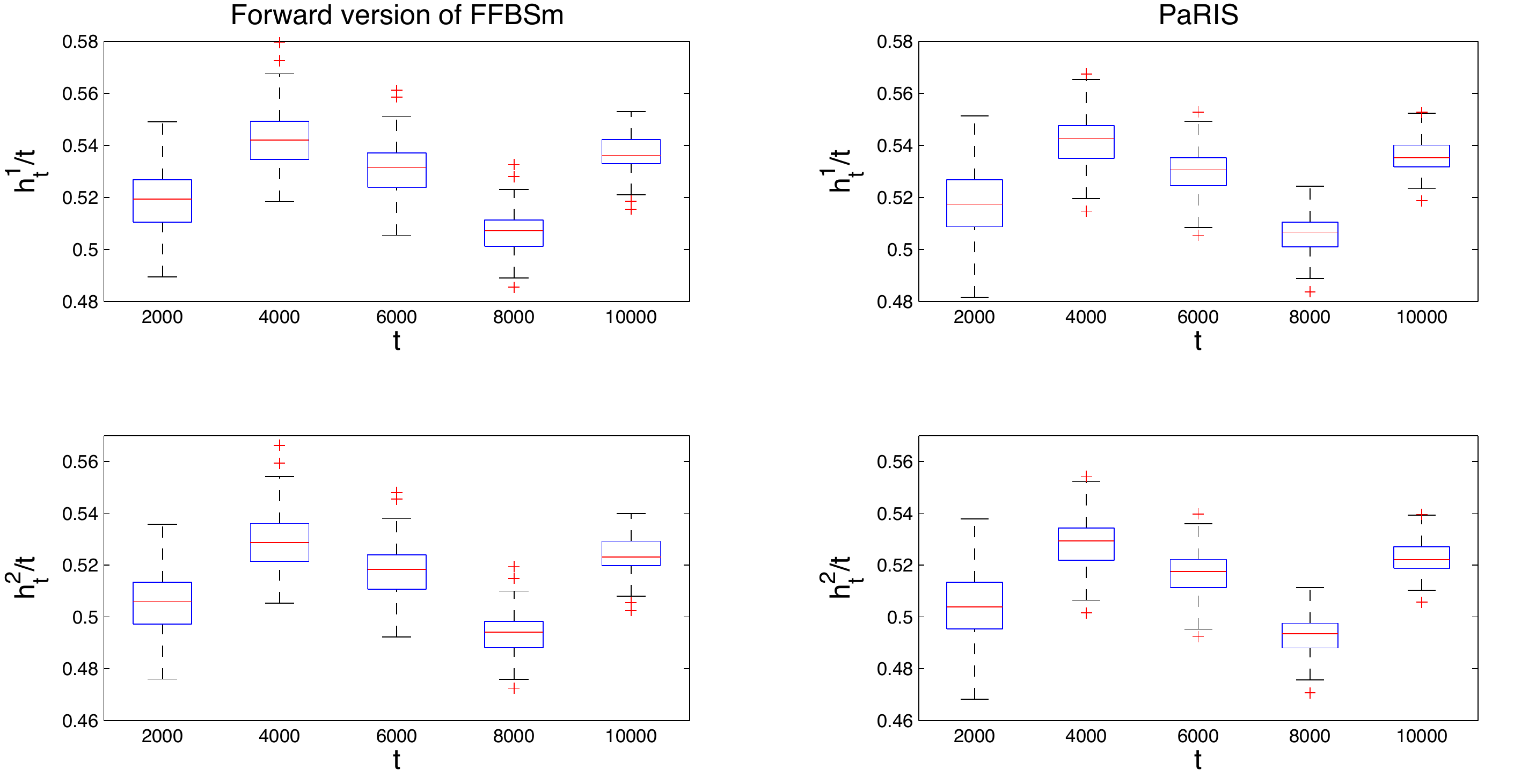}}
 	\caption{Box plots of estimates of smoothed sufficient statistics~\eqref{eqn:sv:af} for the stochastic volatility model \eqref{eq:stovol} produced by PaRIS (right column) and the forward-only version of FFBSm (left column) using $(\N, \K) = (250, 2)$ andÊ $\N = 250$, respectively. With this parameterization, PaRIS was 5 times faster than the FFBSm algorithm. The boxes are based on 100 replicates of the estimates for the same fixed observation sequence.}
 	\label{fig:sv:bp}
\end{figure}

\subsection{Some comments on the implementation}

When applying accept-reject-based backward sampling (\autoref{alg:accept:reject}), some acceptance probabilities will be small due to the random support of the particle-based backward kernel. In order to avoid getting stuck, it may be convenient to equip the algorithm with a threshold for the number of trials used at each accept-reject operation; when the threshold is reached, accept-reject sampling is cancelled and replaced by a draw from original distribution (recall that we are just using accept-reject sampling in order to reduce the computational work). \autoref{fig:LG:num:tries} displays computational time as a function of the size of this threshold for the linear Gaussian model and $\N = 250$ particles. Interestingly, the graph has a minimum for the threshold value $14$, and using this value we run the algorithm and counted the number of trials at any accept-reject sampling operation. The outcome is presented in the histogram plot to the right, from which it is clear that the majority of the particles are accepted after just a few trials (moreover, an index is most commonly accepted at once). In addition, at only $3.55\%$ of the occasions, the number of trials exceeded the threshold. Needless to say, the optimal threshold depends on the model as well as the number of particles (when the number of particles is small, a too high threshold may have significant negative effect on the computational efficiency; on the contrary, when the number of particles is large, the performance of the algorithm is relatively robust vis-\`a-vis the design of the threshold). Further simulations not presented here indicate however that a threshold value around $\sqrt{N}$ could be a rule of thumb. 

\begin{figure}
\centering
\includegraphics[width = \linewidth]{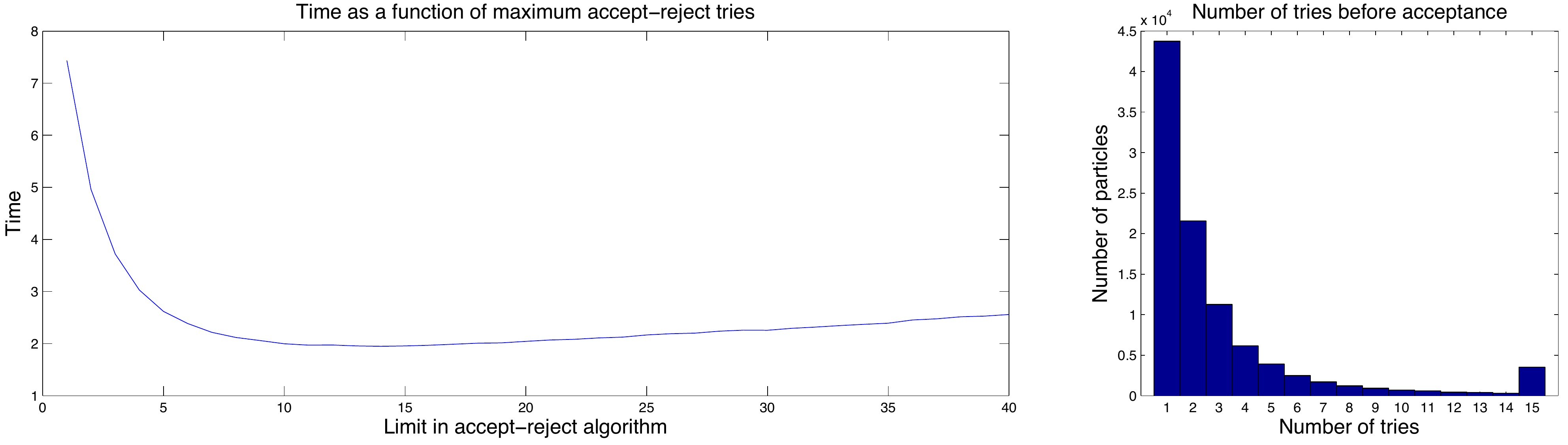}
\caption{Computational time as a function of the size of the accept-reject threshold for the linear Gaussian model and $\N = 250$ particles (left panel). The histogram to the right displays the number of trials needed before acceptance at any accept-reject sampling operation in algorithm when the threshold is $14$ (corresponding to minimal computational time); the bar at $15$ represents $3.55\%$ of the occasions.}
\label{fig:LG:num:tries}
\end{figure}

%% file: ow2014_con.tex

We have presented a novel algorithm, the particle-based, rapid incremental smoother, PaRIS, for computationally efficient online smoothing of additive state functionals in general HMMs. The algorithm, which is based on a backward decomposition of the smoothing distribution which can be implemented recursively for objective functions of additive type, can be viewed as a hybrid between the forward-only implementation of FFBSm and the FFBSi algorithm; more specifically, forward-only FFBSm may be viewed as a Rao-Blackwellized version of PaRIS. The algorithm is furnished with a number of convergence results, where the main result is a CLT of PaRIS's Monte Carlo output at the rate $\sqrt{\N}$. The analysis of PaRIS is considerably more involved than that of the FFBSi algorithm due to the complex dependence structure introduced by the retrospective simulation (on the contrary to FFBSi, where the trajectories are conditionally independent given the particles generated in the forward pass). Interestingly, the design of the precision parameter, i.e., the number of Monte Carlo simulations used for approximating the backward decomposition, turns out to be critical, since using a single backward draw yields a degeneracy phenomenon that resembles closely that of the Poor man's smoother. However, as established theoretically as well as through simulations, using at lest \emph{two} such draws stabilizes completely the support of the estimator. For $\K \geq 2$ Êwe are able to derive $\ordo(1 + 1/(\K - 1))$ and $\ordo(t \{Ê1 + 1/(\K- 1) \})$ bounds on the asymptotic variance in the cases of marginal and joint smoothing, respectively, and since the second term of these bounds is inversely proportional to the precision parameter, we suggest this parameter to be kept at a moderate value in order to gain computational speed. As known to the authors, this is the first analysis ever of this kind. 

The algorithm we propose has a linear complexity in the number of particles while the forward-only implementation of FFBSm has a quadratic complexity, and a numerical comparison between the two shows clearly that PaRIS achieves the same accuracy as FFBSm at a considerably lower computational cost. In addition, similarly to forward-only FFBSm, our smoother has limited and constant memory requirements, as it needs only the current particle sample and a set of estimated auxiliary statistics to be stored at each iteration. 

Smoothing of additive state functionals is a key ingredient of most---frequentistic or Bayesian---online parameter estimation techniques for HMMs. Since these applications are most often characterized by strict computational requirements, PaRIS can be naturally cast into any such framework.

%% file: ow2014_proofs.tex

\subsection{Two prefatory lemmas}

\begin{lemma} \label{lemma:property-1}
    For all $t \in \nset$ and $(\testf[t + 1], \testfp[t + 1]) \in \bmf{\alg{X}}^2$ it holds that 
    $$
        \post{t + 1}(\tstat{t + 1} \af{t + 1} \testf[t + 1] + \testfp[t + 1]) = \frac{\post{t}\{\tstat{t} \af{t} \lk{t} \testf[t + 1] + \lk{t}(\term{t} \testf[t + 1] + \testfp[t + 1]) \}}{\post{t} \lk{t} \1{\set{X}}} \eqsp. 
    $$
\end{lemma}
\begin{proof}
    By combining the definitions of $\tstat{t}$, $\lk{t}$ and using reversibility, 
    \[
        \begin{split}
            \post{t} \lk{t}(\tstat{t + 1} \af{t + 1} \testf[t + 1]) &= \post{t} \hk \{\bk{\post{t}} (\tstat{t} \af{t} + \term{t}) \md{t + 1} \testf[t + 1] \} \\
            &= \post{t} \hk \varotimes \bk{\post{t}} \{( \tstat{t} \af{t} + \term{t} ) \md{t + 1} \testf[t + 1] \} \\
            &= \post{t} \varotimes \hk \{( \tstat{t} \af{t} + \term{t} ) \md{t + 1} \testf[t + 1] \} \\ 
            &= \post{t} \{ \tstat{t} \af{t} \lk{t} \testf[t + 1] + \lk{t}( \term{t} \testf[t + 1] ) \}\eqsp.
        \end{split}
    \]
    Now the statement of the lemma follows by dividing both sides of the previous equation by $\post{t} \lk{t} \1{\set{X}}$ and using the identity $\post{t + 1} = \post{t} \lk{t} / \post{t} \lk{t} \1{\set{X}}$.
\end{proof}


\begin{lemma} \label{lemma:property-2}
    For all $t \in \nset$, $(\testf[t + 1], \testfp[t + 1]) \in \bmf{\alg{X}}^2$, and $(\N, \K) \in (\nsetpos)^2$ the random variables $\{ \wgt{t + 1}{i} (\tstat[i]{t + 1} \testf[t + 1](\epart{t + 1}{i}) + \testfp[t + 1](\epart{t + 1}{i})) \}_{i = 1}^\N$ are, conditionally on $\partfiltbar{t - 1}$, i.i.d. with expectation
    \begin{equation} \label{eq:cond:exp}
    	\E \left[ \wgt{t + 1}{1} \{\tstat[1]{t + 1} \testf[t + 1](\epart{t + 1}{1}) + \testfp[t + 1](\epart{t + 1}{1}) \} \mid \partfiltbar{t} \right] = \sum_{i = 1}^\N \frac{\wgt{t}{i}}{\wgtsum{t}} \{ \tstat[i]{t} \lk{t} \testf[t + 1](\epart{t}{i}) + \lk{t}(\term{t} \testf[t + 1] + \testfp[t + 1])(\epart{t}{i}) \} \eqsp.
    \end{equation}
\end{lemma}

\begin{proof}
    The multinomial selection procedure implies that the particles $\{\epart{t + 1}{i} \}_{i = 1}^\N$ are i.i.d. conditionally on $\partfiltbar{t}$. Hence, since also the backward indices $\{\bi{t + 1}{i}{j} \}_{j = 1}^{\K}$ are i.i.d. conditionally on the particle $\epart{t}{i}$ and the $\sigma$-field $\partfiltbar{t}$, we conclude that $\{ \wgt{t + 1}{i} ( \tstat[i]{t + 1} \testf[t + 1](\epart{t + 1}{i}) + \testfp[t + 1](\epart{t + 1}{i})) \}_{i = 1}^\N$ are i.i.d. conditionally on $\partfiltbar{t}$. 
    
    In order to compute the common conditional expectation we decompose the same according to 
    \begin{multline*}
        \E \left[ \wgt{t + 1}{1} \{\tstat[1]{t + 1} \testf[t + 1](\epart{t + 1}{1}) + \testfp[t + 1](\epart{t + 1}{1}) \} \mid \partfiltbar{t} \right] \\ 
        = \E \left[ \wgt{t + 1}{1} \tstat[1]{t + 1} \testf[t + 1](\epart{t + 1}{1}) \mid \partfiltbar{t} \right] + \E \left[ \wgt{t + 1}{1} \testfp[t + 1](\epart{t + 1}{1}) \mid \partfiltbar{t} \right] \eqsp, 
    \end{multline*}
    where, by the tower property,
    \[
        \begin{split}
        \lefteqn{\E \left[ \wgt{t + 1}{1} \tstat[1]{t + 1} \testf[t + 1](\epart{t + 1}{1}) \mid \partfiltbar{t} \right]} \\
        	&= \E \left[ \wgt{t + 1}{1} \testf[t + 1](\epart{t + 1}{1}) \E \left[ \tstat[1]{t + 1} \mid \partfiltbar{t} \vee \partfilt{t + 1} \right] \mid \partfiltbar{t} \right] \\
        	&= \E \left[ \wgt{t + 1}{1} \testf[t + 1](\epart{t + 1}{1}) \sum_{\ell = 1}^\N \frac{\wgt{t}{\ell} \hd(\epart{t}{\ell}, \epart{t + 1}{1})}{\sum_{\ell' = 1}^\N \wgt{t}{\ell'} \hd(\epart{t}{\ell'}, \epart{t + 1}{1})} \{\tstat[\ell]{t} + \term{t}(\epart{t}{\ell}, \epart{t + 1}{1}) \} \mid \partfiltbar{t} \right] \\
        	&= \sum_{i = 1}^\N \frac{\wgt{t}{i}}{\wgtsum{t}} \int \hd(\epart{t}{i}, x) \md{t + 1}(x) \testf[t + 1](x) \sum_{\ell = 1}^\N \frac{\wgt{t}{\ell} \hd(\epart{t}{\ell}, x)}{\sum_{\ell' = 1}^\N  \wgt{t}{\ell'} \hd(\epart{t}{\ell'}, x) } \{ \tstat[\ell]{t} + \addf{t}(\epart{t}{\ell}, x) \} \, \refM(\rmd x) \\
        	 &= \sum_{\ell = 1}^\N \frac{\wgt{t}{\ell}}{\wgtsum{t}} \{ \tstat[\ell]{t} \lk{t} \testf[t + 1](\epart{t}{\ell}) + \lk{t} (\term{t} \testf[t + 1])(\epart{t}{\ell}) \} \eqsp. \\
        \end{split}
    \]
    We conclude the proof by noting that 
    $$
        \E \left[ \wgt{t + 1}{1} \testfp[t + 1](\epart{t + 1}{1}) \mid \partfiltbar{t} \right] = \sum_{\ell = 1}^\N \frac{\wgt{t}{\ell}}{\wgtsum{t}} \hk(\md{t + 1} \testfp[t + 1])(\epart{t}{\ell}) = \sum_{\ell = 1}^\N \frac{\wgt{t}{\ell}}{\wgtsum{t}} \lk{t} \testfp[t + 1](\epart{t}{\ell}) \eqsp. 
    $$
\end{proof}


\subsection{Proof of \autoref{thm:hoeffding:affine}}

We proceed by induction and assume that the claim of the theorem holds for $t \in \nset$. To establish (i) for $t + 1$, write, using \autoref{lemma:property-1},
\begin{multline*}
    \frac{1}{\N} \sum_{i = 1}^\N \wgt{t+1}{i} \{\tstat[i]{t + 1} \testf[t + 1](\epart{t+1}{i}) + \testfp[t + 1](\epart{t+1}{i}) \} - 
    \post{t} \lk{t} (\tstat{t + 1} \af{t + 1} \testf[t + 1] + \testfp[t + 1]) \\
    = \frac{1}{\N} \sum_{i = 1}^\N \wgt{t + 1}{i} \{\tstat[i]{t + 1} \testf[t + 1](\epart{t + 1}{i}) + \testfp[t + 1](\epart{t+1}{i}) \} - \E\left[ \wgt{t + 1}{1} \{\tstat[1]{t + 1} \testf[t + 1](\epart{t + 1}{1}) + \testfp[t + 1](\epart{t + 1}{1}) \} \mid \partfiltbar{t} \right] \\
    + \sum_{i = 1}^\N \frac{\wgt{t}{i}}{\wgtsum{t}} \{ \tstat[i]{t} \lk{t} \testf[t + 1](\epart{t}{i}) + \lk{t}(\addf{t} \testf[t + 1] + \testfp[t + 1])(\epart{t}{i}) \} - \post{t}\{\tstat{t} \af{t} \lk{t} \testf[t + 1] + \lk{t}(\addf{t} \testf[t + 1] + \testfp[t + 1]) \} \eqsp.
\end{multline*}
Since the functions $\lk{t} \testf[t + 1]$ and $\lk{t}(\term{t} \testf[t + 1] + \testfp[t + 1])$ belong to $\bmf{\alg{X}}$, the induction hypothesis (ii) implies that for all $\varepsilon \in \rsetpos$,
\begin{multline*}
    \prob \left( \left| \sum_{i = 1}^\N \frac{\wgt{t}{i}}{\wgtsum{t}} \{ \tstat[i]{t} \lk{t} \testf[t + 1](\epart{t}{i}) + \lk{t}(\addf{t} \testf[t + 1] + \testfp[t + 1])(\epart{t}{i}) \} - \post{t}\{\tstat{t} \af{t} \lk{t} \testf[t + 1] + \lk{t}(\addf{t} \testf[t + 1] + \testfp[t + 1]) \} \right| \geq \varepsilon \right) \\
    \leq c_t \exp(- \tilde{c}_t \N \varepsilon^2) \eqsp. 
\end{multline*}
In addition, by \autoref{lemma:property-2}, $\{ \wgt{t + 1}{i} ( \tstat[i]{t + 1} \testf[t + 1](\epart{t + 1}{i}) + \testfp[t + 1](\epart{t + 1}{i})) \}_{i = 1}^\N$ are conditionally i.i.d. given $\partfiltbar{t}$; thus, since for all $i \in \intvect{1}{\N}$,
$$
    |\wgt{t + 1}{i} \{\tstat[i]{t + 1} \testf[t + 1](\epart{t + 1}{i}) + \testfp[t + 1](\epart{t + 1}{i}) \} |\leq \supn{\md{t + 1}} ( \supn{\af{t + 1}} \supn{\testf[t + 1]} + \supn{\testfp[t + 1]} ) < \infty, 
$$
the conditional Hoeffding inequality provides constants $(d, \tilde{d}) \in (\rsetpos)^2$ such that 
\begin{multline*}
    \prob \left( \left| \frac{1}{\N} \sum_{i = 1}^\N \wgt{t + 1}{i} \{\tstat[i]{t + 1} \testf[t + 1](\epart{t + 1}{i}) + \testfp[t + 1](\epart{t+1}{i}) \} - \E\left[ \wgt{t + 1}{1} \{\tstat[1]{t + 1} \testf[t + 1] (\epart{t + 1}{1}) + \testfp[t + 1](\epart{t + 1}{1}) \} \mid \partfiltbar{t} \right] \right| \geq \varepsilon \right) \\
    \leq d \exp(- \tilde{d} \N \varepsilon^2) \eqsp. 
\end{multline*}
This establishes (i). 

The inequality (ii) for the self-normalized estimator is an immediate consequence of (i) and the generalized Hoeffding inequality in \cite[Lemma~4]{douc:garivier:moulines:olsson:2010}. 

Finally, we conclude the proof by checking that the result is straightforwardly true for the base case $t = 0$, since $\tstat{0} \af{0} = 0$, $\tstat[i]{0} = 0$ for all $i \in \intvect{1}{\N}$, and the weighted sample $\{ (\epart{0}{i}, \wgt{0}{i}) \}_{i = 1}^\N$ (targeting $\post{0}$) is generated by standard importance sampling.  


\subsection{Proof of \autoref{thm:CLT:affine}}

We proceed by induction and suppose that the claim of the theorem holds true for some $t \in \nset$. Thus, pick $(\testf[t + 1], \testfp[t + 1]) \in \bmf{\alg{X}}^2$ and assume first that $\post{t + 1}(\tstat{t + 1} \af{t + 1} \testf[t + 1] + \testfp[t + 1]) = 0$. Write
\begin{multline*}
    \sqrt{\N} \sum_{i = 1}^{\N} \frac{\wgt{t}{i}}{\wgtsum{t}} \{\tstat[i]{t + 1} \testf[t + 1](\epart{t + 1}{i}) + \testfp[t + 1](\epart{t + 1}{i}) \} = \N \wgtsum{t}^{-1} \frac{1}{\sqrt{\N}} \sum_{i = 1}^{\N} \left( \vphantom{\sum_{\ell = 1}^\N} \wgt{t + 1}{i} \{\tstat[i]{t + 1} \testf[t + 1](\epart{t + 1}{i}) + \testfp[t + 1](\epart{t + 1}{i}) \} \right. \\
    \left. - \sum_{\ell = 1}^\N \frac{\wgt{t}{\ell}}{\wgtsum{t}}\{\tstat[\ell]{t} \lk{t} \testf[t + 1](\epart{t + 1}{\ell}) + \lk{t}(\term{t} \testf[t + 1] + \testfp[t + 1])(\epart{t + 1}{\ell}) \} \right) \\
    + \N \wgtsum{t}^{-1} \sqrt{\N} \sum_{\ell = 1}^\N \frac{\wgt{t}{\ell}}{\wgtsum{t}}\{\tstat[\ell]{t} \lk{t} \testf[t + 1](\epart{t + 1}{\ell}) + \lk{t}(\term{t} \testf[t + 1] + \testfp[t + 1])(\epart{t + 1}{\ell}) \} \eqsp,  
\end{multline*}
where, by \autoref{thm:hoeffding:affine}, $\N \wgtsum{t}^{-1}$ tends to $(\post{t} \lk{t} \1{\set{X}})^{-1}$ in probability. In order to establish the weak convergence of the first term, we will apply \autoref{thm:DM:A-3} to the triangular array 
$$
    \arr{i} \eqdef \frac{1}{\K \sqrt{\N}} \sum_{j = 1}^{\K} \arrterm{\epart{t + 1}{i}}{\bi{t + 1}{i}{j}} \quad (i \in \intvect{1}{\N}, \N \in \nsetpos) \eqsp, 
$$
where 
\begin{multline} \label{eq:def:upsilon:tilde}
    \arrterm{x}{j} \eqdef \md{t + 1}(x) \{(\tstat[j]{t} + \term{t}(\epart{t}{j}, x)) \testf[t + 1](x) + \testfp[t + 1](x)\} \\- \sum_{\ell = 1}^\N \frac{\wgt{t}{\ell}}{\wgtsum{t}}\{\tstat[\ell]{t} \lk{t} \testf[t + 1](\epart{t}{\ell}) + \lk{t}(\term{t} \testf[t + 1] + \testfp[t + 1])(\epart{t}{\ell}) \}  \quad (x \in \set{X}, j \in \intvect{1}{\K}, \N \in \nsetpos) \eqsp,  
\end{multline}
furnished with the filtration $\{\partfiltbar{t} \}_{\N \in \nsetpos}$. Note that for all $i \in \intvect{1}{\N}$, $\cexp[txt]{\arr{i}}{\partfiltbar{t}} = 0$ (by \autoref{lemma:property-2}) and $|\arr{i}| \leq 2 \supn{\md{t + 1}}(\supn{\af{t + 1}} \supn{\testf[t + 1]} + \supn{\testfp[t + 1]}) / \sqrt{\N}$. To check the condition \B{ass:DM:norm:asymptotic:variance} in \autoref{thm:DM:A-3}, write, using, first, that $\{\arr{i} \}_{i = 1}$ are conditionally i.i.d. given $\partfiltbar{t}$ and, second, that the backward indices $\{\bi{t + 1}{i}{j} \}_{j = 1}^{\K}$ are, for all $i \in \intvect{1}{\N}$, i.i.d. conditionally on $\partfiltbar{t}$ and $\epart{t + 1}{i}$,
\begin{align}
    \sum_{i = 1}^{\N} \E \left[ \arr[2]{i} \mid \partfiltbar{t} \right] &= \K[-2] \E \left[ \left( \sum_{j = 1}^{\K} \arrterm{\epart{t + 1}{1}}{\bi{t + 1}{1}{j}} \right)^2 \mid \partfiltbar{t} \right] \nonumber \\
    &= \K[-1] \E\left[ \E \left[ \arrterm[2]{\epart{t + 1}{1}}{\bi{t + 1}{1}{1}}  \mid \partfiltbar{t} \vee \partfilt{t + 1} \right] \mid \partfiltbar{t} \right] \label{CLT:decomp:var-1} \\
    &\quad + \K[-1](\K - 1) \E \left[ \E^2 \left[ \arrterm{\epart{t + 1}{1}}{\bi{t + 1}{1}{1}} \mid \partfiltbar{t} \vee \partfilt{t + 1} \right] \mid \partfiltbar{t} \right] \eqsp. \label{CLT:decomp:var-2}
\end{align}
We treat separately the two terms \eqref{CLT:decomp:var-1} and \eqref{CLT:decomp:var-2}. Concerning \eqref{CLT:decomp:var-1},
\[
    \begin{split} 
        \E\left[ \E \left[ \arrterm[2]{\epart{t + 1}{1}}{\bi{t + 1}{1}{1}}  \mid \partfiltbar{t} \vee \partfilt{t + 1} \right] \mid \partfiltbar{t} \right] &= \E \left[ \sum_{\ell = 1}^{\N} \arrterm[2]{\epart{t + 1}{1}}{\ell} \frac{\wgt{t}{\ell} \hd(\epart{t}{\ell}, \epart{t + 1}{1})}{\sum_{\ell' = 1}^{\N} \wgt{t}{\ell'} \hd(\epart{t}{\ell'}, \epart{t + 1}{1})} \mid \partfiltbar{t} \right] \\
         &= \sum_{i = 1}^{\N} \frac{\wgt{t}{i}}{\wgtsum{t}} \int \hd(\epart{t}{i}, x) \sum_{\ell = 1}^{\N} \arrterm[2]{x}{\ell} \frac{\wgt{t}{\ell} \hd(\epart{t}{\ell}, x)}{\sum_{\ell' = 1}^{\N} \wgt{t}{\ell'} \hd(\epart{t}{\ell'}, x)} \, \refM(\rmd x) \\
         &= \sum_{\ell = 1}^{\N} \frac{\wgt{t}{\ell}}{\wgtsum{t}} \int \hd(\epart{t}{\ell}, x) \arrterm[2]{x}{\ell} \, \refM(\rmd x) \eqsp.
    \end{split}
\]
Now, using the definition \eqref{eq:def:upsilon:tilde}, 
\begin{multline*}
    \sum_{\ell = 1}^{\N} \frac{\wgt{t}{\ell}}{\wgtsum{t}} \int \hd(\epart{t}{\ell}, x) \arrterm[2]{x}{\ell} \, \refM(\rmd x) = \sum_{\ell = 1}^{\N} \frac{\wgt{t}{\ell}}{\wgtsum{t}} (\tstat[\ell]{t})^2 \lk{t}(\md{t + 1} \testf[t + 1]^2)(\epart{t}{\ell}) \\ 
    + \sum_{\ell = 1}^{\N} \frac{\wgt{t}{\ell}}{\wgtsum{t}} \left( \tstat[\ell]{t} 2 \lk{t} \{\md{t + 1} \testf[t + 1]( \term{t} \testf[t + 1] + \testfp[t + 1]) \}(\epart{t}{\ell}) + \lk{t}\{\md{t + 1}(\term{t} \testf[t + 1] + \testfp[t + 1])^2 \}(\epart{t}{\ell}) \right) \\
    - \left( \sum_{\ell = 1}^\N \frac{\wgt{t}{\ell}}{\wgtsum{t}}\{\tstat[\ell]{t} \lk{t} \testf[t + 1](\epart{t}{\ell}) + \lk{t}(\term{t} \testf[t + 1] + \testfp[t + 1])(\epart{t}{\ell}) \} \right)^2 \eqsp. 
\end{multline*}
In the previous expression, by \autoref{lemma:T2:conv},
$$
    \sum_{\ell = 1}^{\N} \frac{\wgt{t}{\ell}}{\wgtsum{t}} (\tstat[\ell]{t})^2 \lk{t}(\md{t + 1} \testf[t + 1]^2)(\epart{t}{\ell}) \convp \post{t}\{\tstat{t}^2 \af{t} \lk{t}(\md{t + 1} \testf[t + 1]^2) \} + \eta_t\{\lk{t}(\md{t + 1} \testf[t + 1]^2) \} \eqsp, 
$$
where the functional $\eta_t$ is defined in \eqref{eq:def:covterm}, and, by \autoref{thm:hoeffding} and \autoref{lemma:property-1} (recalling that $\post{t} \lk{t}(\tstat{t + 1} \testf[t + 1] + \testfp[t + 1]) = 0$ by assumption),
\begin{align}
    &\sum_{\ell = 1}^{\N} \frac{\wgt{t}{\ell}}{\wgtsum{t}} \left( \tstat[\ell]{t} 2 \lk{t} \{\md{t + 1} \testf[t + 1]( \term{t} \testf[t + 1] + \testfp[t + 1]) \}(\epart{t}{\ell}) + \lk{t}\{\md{t + 1}(\term{t} \testf[t + 1] + \testfp[t + 1])^2 \}(\epart{t}{\ell}) \right) \nonumber \\
    &\hspace{25mm} \convp 2 \post{t}(\tstat{t} \af{t} \lk{t} \{\md{t + 1} \testf[t + 1](\term{t} \testf[t + 1] + \testfp[t + 1])\})+ \post{t} \lk{t}\{\md{t + 1}(\term{t} \testf[t + 1] + \testfp[t + 1])^2 \} \eqsp, \nonumber \\
    &\sum_{\ell = 1}^\N \frac{\wgt{t}{\ell}}{\wgtsum{t}}\{\tstat[\ell]{t} \lk{t} \testf[t + 1](\epart{t}{\ell}) + \lk{t}(\term{t} \testf[t + 1] + \testfp[t + 1])(\epart{t}{\ell}) \} \convp 0 \eqsp. \label{eq:zero:convergence:term}
\end{align}
We hence conclude that
\begin{multline} \label{eq:as:var:first:term:limit}
    \eqref{CLT:decomp:var-1} \convp \K[-1] \left( \vphantom{ 2 \post{t}(\tstat{t} \af{t} \lk{t} \{\md{t + 1} \testf[t + 1](\term{t} \testf[t + 1] + \testfp[t + 1])\})+ \post{t} \lk{t}\{\md{t + 1}(\term{t} \testf[t + 1] + \testfp[t + 1])^2 \}} \post{t}\{\tstat{t}^2 \af{t} \lk{t}(\md{t + 1} \testf[t + 1]^2) \} \right. + \eta_t\{\lk{t}(\md{t + 1} \testf[t + 1]^2) \} \\
    \left. + 2 \post{t}(\tstat{t} \af{t} \lk{t} \{\md{t + 1} \testf[t + 1](\term{t} \testf[t + 1] + \testfp[t + 1])\})+ \post{t} \lk{t}\{\md{t + 1}(\term{t} \testf[t + 1] + \testfp[t + 1])^2 \} \right) \\
    = \K[-1] \left( \post{t} \lk{t}( \md{t + 1} \{ (\tstat{t} \af{t} + \term{t} ) \testf[t + 1] + \testfp[t + 1] \}^2)+ \eta_t \{\lk{t}(\md{t + 1} \testf[t + 1]^2 ) \} \right) \eqsp. 
\end{multline}

We turn to \eqref{CLT:decomp:var-2} and write 
\[
    \begin{split}
        \lefteqn{\E \left[ \E^2 \left[ \arrterm{\epart{t + 1}{1}}{\bi{t + 1}{1}{1}} \mid \partfiltbar{t} \vee \partfilt{t + 1} \right] \mid \partfiltbar{t} \right]} \hspace{20mm} \\ 
        &= \E \left[ \left( \sum_{\ell = 1}^{\N} \arrterm{\epart{t + 1}{1}}{\ell} \frac{\wgt{t}{\ell} \hd(\epart{t}{\ell}, \epart{t + 1}{1})}{\sum_{\ell' = 1}^{\N} \wgt{t}{\ell'} \hd(\epart{t}{\ell'}, \epart{t + 1}{1})} \right)^2 \mid \partfiltbar{t} \right] \\
                 &= \sum_{i = 1}^{\N} \frac{\wgt{t}{i}}{\wgtsum{t}} \int \hd(\epart{t}{i}, x) \left( \sum_{\ell = 1}^{\N} \arrterm{x}{\ell} \frac{\wgt{t}{\ell} \hd(\epart{t}{\ell}, x)}{\sum_{\ell' = 1}^{\N} \wgt{t}{\ell'} \hd(\epart{t}{\ell'}, x)} \right)^2 \, \refM(\rmd x) \eqsp,
    \end{split}
\]
where we note that  the right hand side can, by \eqref{eq:def:upsilon:tilde}, be written as $\post[part]{t} \hk \lebfun[\N]$, with
\begin{multline*}
    \lebfun[\N](x) \eqdef \left( \md{t + 1}(x) \testf[t + 1](x) \sum_{\ell = 1}^{\N} \frac{\wgt{t}{\ell} \hd(\epart{t}{\ell}, x)}{\sum_{\ell' = 1}^{\N} \wgt{t}{\ell'} \hd(\epart{t}{\ell'}, x)} \{ \tstat[\ell]{t} + \term{t}(\epart{t}{\ell}, x) \} + \md{t + 1}(x) \testfp[t + 1](x) \right. \\
    \left. - \sum_{\ell' = 1}^\N \frac{\wgt{t}{\ell'}}{\wgtsum{t}}\{\tstat[\ell']{t} \lk{t} \testf[t + 1](\epart{t}{\ell'}) + \lk{t}(\term{t} \testf[t + 1] + \testfp[t + 1])(\epart{t}{\ell'}) \} \right)^2 \quad (x \in \set{X}) \eqsp. 
\end{multline*}
Note that $\supn{\lebfun[\N]} \leq 4 \supn{\md{t + 1}}^2 (\supn{\af{t + 1}} \supn{\testf[t + 1]} + \supn{\testfp[t + 1]})^2$ for all $\N \in \nset$. Moreover, by \autoref{thm:hoeffding} (and, in particular, the implication \eqref{eq:zero:convergence:term}) it holds, for all $x \in \set{X}$, $\prob$-a.s., 
\[
\begin{split}
    \lebfun[\N](x) \rightarrow& \hphantom{\ }\md{t + 1}^2(x) \left( \testf[t + 1](x) \frac{\int \hd(\tilde{x}, x) \{\tstat{t} \af{t} (\tilde{x}) + \term{t}(\tilde{x}, x) \} \, \post{t}(\rmd \tilde{x})}{\int \hd(\tilde{x}, x) \, \post{t}(\rmd \tilde{x})} + \testfp[t + 1](x) \right)^2 \\
            &= \md{t + 1}^2(x) \{\testf[t + 1](x) \bk{\post{t}}(\tstat{t} \af{t} + \term{t})(x) + \testfp[t + 1](x) \}^2 \\
            &= \md{t + 1}^2(x) \{\testf[t + 1](x) \tstat{t + 1} \af{t + 1}(x) + \testfp[t + 1](x) \}^2 \eqsp. 
\end{split}
\]
Thus, under \autoref{ass:boundedness:g:q} we may apply \autoref{lemma:generalized:lebesgue}, yielding 
$$
    \post[part]{t} \hk \lebfun[\N] \convp \post{t} \lk{t}\{\md{t + 1}(\testf[t + 1] \tstat{t + 1} \af{t + 1} + \testfp[t + 1])^2\} \eqsp,
$$
and we may hence conclude that 
$$
    \eqref{CLT:decomp:var-2} \convp \K^{-1}(\K - 1) \post{t} \lk{t} \{ \md{t + 1} (\testf[t + 1] \tstat{t + 1} \af{t + 1} + \testfp[t + 1] )^2 \} \eqsp.
$$
Finally, by combining this limit with \eqref{eq:as:var:first:term:limit} we obtain, 
using the identity
\[
    \begin{split}
        \lefteqn{\post{t} \lk{t}( \md{t + 1} \{ (\tstat{t} \af{t} + \term{t} ) \testf[t + 1] + \testfp[t + 1] \}^2)- \post{t} \lk{t}\{\md{t + 1} (\testf[t + 1] \tstat{t + 1} \af{t + 1} + \testfp[t + 1])^2\}} \hspace{10mm} \\
        &= \post{t} \varotimes \hk( \md{t + 1}^2 \{ (\tstat{t} \af{t} + \term{t}) \testf[t + 1] + \testfp[t + 1] \}^2 - \md{t + 1}^2 \{ \testf[t + 1] \tstat{t + 1} \af{t + 1} + \testfp[t + 1] \}^2) \\
        &= \post{t} \hk \varotimes \bk{\post{t}} ( \md{t + 1}^2 \{ (\tstat{t} \af{t} + \term{t}) \testf[t + 1] + \testfp[t + 1] \}^2 - \md{t + 1}^2 \{ \testf[t + 1] \tstat{t + 1} \af{t + 1} + \testfp[t + 1] \}^2) \\
        &= \post{t} \hk \varotimes \bk{\post{t}} \{ \md{t + 1}^2 \testf[t + 1]^2 ( \tstat{t} \af{t} + \term{t} - \tstat{t + 1} \af{t + 1} )^2 \} \\
        &= \post{t} \lk{t} \bk{\post{t}} \{ \md{t + 1} \testf[t + 1]^2 ( \tstat{t} \af{t} + \term{t} - \tstat{t + 1} \af{t + 1} )^2 \} \eqsp,
    \end{split}
\]
the convergence
\begin{multline} \label{eq:var:conv-1}
    \sum_{i = 1}^{\N} \E \left[ \arr[2]{i} \mid \partfiltbar{t} \right] \convp \post{t} \lk{t} \{\md{t + 1} (\testf[t + 1] \tstat{t + 1} \af{t + 1} + \testfp[t + 1] )^2 \} \\ 
    + \K^{-1} \left( \post{t} \lk{t} \bk{\post{t}} \{\md{t + 1} \testf[t + 1]^2 ( \tstat{t} \af{t} + \term{t} - \tstat{t + 1} \af{t + 1} )^2 \} + \eta_t\{\lk{t}(\md{t + 1} \testf[t + 1]^2) \}\right) \eqsp,
\end{multline}
which verifies \B{ass:DM:norm:asymptotic:variance}. In order to check also the condition \B{ass:DM:norm:Lindeberg}, write, for $\varepsilon \in \rsetpos$,
\begin{multline*}
        \sum_{i = 1}^{\N} \cexp{\arr[2]{i} \1{\{|\arr{i}| \geq \varepsilon \}}}{\partfiltbar{t}} \leq 4 \supn{\md{t + 1}}^2 (\supn{\af{t + 1}} \supn{\testf[t + 1]} + \supn{\testfp[t + 1]})^2 \\
        \times \1{\{2 \supn{\md{t + 1}}(\supn{\af{t + 1}} \supn{\testf[t + 1]} + \supn{\testfp[t + 1]}) \geq \varepsilon \sqrt{\N} \}} \eqsp,
\end{multline*}
where the indicator function on the right hand side is zero for $\N$ large enough. This shows the condition \B{ass:DM:norm:Lindeberg}. Hence, for general $(\testf[t + 1], \testfp[t + 1])$ (by just replacing $\testfp[t + 1]$ by $\testfp[t + 1] - \post{t + 1}(\tstat{t + 1} \af{t + 1} \testf[t + 1] + \testfp[t + 1])$), by \autoref{thm:DM:A-3}, \cite[Lemma~A.5]{verge:delmoral:moulines:olsson:2014}, and Slutsky's lemma, 
$$
    \sqrt{\N} \sum_{i = 1}^{\N} \frac{\wgt{t + 1}{i}}{\wgtsum{t + 1}} \{\tstat[i]{t + 1} \testf[t + 1](\epart{t + 1}{i}) + \testfp[t + 1](\epart{t + 1}{i}) - \post{t + 1}(\tstat{t + 1} \af{t + 1} \testf[t + 1] + \testfp[t + 1]) \} \convd \asvar{t + 1}{\testf[t + 1]}{\testfp[t + 1]} Z \eqsp,
$$
where $Z$ is a standard Gaussian variable and 
\begin{multline} \label{eq:CLT:recursive:variance}
    \asvar[2]{t + 1}{\testf[t + 1]}{\testfp[t + 1]} \eqdef \frac{\post{t} \lk{t} (\md{t + 1} \{\testf[t + 1] \tstat{t + 1} \af{t + 1} + \testfp[t + 1] - \post{t + 1}(\tstat{t + 1} \af{t + 1} \testf[t + 1] + \testfp[t + 1]) \}^2 )}{(\post{t} \lk{t} \1{\set{X}})^2} \\ 
     + \sum_{\ell = 0}^t \K^{\ell - (t + 1)} \frac{\post{\ell} \lk{\ell} \{ \bk{\post{\ell}} ( \tstat{\ell} \af{\ell} + \term{\ell} - \tstat{\ell + 1} \af{\ell + 1} )^2 \lk{\ell + 1} \cdots \lk{t} (\md{t + 1} \testf[t + 1]^2) \}}{(\post{\ell} \lk{\ell} \cdots \lk{t - 1} \1{\set{X}} )( \post{t} \lk{t} \1{\set{X}} )^2} \\
      + \frac{\sigma^2_t (\lk{t} \testf[t + 1], \lk{t} \{ \term{t + 1} \testf[t + 1] + \testfp[t + 1] - \post{t + 1}(\tstat{t + 1} \af{t + 1} \testf[t + 1] + \testfp[t + 1]) \} )}{(\post{t} \lk{t} \1{\set{X}})^2} \eqsp. 
\end{multline}
We now apply the induction hypothesis to the last term. For this purpose, note that, by \autoref{lemma:property-1},
\begin{multline*}
    \af{t} \lk{t} \testf[t + 1] + \lk{t} \{ \term{t + 1} \testf[t + 1] + \testfp[t + 1] - \post{t + 1}(\tstat{t + 1} \af{t + 1} \testf[t + 1] + \testfp[t + 1]) \} \\
    - \post{t}( \af{t} \lk{t} \testf[t + 1] + \lk{t} \{ \term{t + 1} \testf[t + 1] + \testfp[t + 1] - \post{t + 1}(\tstat{t + 1} \af{t + 1} \testf[t + 1] + \testfp[t + 1]) \} ) \\
    = \lk{t} \{\af{t + 1}  \testf[t + 1] + \testfp[t + 1] - \post{t + 1}( \tstat{t + 1} \af{t + 1} \testf[t + 1] + \testfp[t + 1] ) \} \eqsp,
\end{multline*}
yielding, for all $\nset \ni s < t$, 
\[ 
    \begin{split}
        \lefteqn{\BFcent{s + 1}{t}( \af{t} \lk{t} \testf[t + 1] + \lk{t} \{ \term{t + 1} \testf[t + 1] + \testfp[t + 1] - \post{t + 1}(\tstat{t + 1} \af{t + 1} \testf[t + 1] + \testfp[t + 1]) \} )} \hspace{10mm} \\
        &= \BF{s + 1}{t} \lk{t} \{ \af{t + 1}  \testf[t + 1] + \testfp[t + 1] - \post{t + 1}(\tstat{t + 1} \af{t + 1} \testf[t + 1] + \testfp[t + 1]) \} \\
        &= \BFcent{s + 1}{t + 1}( \af{t + 1}  \testf[t + 1] + \testfp[t + 1] ) \eqsp.
    \end{split}
\]
We may hence conclude that 
\begin{multline*}
    \frac{\sigma^2_t (\lk{t} \testf[t + 1], \lk{t} \{ \term{t + 1} \testf[t + 1] + \testfp[t + 1] - \post{t + 1}(\tstat{t + 1} \af{t + 1} \testf[t + 1] + \testfp[t + 1]) \} )}{(\post{t} \lk{t} \1{\set{X}})^2} \\
    = \sum_{s = 0}^{t - 1} \frac{\post{s} \lk{s} \{\md{s + 1} \BFcent[2]{s + 1}{t + 1} ( \af{t + 1} \testf[t + 1] + \testfp[t + 1] ) \}}{(\post{s} \lk{s} \cdots \lk{t} \1{\set{X}})^2} \\
    + \sum_{s = 0}^{t - 1} \sum_{\ell = 0}^s \K^{\ell - (s + 1)} \frac{\post{\ell} \lk{\ell} \{ \bk{\post{\ell}} ( \tstat{\ell} \af{\ell} + \term{\ell} - \tstat{\ell + 1} \af{\ell + 1} )^2 \lk{\ell + 1} \cdots \lk{s} (\md{s + 1} \{\lk{s + 1} \cdots \lk{t} \testf[t + 1] \}^2 ) \}}{(\post{\ell} \lk{\ell} \cdots \lk{s - 1} \1{\set{X}}) (\post{s} \lk{s} \cdots \lk{t} \1{\set{X}})^2} \eqsp.
\end{multline*}
Finally, we complete the induction step by noting that 
\begin{multline*}
     \frac{\post{t} \lk{t} (\md{t + 1} \{\testf[t + 1] \tstat{t + 1} \af{t + 1} + \testfp[t + 1] - \post{t + 1}(\tstat{t + 1} \af{t + 1} \testf[t + 1] + \testfp[t + 1]) \}^2 )}{(\post{t} \lk{t} \1{\set{X}})^2} \\
    = \frac{\post{t} \lk{t} \{\md{t + 1} \BFcent[2]{t + 1}{t + 1} ( \af{t + 1} \testf[t + 1] + \testfp[t + 1] ) \}}{(\post{t} \lk{t} \1{\set{X}})^2} \eqsp.
\end{multline*}

It remains to check the base case; however, letting, in \eqref{eq:CLT:recursive:variance}, $t = 0$ and $\sigma^2_0 \equiv 0$ (as $\tstat{0} \af{0} = 0$ and $\tstat[i]{0} = 0$ for all $i \in \intvect{1}{\N}$) yields 
\[
    \begin{split}
        \asvar[2]{1}{\testf[1]}{\testfp[1]} &= \frac{\post{0} \lk{0} (\md{1} \{\testf[1] \tstat{1} \af{1} + \testfp[1] - \post{1}(\tstat{1} \af{1} \testf[1] + \testfp[1]) \}^2 )}{(\post{0} \lk{0} \1{\set{X}})^2} + \K^{-1} \frac{\post{0} \lk{0} \{ \bk{\post{0}} ( \term{0} - \tstat{1} \af{1} )^2 \md{1} \testf[1]^2 \}}{(\post{0} \lk{0} \1{\set{X}})^2} \\
        &=  \frac{\post{0} \lk{0} \{\md{1} \BFcent{1}{1}(\af{1} \testf[1] + \testfp[1] )^2 \}}{(\post{0} \lk{0} \1{\set{X}})^2} +  \K^{-1} \frac{\post{0} \lk{0} \{ \bk{\post{0}} ( \term{0} - \tstat{1} \af{1} )^2 \md{1} \testf[1]^2 \}}{(\post{0} \lk{0} \1{\set{X}})^2}
    \end{split}
\]
which is, under the standard convention that $\lk{m} \lk{n} = \operatorname{id}$ if $m > n$, in agreement with \eqref{eq:asvar:closed:form}. This completes the proof.


\begin{lemma} \label{lemma:T2:conv}
    Let \autoref{ass:boundedness:g:q} hold. Then for all $t \in \nset$, $\testf[t] \in \bmf{\alg{X}}$, and $\K \in \nsetpos$,  
    $$
        \sum_{i = 1}^\N \frac{\wgt{t}{i}}{\wgtsum{t}}(\tstat[i]{t})^2 \testf[t](\epart{t}{i}) \convp \post{t}(\tstat{t}^2 \af{t} \testf[t]) + \covterm{t}{\testf[t]} \eqsp, 
    $$
    where
    \begin{equation} \label{eq:def:covterm}
        \covterm{t}{\testf[t]} \eqdef \sum_{\ell = 0}^{t - 1} \K^{\ell - t} \frac{\post{\ell} \lk{\ell} \{ \bk{\post{\ell}} ( \tstat{\ell} \af{\ell} + \term{\ell} - \tstat{\ell + 1} \af{\ell + 1} )^2 \lk{\ell + 1} \cdots \lk{t - 1} \testf[t] \}}{\post{\ell} \lk{\ell} \cdots \lk{t - 1} \1{\set{X}}} \eqsp.
    \end{equation}
\end{lemma}

\begin{proof}
    Again, we proceed by induction. First, the base case $t = 0$ is trivially true since $\tstat{0} \af{0} = 0$, $\tstat[i]{0} = 0$ for all $i \in \intvect{1}{\N}$, and $\sum_{\ell = 0}^{-1} = 0$ by convention. We now assume that the claim of the lemma holds true for some $t \in \nset$. Since \autoref{thm:hoeffding} implies that $\N^{-1} \wgtsum{t + 1} \convp \post{t} \lk{t} \1{\set{X}}$ it is enough to study the convergence of $\N^{-1} \sum_{i = 1}^\N \wgt{t + 1}{i}(\tstat[i]{t + 1})^2 \testf[t + 1](\epart{t + 1}{i})$. For this purpose we will apply \autoref{thm:DM:A-1} to the triangular array 
    $$
        \arr{i} \eqdef \N^{-1} \wgt{t + 1}{i} (\tstat[i]{t + 1})^2 \testf[t + 1](\epart{t + 1}{i}) \quad (i \in \intvect{1}{\N}, \N \in \nsetpos)  
    $$
    furnished with the filtration $\{\partfiltbar{t} \}_{\N \in \nsetpos}$. Note that $|\arr{i}| \leq \supn{\md{t + 1}} \supn{\af{t + 1}}^2 \supn{\testf[t + 1]} / \N$ for all $i \in \intvect{1}{\N}$ and $\N \in \nsetpos$. In addition, using, first, that $\{\arr{i} \}_{i = 1}^{\N}$ are conditionally i.i.d given $\partfiltbar{t}$ and, second, that for all $i \in \intvect{1}{\N}$, the backward indices $\{\bi{t + 1}{i}{j} \}_{j = 1}^{\K}$are conditionally i.i.d. given $\partfiltbar{t}$ and $\epart{t + 1}{i}$, 
    \begin{align}
        \lefteqn{\sum_{i = 1}^{\N} \E \left[\arr{i} \mid \partfiltbar{t} \right]} \nonumber \\
        &= \E \left[ \wgt{t + 1}{1} (\tstat[1]{t + 1})^2 \testf[t + 1](\epart{t + 1}{1}) \mid \partfiltbar{t} \right] \nonumber \\
        &= \K^{-1} \E \left[ \wgt{t + 1}{1} \testf[t + 1](\epart{t + 1}{1}) \E \left[ \left( \tstat[\bi{t + 1}{1}{1}]{t} + \term{t}(\epart{t}{\bi{t + 1}{1}{1}}, \epart{t + 1}{1}) \right)^2 \mid \partfiltbar{t} \vee \partfilt{t + 1} \right] \mid \partfiltbar{t} \right] \label{eq:T2-decomp-1} \\
        &\quad + \K^{-1} (\K - 1) \E \left[ \wgt{t + 1}{1} \testf[t + 1](\epart{t + 1}{1}) \E^2 \left[ \tstat[\bi{t + 1}{1}{1}]{t} + \addf{t}(\epart{t}{\bi{t + 1}{1}{1}}, \epart{t + 1}{1}) \mid \partfiltbar{t} \vee \partfilt{t + 1} \right] \mid \partfiltbar{t} \right] \eqsp. \label{eq:T2-decomp-2}
    \end{align}
    We treat separately the two terms \eqref{eq:T2-decomp-1} and \eqref{eq:T2-decomp-2}. First,  
    \[
    \begin{split}
        \eqref{eq:T2-decomp-1} &=  \K^{-1} \E \left[ \wgt{t + 1}{1} \testf[t + 1](\epart{t + 1}{1}) \sum_{\ell = 1}^\N \frac{\wgt{t}{\ell} \hd(\epart{t}{\ell}, \epart{t + 1}{1})}{\sum_{\ell' = 1}^\N \wgt{t}{\ell'} \hd(\epart{t}{\ell'}, \epart{t + 1}{1})} \{\tstat[\ell]{t} + \term{t}(\epart{t}{\ell}, \epart{t + 1}{1}) \}^2 \mid \partfiltbar{t} \right] \\
    	&= \K^{-1}  \sum_{i = 1}^\N \frac{\wgt{t}{i}}{\wgtsum{t}} \int \hd(\epart{t}{i}, x) \md{t + 1}(x) \testf[t + 1](x) \sum_{\ell = 1}^\N \frac{\wgt{t}{\ell} \hd(\epart{t}{\ell}, x)}{\sum_{\ell' = 1}^\N \wgt{t}{\ell'} \hd(\epart{t}{\ell'}, x)} \{\tstat[\ell]{t} + \term{t}(\epart{t}{\ell}, x) \}^2  \, \refM(\rmd x) \\
    	&= \K^{-1} \sum_{\ell = 1}^\N \frac{\wgt{t}{\ell}}{\wgtsum{t}} \int \hd(\epart{t}{\ell}, x) \, \md{t + 1}(x) \testf[t + 1](x) \{\tstat[\ell]{t} + \term{t}(\epart{t}{\ell}, x) \}^2 \, \refM(\rmd x).
    \end{split}
    \]
    Using \autoref{thm:hoeffding} and the induction hypothesis we obtain the limits 
        \[
        \begin{split}
        	&\sum_{\ell = 1}^\N \frac{\wgt{t}{\ell}}{\wgtsum{t}} (\tstat[\ell]{t})^2 \lk{t}(\epart{t}{\ell}, \testf[t + 1]) \convp \post{t}(\tstat{t}^2 \af{t} \lk{t} \testf[t + 1]) + \covterm{t}{\lk{t} \testf[t + 1]} \eqsp, \\
            	&\sum_{\ell = 1}^\N \frac{\wgt{t}{\ell}}{\wgtsum{t}} \tstat[\ell]{t} \lk{t}(\epart{t}{\ell}, \term{t} \testf[t + 1]) \convp \post{t}\{\tstat{t} \af{t} \lk{t}(\term{t} \testf[t + 1]) \} \eqsp, \\
    	&\sum_{\ell = 1}^\N \frac{\wgt{t}{\ell}}{\wgtsum{t}} \lk{t}(\epart{t}{\ell}, \term{t}^2 \testf[t + 1]) \convp \post{t} \lk{t}(\term{t}^2 \testf[t + 1]) \eqsp,
        \end{split}
    \]
    which yield
    \begin{multline*}
        \eqref{eq:T2-decomp-1} \convp \K^{-1} \left( \post{t}(\tstat{t}^2 \af{t} \lk{t} \testf[t + 1]) + \covterm{t}{\lk{t} \testf[t + 1]} + 2 \post{t}\{\tstat{t} \af{t} \lk{t}(\term{t} \testf[t + 1]) \} + \post{t} \lk{t}(\term{t}^2 \testf[t + 1])\right) \\
        = \K^{-1} \left( \post{t} \lk{t}\{(\tstat{t} \af{t} + \term{t})^2 \testf[t + 1] \}+ \covterm{t}{\lk{t} \testf[t + 1]}\right) \eqsp. 
    \end{multline*}
    We turn to the second term \eqref{eq:T2-decomp-2} and equate the same with $\K^{-1}(\K - 1) \post[part]{t} \hk \lebfun[\N]$, where 
    $$
        \lebfun[\N](x) \eqdef \md{t + 1}(x) \testf[t + 1](x) \left( \sum_{\ell = 1}^\N \frac{\wgt{t}{\ell} \hd(\epart{t}{\ell}, x)}{\sum_{\ell' = 1}^\N \wgt{t}{\ell'} \hd(\epart{t}{\ell'}, x)} \{ \tstat[\ell]{t} + \term{t}(\epart{t}{\ell}, x) \} \right)^2 \quad (x \in \set{X}) \eqsp. 
    $$
    Since, $\supn{\lebfun[\N]} \leq \supn{\md{t + 1}} \supn{\testf[t + 1]} \supn{\af{t + 1}}$ for all $\N \in \nset$, and, by \autoref{thm:hoeffding}, for all $x \in \set{X}$, $\prob$-a.s., 
    \[
        \begin{split}
            \lebfun[\N](x) \rightarrow& \hphantom{\ } \md{t + 1}(x) \testf[t + 1](x) \left( \frac{\int \hd(\tilde{x}, x) \{\tstat{t} \af{t} (\tilde{x}) + \term{t}(\tilde{x}, x) \} \, \post{t}(\rmd \tilde{x})}{\int \hd(\tilde{x}, x) \, \post{t}(\rmd \tilde{x})} \right)^2 \\
            &= \md{t + 1}(x) \testf[t + 1](x) \bk[2]{\post{t}}(\tstat{t} \af{t} + \term{t})(x) \\
            &= \md{t + 1}(x) \testf[t + 1](x) \tstat{t + 1}^2 \af{t + 1}(x) \eqsp, 
        \end{split}
    \]
    we may, under \autoref{ass:boundedness:g:q}, apply \autoref{lemma:generalized:lebesgue}, yielding
    $$
        \post[part]{t} \hk \lebfun[\N] \convp \post{t} \hk (\md{t + 1} \testf[t + 1] \tstat{t + 1}^2 \af{t + 1}) = \post{t} \lk{t}(\tstat{t + 1}^2 \af{t + 1} \testf[t + 1]) \eqsp.
    $$
    Consequently,  
    \begin{multline} \label{cond:exp:limit-1}
        \sum_{i = 1}^{\N} \E \left[\arr{i} \mid \partfiltbar{t} \right] \convp \post{t} \lk{t}(\tstat{t + 1}^2 \af{t + 1} \testf[t + 1] ) \\
        + \K^{-1} \left( \post{t} \lk{t}\{(\tstat{t} \af{t} + \term{t})^2 \testf[t + 1] \}- \post{t} \lk{t}(\tstat{t + 1}^2 \af{t + 1} \testf[t + 1] ) + \covterm{t}{\lk{t} \testf[t + 1]}\right) \eqsp. 
    \end{multline}
    In order to show that $\sum_{i = 1}^{\N} \arr{i}$ has the same limit \eqref{cond:exp:limit-1} in probability we use \autoref{thm:DM:A-1}. Condition \A{ass:DM:cons:tightness} is easily checked by reusing \eqref{cond:exp:limit-1} with $\testf[t + 1]$ replaced by $|\testf[t + 1]|$. In order to check \A{ass:DM:cons:Lindeberg} we simply note that for all $\varepsilon \in \rsetpos$,
    $$
        \sum_{i = 1}^{\M[\N]} \cexp{|\arr{i}| \1{\{|\arr{i}| \geq \varepsilon \}}}{\partfiltbar{t}} \leq \supn{\md{t + 1}} \supn{\af{t + 1}}^2 \supn{\testf[t + 1]} \1{ \{\supn{\md{t + 1}} \supn{\af{t + 1}}^2 \supn{\testf[t + 1]} \geq \varepsilon \N \}} \eqsp,
    $$
    where the right hand side is zero for $\N$ large enough. Thus, \autoref{thm:DM:A-1} applies and since, by reversibility, 
    \[
        \begin{split}
            \lefteqn{\post{t} \lk{t}\{(\tstat{t} \af{t} + \term{t})^2 \testf[t + 1] \}- \post{t} \lk{t}(\tstat{t + 1}^2 \af{t + 1} \testf[t + 1])}\hspace{10mm} \\
            &= \post{t} \varotimes \hk ( \{(\tstat{t} \af{t} + \term{t})^2- \tstat{t + 1}^2 \af{t + 1} \}\md{t + 1} \testf[t + 1]) \\
            &= \post{t} \hk \varotimes\bk{\post{t}} \{(\tstat{t} \af{t} + \term{t}- \tstat{t + 1} \af{t + 1} )^2 \md{t + 1} \testf[t + 1] \} \\
            &= \post{t} \lk{t} \{\bk{\post{t}} (\tstat{t} \af{t} + \term{t}- \tstat{t + 1} \af{t + 1} )^2 \testf[t + 1] \} \eqsp,
        \end{split}
    \]
    Slutsky's lemma implies 
    \begin{multline*} 
        \sum_{i = 1}^\N \frac{\wgt{t + 1}{i}}{\wgtsum{t + 1}}(\tstat[i]{t + 1})^2 \testf[t + 1](\epart{t + 1}{i}) = \N \wgtsum{t + 1}^{-1} \sum_{i = 1}^{\N} \arr{i} \convp \post{t + 1}(\tstat{t + 1}^2 \af{t + 1} \testf[t + 1]) \\
        + \frac{1}{\K (\post{t} \lk{t} \1{\set{X}})} \left( \post{t} \lk{t} \{\bk{\post{t}} (\tstat{t} \af{t} + \term{t}- \tstat{t + 1} \af{t + 1} )^2 \testf[t + 1] \} + \covterm{t}{\lk{t} \testf[t + 1]}\right) \eqsp.
    \end{multline*}
    We may now conclude the proof by noting, using the induction hypothesis, the identity 
    $$
        (\post{\ell} \lk{\ell} \cdots \lk{t - 1} \1{\set{X}})(\post{t} \lk{t} \1{\set{X}}) = \post{\ell} \lk{\ell} \cdots \lk{t} \1{\set{X}} \eqsp,
    $$
    and the convention $\lk{t + 1} \lk{t} = \operatorname{id}$, that
    \begin{multline*}
        \frac{1}{\K (\post{t} \lk{t} \1{\set{X}})} \left( \post{t} \lk{t} \{\bk{\post{t}} (\tstat{t} \af{t} + \term{t}- \tstat{t + 1} \af{t + 1} )^2 \testf[t + 1] \} + \covterm{t}{\lk{t} \testf[t + 1]}\right) \\
        = \sum_{\ell = 0}^t \K^{\ell - (t + 1)} \frac{\post{\ell} \lk{\ell} \{ \bk{\post{\ell}} ( \tstat{\ell} \af{\ell} + \term{\ell} - \tstat{\ell + 1} \af{\ell + 1} )^2 \lk{\ell + 1} \cdots \lk{t} \testf[t + 1] \}}{\post{\ell} \lk{\ell} \cdots \lk{t} \1{\set{X}}} \eqsp.
    \end{multline*}
\end{proof}

The following lemma formalizes an argument used in the proof of \cite[Theorem~8]{douc:garivier:moulines:olsson:2010}. 


\begin{lemma} \label{lemma:generalized:lebesgue}
    Let \autoref{ass:boundedness:g:q} hold. Let $\lebker$ be a possibly unnormalized transition kernel on $(\set{X}, \alg{X})$ having transition density $\lebden \in \bmf{\alg{X}^2}$ with respect to some reference measure $\lebref$. Moreover, let $\{\lebfun[\N] \}_{\N \in \nsetpos}$ be a sequence of functions in $\bmf{\alg{X}}$ for which 
    \begin{enumerate}[(i)]
        \item there exists $\lebfun \in \bmf{\alg{X}}$ such that for all $x \in \set{X}$, $\lebfun[\N](x) \rightarrow \lebfun(x)$, $\prob$-a.s., and
        \item there exists $\lebfunbd \in \rsetpos$ such that $\supn{\lebfun[\N]} \leq \lebfunbd$ for all $\N \in \nsetpos$. 
    \end{enumerate}
    Then for all $t \in \nset$, $\post[part]{t} \lebker \lebfun[\N] \convp \post{t} \lebker \lebfun$.
\end{lemma}

\begin{proof}
    Since, by \autoref{thm:hoeffding}, $\post[part]{t} \lebker \lebfun \convp \post{t} \lebker \lebfun$, it is enough to establish that 
    \begin{equation*}
        \post[part]{t} \lebker \lebfun[\N] \convp \post[part]{t} \lebker \lebfun \eqsp. 
    \end{equation*}
    For this purpose, set 
    \[
        \begin{split}
            \leba(x) &\eqdef | \lebfun[\N](x) - \lebfun(x) |\int \lebden(\tilde{x}, x) \, \post[part]{t}(\rmd \tilde{x}) \eqsp, \\
            \lebb(x) &\eqdef \int \lebden(\tilde{x}, x) \, \post[part]{t}(\rmd \tilde{x})
        \end{split}
        (\N \in \nsetpos, x \in \set{X}) \eqsp.
    \]
    Since $| \post[part]{t} \lebker \lebfun[\N] - \post[part]{t} \lebker \lebfun | \leq \lebref \leba$ it is, by Markov's inequality, enough to show that $\E[\lebref \leba]$ tends to zero as $\N$ tends to infinity. However, by Fubini's theorem, 
    $$
        \lim_{\N \rightarrow \infty} \E[\lebref \leba] = \lim_{\N \rightarrow \infty} \int \E[ \leba(x)] \, \lebref(\rmd x)= 0 \eqsp,
    $$ 
    where the last equality is a consequence of the generalized Lebesgue dominated convergence theorem provided that 
    \begin{enumerate}[(i)]
        \item  $\lim_{\N \rightarrow \infty} \E[ \leba(x)] = 0$ for all $x \in \set{X}$,
        \item  there exists $c \in \rsetpos$ such that $\E[ \leba(x) ] \leq c \, \E[ \lebb(x) ]$ for all $x \in \set{X}$, 
        \item  $\lim_{\N \rightarrow \infty} \int \E[ \lebb(x)] \, \lebref(\rmd x) = \int \lim_{\N \rightarrow \infty} \E[ \lebb(x)] \, \lebref(\rmd x)$. 
    \end{enumerate}
    Here (i) is implied by \autoref{thm:hoeffding} and, as $|\leba(x)| \leq \supn{\lebden} (\supn{\lebfun} + \lebfunbd)$ for all $x \in \set{X}$, the standard dominated convergence theorem. Moreover, (ii) is satisfied with $c = \supn{\lebfun} + \lebfunbd$. Finally, to check (iii), notice that
    \begin{multline*}
        \lim_{\N \to \infty} \int \E[ \lebb(x) ] \, \lebref(\rmd x) \overset{(a)}{=} \lim_{\N \to \infty} \E [ \post[part]{t} \lebker \1{\set{X}} ] \overset{(b)}{=} \post{t} \lebker \1{\set{X}} \overset{(c)}{=} \iint \lebden(\tilde{x}, x) \, \post{t}(\rmd \tilde{x}) \, \lebref(\rmd x) \\\overset{(d)}{=} \int \lim_{\N \to \infty} \E [ \lebb(x) ] \, \lebref(\rmd x) \eqsp,
    \end{multline*}
    where (a) and (c) follow by Fubini's theorem and (b) and (d) are obtained from~\autoref{thm:hoeffding} and the standard dominated convergence theorem (as $\lebker \1{\set{X}} \in \bmf{\alg{X}}$ and $\lebden \in \bmf{\alg{X}^2}$ by assumption). This completes the proof. 
\end{proof}


\subsection{Proof of \autoref{prop:bound:normalized:FFBSm}}

By~\cite[Lemma 1]{dubarry:lecorff:2013}, 
\begin{equation} \label{eq:cyrille:bd}
    \supn{\BFcent{s + 1}{t} \af{t}} \leq \supn{\lk{s+1} \ldots \lk{t} \1{\set{X}}} \sum_{\ell = 0}^{t - 1} \mr^{\max \{s - \ell + 2, \ell - s - 3, 0 \}} \oscn{\term{\ell}} \eqsp,
\end{equation}
and, consequently, 
$$
    \sum_{s = 0}^{t - 1} \frac{\post{s} \lk{s} ( \md{s + 1} \BFcent[2]{s + 1}{t} \af{t} )}{(\post{s} \lk{s} \cdots \lk{t - 1} \1{\set{X}})^2} \leq \hbd^2 \sum_{s = 0}^{t - 1} \frac{( \post{s} \lk{s} \md{s + 1} ) \supn{\lk{s+1} \ldots \lk{t} \1{\set{X}}}^2 }{(\post{s} \lk{s} \cdots \lk{t - 1} \1{\set{X}})^2} \left( \sum_{\ell = 0}^{t - 1} \mr^{\max \{s - \ell + 2, \ell - s - 3, 0 \}} \right)^2 \eqsp.
$$
Now, under \autoref{ass:strong:mixing}, for all $x \in \set{X}$,
\begin{equation} \label{eq:L:prod:up:low:bd}
    \hklow \refM( \md{s + 2} \lk{s + 2} \cdots \lk{t} \1{\set{X}} ) \leq \lk{s + 1} \cdots \lk{t} \1{\set{X}}(x) \leq \hkup \refM( \md{s + 2} \lk{s + 2} \ldots \lk{t} \1{\set{X}} )  
\end{equation}
implying that 
$$
    \frac{( \post{s} \lk{s} \md{s + 1} ) \supn{\lk{s + 1} \cdots \lk{t} \1{\set{X}}}^2 }{(\post{s} \lk{s} \cdots \lk{t - 1} \1{\set{X}})^2} = \frac{(\post{s + 1} \md{s + 1})  \supn{\lk{s + 1} \cdots \lk{t} \1{\set{X}}}^2}{(\post{s} \lk{s} \1{\set{X}})(\post{s + 1} \lk{s + 1} \cdots \lk{t - 1} \1{\set{X}})^2} \leq \frac{\mdup}{\mdlow} \left( \frac{\hkup}{\hklow} \right)^2 = \frac{\mdup}{\mdlow (1 - \mr)^2} \eqsp.
$$
Moreover, as 
\begin{align*}
    \sum_{s = 0}^{t - 1} \left( \sum_{\ell = 0}^{t - 1} \mr^{\max \{s - \ell + 2, \ell - s - 3, 0 \}} \right)^2 = \sum_{s = 0}^{t - 4} \left( \frac{2 - \mr^{s + 3} - \mr^{t - s - 3}}{1 - \mr} \right)^2 \\
    + \sum_{s = t-3}^{t-1} \left( \rho^{s+2} \frac{ 1 - \rho^{-t} }{1 - \rho^{-1}} \right) = \frac{4 t}{(1 - \mr)^2} + o(t) \eqsp,
\end{align*}
we conclude that 
$$
    \limsup_{t \rightarrow \infty} \frac{1}{t} \sum_{s = 0}^{t - 1} \frac{\post{s} \lk{s} ( \md{s + 1} \BFcent[2]{s + 1}{t} \af{t} )}{(\post{s} \lk{s} \cdots \lk{t - 1} \1{\set{X}})^2} \leq \hbd^2 \frac{4 \mdup}{\mdlow (1 - \mr)^4} \eqsp. 
$$


\subsection{Proof of \autoref{thm:linear:bound:normalized:PaRIS}}

The first term of $\asvarstd[2]{t}$ is the asymptotic variance of the FFBSm algorithm, which is,  by \autoref{prop:bound:normalized:FFBSm}, bounded by $4 \hbd^2 \mdup / \{\mdlow (1 - \mr)^4 \}$. To treat the second term, we bound, using \eqref{eq:L:prod:up:low:bd}, 
    \begin{align}
        \frac{\lk{\ell + 1} \cdots \lk{s} (\md{s + 1} \{\lk{s + 1} \cdots \lk{t - 1} \1{\set{X}} \}^2 ) }{(\post{\ell + 1} \lk{\ell + 1} \cdots \lk{s - 1} \1{\set{X}}) (\post{s} \lk{s} \cdots \lk{t - 1} \1{\set{X}})^2} &\leq \frac{\supn{\lk{\ell + 1} \cdots \lk{s - 1} \1{\set{X}}} \mdup \supn{\lk{s + 1} \cdots \lk{t - 1} \1{\set{X}} }^2}{(\post{\ell + 1} \lk{\ell + 1} \cdots \lk{s - 1} \1{\set{X}}) (\post{s} \lk{s} \1{\set{X}}) (\post{s + 1} \lk{s + 1} \cdots \lk{t - 1} \1{\set{X}})^2} \nonumber \\
        &\leq \frac{\mdup}{\mdlow (1 - \mr)^3} \label{eq:cov:term:bd:add:func} \eqsp. 
    \end{align}
Moreover, since $\tstat{\ell + 1} \af{\ell + 1} = \bk{\post{\ell}}(\tstat{\ell} \af{\ell} + \term{\ell})$ and $\tstat{\ell} \af{\ell} = \BF{\ell}{\ell} \af{\ell}$, we obtain, by reusing \eqref{eq:cyrille:bd},
\begin{equation*}
    \supn{\tstat{\ell} \af{\ell} + \term{\ell} - \tstat{\ell + 1} \af{\ell + 1}} \leq \oscn{\tstat{\ell} \af{\ell}} + \oscn{\term{\ell}} \leq 4 \supn{ \BFcent{\ell}{\ell} \af{\ell}} + \oscn{\term{\ell}} \leq \hbd \left( \frac{4 \mdup \mr^2}{1 - \mr} + 1 \right) \eqsp. 
\end{equation*}
Thus, 
\begin{multline*}
    \frac{1}{t} \sum_{s = 0}^{t - 1} \sum_{\ell = 0}^s \K^{\ell - (s + 1)} \frac{\post{\ell + 1}\{ \bk{\post{\ell}} ( \tstat{\ell} \af{\ell} + \term{\ell} - \tstat{\ell + 1} \af{\ell + 1} )^2 \lk{\ell + 1} \cdots \lk{s} (\md{s + 1} \{\lk{s + 1} \cdots \lk{t - 1} \testf[t] \}^2 ) \}}{(\post{\ell + 1} \lk{\ell + 1} \cdots \lk{s - 1} \1{\set{X}}) (\post{s} \lk{s} \cdots \lk{t - 1} \1{\set{X}})^2} \\\leq \hbd^2 \frac{\mdup}{\mdlow (1 - \mr)^3} \left( \frac{4 \mdup \mr^2}{1 - \mr} + 1 \right)^2 \frac{1}{t} \sum_{s = 0}^{t - 1} \sum_{\ell = 0}^s \K^{\ell - (s + 1)} \eqsp,
\end{multline*}
and since 
$$
    \lim_{s \rightarrow \infty} \sum_{\ell = 0}^s \K^{\ell - (s + 1)} = (\K - 1)^{-1} 
$$ 
we may conclude the proof by taking the Ces\`{a}ro mean.

\subsection{Proof of \autoref{thm:stability:marginal:smoothing:PaRIS}}

In the case of marginal smoothing, \cite[Theorem~12]{douc:garivier:moulines:olsson:2010} provides, for $t \geq \sptime$ (since the variance vanishes for $t < \sptime$, the result holds trivially true in this case), the time uniform bound 
$$
    \asvarFFBSmstd[2]{t} \leq \oscn[2]{\term{\sptime}} \frac{\mdup^2 (1 + \mr^2)}{(1 + \mr)(1 - \mr)^3} 
$$
and hence, since all terms are zero except $\term{\sptime}$, it is enough to bound the quantity 
\begin{equation} \label{eq:key:reminder:marginal:smoothing}
    \sum_{s = \sptime}^{t - 1} \sum_{\ell = \sptime}^s \K^{\ell - (s + 1)} \frac{\post{\ell + 1}\{ \bk{\post{\ell}} ( \tstat{\ell} \af{\ell} + \term{\ell} - \tstat{\ell + 1} \af{\ell + 1} )^2 \lk{\ell + 1} \cdots \lk{s} (\md{s + 1} \{\lk{s + 1} \cdots \lk{t - 1} \testf[t] \}^2 ) \}}{(\post{\ell + 1} \lk{\ell + 1} \cdots \lk{s - 1} \1{\set{X}}) (\post{s} \lk{s} \cdots \lk{t - 1} \1{\set{X}})^2} 
\end{equation}
(where $\term{\ell} = 0$ for $\ell > \sptime$). In addition, by \cite[Lemma~10]{douc:garivier:moulines:olsson:2010}, for all $\ell \geq \sptime$, 
$$
    \supn{\BFcent{\ell}{\ell} \af{\ell}} \leq \mr^{\ell - \sptime} \oscn{\term{\sptime}} \eqsp, 
$$
yielding
$$
    \supn{\tstat{\ell} \af{\ell} + \term{\ell} - \tstat{\ell + 1} \af{\ell + 1}} \leq 2 \mr^{\ell - \sptime} \oscn{\term{\sptime}} \eqsp. 
$$
By combining this with \eqref{eq:cov:term:bd:add:func} we obtain, via standard operations on geometric sums, 
\[
    \begin{split}
        \eqref{eq:key:reminder:marginal:smoothing} &\leq 4 \oscn[2]{\term{\sptime}} \frac{\mdup}{\mdlow (1 - \mr)^3} \sum_{s = \sptime}^{t - 1} \sum_{\ell = \sptime}^s \K^{\ell - (s + 1)} \mr^{2(\ell - \sptime)} \\
        &= 4 \oscn[2]{\term{\sptime}} \frac{\mdup}{\mdlow (1 - \mr)^3} \times 
        \begin{cases} 
            \displaystyle \frac{1}{(1 - \K \mr^2)} \left( \frac{1 - \K^{- (t - \sptime)}}{\K - 1} - \mr^2 \frac{1 - \mr^{2(t - \sptime)}}{1 - \mr^2} \right) & \mbox{if } \K \mr^2 \neq 1 \eqsp, \\ 
            \displaystyle \frac{1}{\K - 1} \left( 2 - \K^{- (t - \sptime)} - (t - \sptime) \K^{- (t - \sptime) + 1} - \K^{- (t - \sptime)} \right) & \mbox{if }\K \mr^2 = 1 \eqsp,
        \end{cases}
    \end{split}
\]
and hence, letting $t$ tend to infinity, 
$$
    \eqref{eq:key:reminder:marginal:smoothing} \leq 4 \oscn[2]{\term{\sptime}} \frac{\mdup}{\mdlow (1 - \mr)^3} \times 
    \begin{cases} 
        \displaystyle \frac{1}{(\K - 1) (1 - \mr^2)} & \mbox{if } \K \mr^2 \neq 1 \eqsp, \\ 
        \displaystyle \frac{2}{\K - 1} & \mbox{if }\K \mr^2 = 1 \eqsp,
    \end{cases}
$$
which concludes the proof.

%% file: ow2014_lemmas.tex

\subsection{Conditional limit theorems for triangular arrays of dependent random variables}

We first recall two results, obtained in \cite{douc:moulines:2008} (but reformulated slightly here for our purposes), which are essential for the developments of the present paper. 

\begin{theorem}[\cite{douc:moulines:2008}] \label{thm:DM:A-1}
    Let $(\Omega, \mathcal{A}, \{\genfd[\N] \}_{\N \in \nsetpos}, \prob)$ be a filtered probability space. In addition, let $\{\arr{i} \}_{i = 1}^{\M[\N]}$, $\N \in \nsetpos$, be a triangular array of random variables on $(\Omega, \mathcal{A}, \prob)$ such that for all $\N \in \nsetpos$, the variables $\{\arr{i} \}_{i = 1}^{\M[\N]}$ are conditionally independent given $\genfd[\N]$ with $\cexp[txt]{|\arr{i}|}{\genfd[\N]} < \infty$, $\prob\mbox{-a.s.}$, for all $i \in \intvect{1}{\M[\N]}$. Moreover, assume that    
    \begin{hypA} \label{ass:DM:cons:tightness}
        $ 
        		\displaystyle \quad \lim_{\lambda \rightarrow \infty} \sup_{\N \in \nsetpos} \prob \left( \sum_{i = 1}^{\M[\N]} \cexp[txt]{|\arr{i}|}{\genfd[\N]} \geq \lambda \right) = 0 \eqsp.
        $
    \end{hypA}
    \begin{hypA} \label{ass:DM:cons:Lindeberg}
        For all $\varepsilon > 0$, as $\N \rightarrow \infty$,
        $$
        		\sum_{i = 1}^{\M[\N]} \cexp{|\arr{i}| \1{\{|\arr{i}| \geq \varepsilon \}}}{\genfd[\N]} \convp 0 \eqsp. 
        $$
    \end{hypA}
    Then, as $\N \rightarrow \infty$, 
    $$ 
        \max_{m \in \intvect{1}{\M[\N]}} \left|\sum_{i = 1}^m \arr{i} - \sum_{i = 1}^m \cexp{\arr{i}}{\genfd[\N]} \right| \convp 0 \eqsp.  
    $$
\end{theorem}

\begin{theorem}[{\cite{douc:moulines:2008}}] \label{thm:DM:A-3}
    Let the assumptions of Theorem~\ref{thm:DM:A-1} hold with $\cexp[txt]{\arr[2]{i}}{\genfd[\N]} < \infty$, $\prob\mbox{-a.s.}$, for all $i \in \intvect{1}{\M[\N]}$, and
    \A{ass:DM:cons:tightness} and \A{ass:DM:cons:Lindeberg} replaced by
    \begin{hypB} \label{ass:DM:norm:asymptotic:variance}
        For some constant $\varsigma^2 > 0$, as $\N \rightarrow \infty$,
        $$ 
        		\sum_{i = 1}^{\N} \left( \cexp[txt]{\arr[2]{i}}{\genfd[\N]} - \E^2 \left[ \arr{i} \mid \genfd[\N] \right] \right) \convp \varsigma^2 \eqsp.
        $$
    \end{hypB}
    \begin{hypB} \label{ass:DM:norm:Lindeberg}
        For all $\varepsilon > 0$, as $\N \rightarrow \infty$,
        $$
        		\sum_{i = 1}^{\M[\N]} \cexp{\arr[2]{i} \1{\{|\arr{i}| \geq \varepsilon \}}}{\genfd[\N] } \convp 0 \eqsp.
        $$
    \end{hypB}
    Then, for all $u \in \rset$, as $\N \rightarrow \infty$, 
    $$ 
        \cexp{ \exp \left( \operatorname{i} u \sum_{i = 1}^{\M[\N]} \left\{ \arr{i} - \cexp[txt]{\arr{i}}{\genfd[\N]} \right\} \right) }{\genfd[\N]}  \convp \exp \left( -u^2 \varsigma^2 /2 \right) \eqsp.  
    $$
\end{theorem}

\subsection{An accept-reject-based algorithm for backward sampling}
\label{sec:accept-reject}

Given two subsequent particle samples $\{ (\epart{s - 1}{i}, \wgt{s - 1}{i})\}_{i = 1}^\N$Êand $\{ (\epart{s}{i}, \wgt{s}{i})\}_{i = 1}^\N$, the following algorithm, which is a trivial adjustment of \cite[Algorithm~1]{douc:garivier:moulines:olsson:2010}, simulates the full set $\{ \bi{s}{i}{j} : (i, \k) \in \intvect{1}{\N} \times \intvect{1}{\K} \}$ of backward indices required for one iteration of PaRIS. The algorithm requires \autoref{ass:boundedness:g:q}(ii) to hold true. 

\begin{algorithm}[htb]
    \caption{Accept-reject-based backward sampling}
    \label{alg:accept:reject}
    \begin{algorithmic}[1]
        \Require Particle samples $\{ (\epart{s - 1}{i}, \wgt{s - 1}{i})\}_{i = 1}^\N$Êand $\{ (\epart{s}{i}, \wgt{s}{i})\}_{i = 1}^\N$.
        \For{$\k = 1 \to \K$}
            \State set $\set{L} \gets \intvect{1}{\N}$;
            \While {$\set{L} \neq \varnothing$}
                \State set $n \gets \# \set{L}$;
                \State draw $(I_1, \ldots, I_\N) \sim \probdist( \{ \wgt{s - 1}{i} \}_{i = 1}^\N )^{\varotimes \N}$;
                \State draw $(U_1, \ldots, U_\N) \sim \unif(0,1)^{\varotimes \N}$; 
                \State set $\set{L}_n \gets \varnothing$;
                \For{$k = 1 \to n$}
                    \If{$U_k \leq \hd(\epart{s-1}{I_k},\epart{s}{\set{L}(k)})/\hkup$}
                        \State set $\bi{s}{\set{L}(k)}{\k} \gets I_k$;
                    \Else
                        \State set $\set{L}_n \gets \set{L}_n \cup \{ \set{L}(k) \}$;
                    \EndIf
                \EndFor
                \State set $\set{L} \gets \set{L}_n$;
            \EndWhile
        \EndFor
        \State \Return $\{ \bi{s}{i}{j} : (i, \k) \in \intvect{1}{\N} \times \intvect{1}{\K} \}$
    \end{algorithmic}
\end{algorithm}

%% file: ow2014.bbl
\begin{thebibliography}{10}

\bibitem{baum:petrie:soules:weiss:1970}
L.~E. Baum, T.~P. Petrie, G.~Soules, and N.~Weiss.
\newblock A maximization technique occurring in the statistical analysis of
  probabilistic functions of {M}arkov chains.
\newblock {\em Ann. Math. Statist.}, 41(1):164--171, 1970.

\bibitem{cappe:2001:hmmbib}
O.~Capp\'{e}.
\newblock Ten years of {HMM}s (online bibliography 1989--2000), March 2001.

\bibitem{cappe:2009}
O.~Capp\'{e}.
\newblock Online {EM} algorithm for hidden {M}arkov models.
\newblock {\em Journal of Computational and Graphical Statistics},
  20(3):728--749, 2011.

\bibitem{cappe:godsill:moulines:2007}
O.~Capp\'{e}, S.~J. Godsill, and E.~Moulines.
\newblock An overview of existing methods and recent advances in sequential
  {M}onte {C}arlo.
\newblock {\em IEEE Proceedings}, 95(5):899--924, 2007.

\bibitem{cappe:moulines:ryden:2005}
O.~Capp\'{e}, E.~Moulines, and T.~Ryd\'{e}n.
\newblock {\em Inference in Hidden {M}arkov Models}.
\newblock Springer, 2005.

\bibitem{chib:02}
S.~Chib, F.~Nadari, and N.~Shephard.
\newblock {M}arkov chain {M}onte {C}arlo methods for stochastic volatility
  models.
\newblock {\em J. Econometrics}, 108:281--316, 2002.

\bibitem{crisan:heine:2008}
D.~Crisan and K.~Heine.
\newblock Stability of the discrete time filter in terms of the tails of noise
  distributions.
\newblock {\em J. Lond. Math. Soc. (2)}, 78(2):441--458, 2008.

\bibitem{delmoral:2004}
P.~{Del Moral}.
\newblock {\em {F}eynman-Kac {F}ormulae. {G}enealogical and Interacting
  Particle Systems with Applications}.
\newblock Springer, 2004.

\bibitem{delmoral:doucet:singh:2010}
{P}. {D}el {M}oral, {A}. {D}oucet, and {S}. {S}ingh.
\newblock {A} {B}ackward {P}article {I}nterpretation of {F}eynman-{K}ac
  {F}ormulae.
\newblock {\em ESAIM M2AN}, 44(5):947--975, 2010.

\bibitem{delmoral:guionnet:2001}
P.~{Del Moral} and A.~Guionnet.
\newblock On the stability of interacting processes with applications to
  filtering and genetic algorithms.
\newblock {\em Annales de l'Institut Henri Poincar\'e}, 37:155--194, 2001.

\bibitem{douc:garivier:moulines:olsson:2010}
R.~Douc, A.~Garivier, E.~Moulines, and J.~Olsson.
\newblock Sequential {M}onte {C}arlo smoothing for general state space hidden
  {M}arkov models.
\newblock {\em Ann. Appl. Probab.}, 21(6):2109--2145, 2011.

\bibitem{douc:moulines:2008}
R.~Douc and E.~Moulines.
\newblock Limit theorems for weighted samples with applications to sequential
  {M}onte {C}arlo methods.
\newblock {\em Ann. Statist.}, 36(5):2344--2376, 2008.

\bibitem{douc:moulines:olsson:2014}
R.~Douc, E.~Moulines, and J.~Olsson.
\newblock Long-term stability of sequential {M}onte {C}arlo methods under
  verifiable conditions.
\newblock {\em Ann. Appl. Probab.}, 24(5):1767--1802, 2014.

\bibitem{douc:moulines:stoffer:2014}
R.~Douc, E.~Moulines, and D.~Stoffer.
\newblock {\em Nonlinear Time Series: Theory, Methods and Applications with R
  Examples}.
\newblock Chapman \& Hall/CRC Texts in Statistical Science, 2014.

\bibitem{doucet:defreitas:gordon:2001}
A.~Doucet, N.~{De Freitas}, and N.~Gordon, editors.
\newblock {\em Sequential {M}onte {C}arlo Methods in Practice}.
\newblock Springer, New York, 2001.

\bibitem{doucet:godsill:andrieu:2000}
A.~Doucet, S.~Godsill, and C.~Andrieu.
\newblock On sequential {M}onte-{C}arlo sampling methods for {B}ayesian
  filtering.
\newblock {\em Stat. Comput.}, 10:197--208, 2000.

\bibitem{doucet:johansen:2009}
A.~Doucet and A.~M. Johansen.
\newblock A tutorial on particle filtering and smoothing: fifteen years later.
\newblock {\em Oxford handbook of nonlinear filtering}, 2009.

\bibitem{dubarry:lecorff:2013}
C.~Dubarry and S.~Le~Corff.
\newblock Non-asymptotic deviation inequalities for smoothed additive
  functionals in nonlinear state-space models.
\newblock {\em Bernoulli}, 19(5B):2222--2249, 2013.

\bibitem{fearnhead:wyncoll:tawn:2010}
P.~Fearnhead, D.~Wyncoll, and J.~Tawn.
\newblock A sequential smoothing algorithm with linear computational cost.
\newblock {\em {B}iometrika}, 97(2):447--464, 2010.

\bibitem{godsill:doucet:west:2004}
S.~J. Godsill, A.~Doucet, and M.~West.
\newblock {M}onte {C}arlo smoothing for non-linear time series.
\newblock {\em J. Am. Statist. Assoc.}, 50:438--449, 2004.

\bibitem{gordon:salmond:smith:1993}
N.~Gordon, D.~Salmond, and A.~F. Smith.
\newblock Novel approach to nonlinear/non-{G}aussian {B}ayesian state
  estimation.
\newblock {\em IEE Proc. F, Radar Signal Process.}, 140:107--113, 1993.

\bibitem{hull:white:1987}
J.~Hull and A.~White.
\newblock The pricing of options on assets with stochastic volatilities.
\newblock {\em J. Finance}, 42:281--300, 1987.

\bibitem{huerzeler:kuensch:1998}
M.~H{\"u}rzeler and H.~R. K{\"u}nsch.
\newblock {M}onte {C}arlo approximations for general state-space models.
\newblock {\em J. Comput. Graph. Statist.}, 7:175--193, 1998.

\bibitem{jacob:murray:rubenthaler:2013}
P.~E. Jacob, L.~M. Murray, and S.~Rubenthaler.
\newblock Path storage in the particle filter.
\newblock {\em Statistics and Computing}, pages 1--10, 2013.

\bibitem{kitagawa:1996}
G.~Kitagawa.
\newblock {M}onte-{C}arlo filter and smoother for non-{G}aussian nonlinear
  state space models.
\newblock {\em J. Comput. Graph. Statist.}, 1:1--25, 1996.

\bibitem{kitagawa:sato:2001}
G.~Kitagawa and S.~Sato.
\newblock Monte {C}arlo smoothing and self-organising state-space model.
\newblock In {\em Sequential {M}onte {C}arlo methods in practice}, Stat. Eng.
  Inf. Sci., pages 177--195. Springer, New York, 2001.

\bibitem{koski:2001}
T.~Koski.
\newblock {\em Hidden {M}arkov Models for Bioinformatics}.
\newblock Kluwer, 2001.

\bibitem{lindsten:schon:2013}
F.~Lindsten and T.~B. Sch\"{o}n.
\newblock Backward simulation methods for {M}onte {C}arlo statistical
  inference.
\newblock {\em Foundations and Trends in Machine Learning}, 6(1):1--143, 2013.

\bibitem{mongillo:deneve:2008}
G.~Mongillo and S.~Den\`{e}ve.
\newblock Online learning with hidden {M}arkov models.
\newblock {\em Neural Computation}, 20(7):1706--1716, 2008.

\bibitem{olsson:cappe:douc:moulines:2006}
J.~Olsson, O.~Capp\'e, R.~Douc, and E.~Moulines.
\newblock Sequential {M}onte {C}arlo smoothing with application to parameter
  estimation in non-linear state space models.
\newblock {\em Bernoulli}, 14(1):155--179, 2008.
\newblock arXiv:math.ST/0609514.

\bibitem{olsson:strojby:2010}
J.~Olsson and J.~Str\"{o}jby.
\newblock Convergence of random weight particle filters.
\newblock Technical report, Lund University, 2010.

\bibitem{olsson:strojby:2011}
J.~Olsson and J.~Str\"{o}jby.
\newblock Particle-based likelihood inference in partially observed diffusion
  processes using generalised {P}oisson estimators.
\newblock {\em Electron. J. Statist.}, 5:1090--1122, 2011.

\bibitem{olsson:westerborn:2014}
J.~Olsson and J.~Westerborn.
\newblock Efficient particle-based online smoothing in general hidden {M}arkov
  models.
\newblock In {\em IEEE 2014 International Conference on Acoustics, Speech, and
  Signal Processing (ICASSP 2014)}, 2014.

\bibitem{pitt:shephard:1999}
M.~K. Pitt and N.~Shephard.
\newblock Filtering via simulation: Auxiliary particle filters.
\newblock {\em J. Am. Statist. Assoc.}, 94(446):590--599, 1999.

\bibitem{rabiner:juang:1993}
L.~R. Rabiner and B-H. Juang.
\newblock {\em Fundamentals of Speech Recognition}.
\newblock Prentice-Hall, 1993.

\bibitem{rauch:tung:striebel:1965}
H.~Rauch, F.~Tung, and C.~Striebel.
\newblock Maximum likelihood estimates of linear dynamic systems.
\newblock {\em AIAA Journal}, 3(8):1445--1450, 1965.

\bibitem{verge:delmoral:moulines:olsson:2014}
C.~Verg\'{e}, P.~Del~Moral, E.~Moulines, and J.~Olsson.
\newblock {C}onvergence properties of weighted particle islands with
  application to the double bootstrap algorithm.
\newblock Preprint, 2014.

\end{thebibliography}
